\documentclass[acmsmall,screen]{acmart} 
\acmJournal{PACMPL}

\setcopyright{none}

\usepackage{thm-restate}
\usepackage{listings}

\usepackage{mathpartir}
\usepackage{stmaryrd}
\usepackage[]{todonotes}

\usepackage{xspace}

\theoremstyle{acmplain}
\newtheorem{postulate}{Postulate}[section]

\usepackage{booktabs}   
\usepackage{subcaption} 

\newcommand{\Ga}{\Gamma}

\newcommand{\mypara}[1]{ {\noindent {\bf #1.}} }

\newcommand{\smallkind}{\mathcal{K}}

\newcommand{\btnote}[1]{{\color{red} #1}}

\newcommand\cproprestrict[1]{\overline{#1}}

\newcommand\stateconf[2]{\langle #1;  #2 \rangle}
\newcommand\stateconfH[1]{\stateconf{H}{#1}}
\newcommand{\red}{\longrightarrow}

\newcommand\sep{\mathrel{|}}

\newcommand\bnfor{\mathrel|}

\newcommand{\llb}{\llbracket}
\newcommand{\rrb}{\rrbracket}

\newcommand\kwd[1]{\mathbf{#1}}

\newcommand{\m}[1]{\mathsf{#1}}

\newcommand\typevarset{\ensuremath{\mathcal{X}}}
\newcommand\varset{\ensuremath{\mathcal{V}}}
\newcommand\nameset{\ensuremath{\mathcal N}}

\newcommand\kind[1]{\ensuremath{\mathsf{#1}}}
\newcommand\kref[2]{\{#1 \sep #2\}}

\newcommand\kname{\kind{Nm}}
\newcommand\krecord{\kind{Rec}}
\newcommand\kcollection{\kind{Col}}
\newcommand\kreference{\kind{Ref}}
\newcommand\kfunction{\kind{Fun}}
\let\kfun=\kfunction
\newcommand\ktype{\kind{Type}}
\newcommand\kpolyfun{\kind{Gen}}

\newcommand\kforall[1]{\Pi #1.\,}
\newcommand\tabs[1]{\mathsf{\lambda} #1.\,}

\newcommand{\nonemp}[1]{\neg\m{empty}(#1)}

\newcommand\basetype[1]{\ensuremath{\mathsf{#1}}}
\newcommand\casearrow{\Rightarrow}
\newcommand\caseor{\;|\;}
\newcommand{\colcase}[3]{\kwd{case}\,#1\,\kwd{of}\,(\emptycol{}
  \Rightarrow #2\caseor #3)}

\newcommand\treccons[3]{\trecord{#1 : #2}@ #3}
\newcommand\trecnil{\trecord{}}
\newcommand\trecord[1]{\boldsymbol{\langle}#1\boldsymbol{\rangle}}

\newcommand\tbool{\basetype{Bool}}

\newcommand\tfun{\mathrel{\rightarrow}}

\newcommand\tcol[1]{{#1}^{\star}}

\newcommand\tref[1]{\kwd{ref}\, #1}

\newcommand\tfix[4]{\boldsymbol{\mu} #1 : (\kforall{ #4 {:} #2 } #3). \tabs{#4 {::} #2}}

\newcommand\tkindof[5]{\kwd{if}\, #1 :: #3\, \kwd{as}\, #2 \casearrow #4 \, \kwd{else}\, #5 }
\newcommand\tkindofP[4]{\tkindof{#1}{t}{#2}{#3}{#4}}

\newcommand{\lb}[1]{\kwd{lab}(#1)}

\newcommand{\one}{\mathbf{1}}
\newcommand{\unitelem}{\diamond}

\newcommand{\elim}[2]{\mathit{elim}_{#1}(#2)}

\newcommand{\hdlabel}[1]{\kwd{headLabel}(#1)}
\newcommand{\hdtype}[1]{\kwd{headType}(#1)}
\newcommand{\tl}[1]{\kwd{tail}(#1)}
\newcommand{\refof}[1]{\kwd{refOf}(#1)}
\newcommand{\inst}[1]{\kwd{tmap}(#1)}
\newcommand{\colof}[1]{\kwd{colOf}(#1)}
\newcommand{\codom}[1]{\kwd{img}(#1)}
\newcommand{\dom}[1]{\kwd{dom}(#1)}
\newcommand\typeabs[2]{\Lambda #1. #2 }
\newcommand{\typeapp}[2]{#1[#2]}

\newcommand\record[1]{\langle #1\rangle}

\newcommand\reccons[3]{\record{#1 = #2}@ #3 }

\newcommand{\vrec}[3]{\boldsymbol{\mu} #1 {:} #2 . #3}

\newcommand\mkref[1]{\kwd{ref}\, #1}
\newcommand\deref[1]{{!#1}}

\newcommand\ite[3]{\kwd{if}\,  #1\, \kwd{then}\, #2\, \kwd{else}\, #3}

\newcommand\const[1]{\mathit{#1}}
\newcommand\true{\const{true}}
\newcommand\false{\const{false}}

\newcommand\emptycol[1]{\varepsilon}

\newcommand\abs[1]{\lambda #1.\,}

\newcommand{\tmhdlabel}[1]{\kwd{recHeadLabel}(#1)}
\newcommand{\tmhdterm}[1]{\kwd{recHeadTerm}(#1)}
\newcommand{\tmtl}[1]{\kwd{recTail}(#1)}

\newcommand{\tmcolhd}[1]{\kwd{colHead}(#1)}
\newcommand{\tmcoltl}[1]{\kwd{colTail}(#1)}

\newcommand\subst[2]{\{{#1}/{#2}\}}

\begin{document}

\title[]{Refinement Kinds}         
\subtitle{Type-safe Programming with Practical Type-level Computation}



\author{Lu\'{i}s Caires}
\orcid{0000-0002-3215-6734}
\affiliation{
  \department{Departamento de Inform\'{a}tica}
  \institution{NOVA-LINCS, FCT-NOVA, Universidade Nova de Lisboa}
  \country{Portugal}}
\email{lcaires@fct.unl.pt}

\author{Bernardo Toninho}
\orcid{0000-0002-0746-7514}
\affiliation{
  \department{Departamento de Inform\'{a}tica}
  \institution{NOVA-LINCS, FCT-NOVA, Universidade Nova de Lisboa}
  \country{Portugal}}
\email{btoninho@fct.unl.pt}


\begin{abstract}
  This work introduces the novel concept of {\em kind refinement},
  which we develop in the context of an explicitly
  polymorphic ML-like language with type-level computation.
  Just as type
  refinements embed rich specifications by means of comprehension
  principles expressed by predicates over values in the type domain,
  kind refinements provide rich {\em kind} specifications by means of
  predicates over {\em types} in the kind domain.
  By leveraging our powerful refinement kind discipline, types in our
  language are not just used to statically classify program
  expressions and values, but also conveniently manipulated as
  tree-like data structures, with their kinds refined by logical
  constraints on such structures.  Remarkably, the resulting typing
  and kinding disciplines allow for powerful forms of type reflection,
  ad-hoc polymorphism and type meta-programming, which are
  often found in modern software development, but not typically expressible in a
  type-safe manner in general purpose languages.
  We validate our approach both formally and pragmatically
  by establishing the standard meta-theoretical results of type safety and
  via a prototype implementation of a kind checker, type checker and
  interpreter for our language.

\end{abstract}



\maketitle

\section{Introduction}\label{sect:intro}

Current software development practices increasingly rely on many forms of
automation, often based on tools that generate code from various types
of specifications, leveraging the various reflection and
meta-programming facilities that modern programming languages provide. 
A simple example would be a function that given 
any record type would produce a factory of mutable instances of the given record type.
As a more involved and useful example
consider a code generator that given as input an XML database
schema, produces all the code needed to create and manipulate a database instance of such schema
with some appropriate database connector.   

Automated
code generation, domain specific languages, and meta-programming are
increasingly becoming productivity drivers for the software industry,
while also making programming more accessible to non-experts,
and, more generally, increasing the level of abstraction expressible in languages
and tools for program construction.
Meta-programming is better
supported by so-called dynamic languages and related frameworks, such
as Ruby and Ruby on Rails, JavaScript and Node.js, but is also
present in static languages such as Java, Scala, Go and F\#, that
provide support for reflection and other facilities,
allowing both code and types to be manipulated as data by
programs.

Unfortunately, meta-programming constructs and idioms aggressively
challenge the safety guarantees of static typing, which becomes
especially problematic given that meta-programs are notoriously hard to
test for correctness.
This challenge is then the key motivation for our paper, which
introduces for the first time the concept of {\em refinement kinds}
and illustrates how the associated discipline cleanly supports static
type checking of type-level reflection, parametric and ad-hoc
polymorphism, which can all be combined to implement interesting
meta-programming idioms.

Refinement kinds are a natural transposition of the well-known concept
of refinement types (of values)
\cite{DBLP:journals/toplas/BengtsonBFGM11,DBLP:conf/pldi/RondonKJ08,
  DBLP:conf/esop/VazouRJ13} to the realm of kinds (of types).  Several
systems of refinement types have been proposed in the literature,
generally motivated as a pragmatic compromise between usability and
the expressiveness of full-fledged dependent types, which require
proof objects to be explicitly constructed by programmers.
Our work aims to show that the arguably natural notion of introducing
refinements in the kind structure allows us to cleanly support
sophisticated statically typed meta-programming concepts, which we
illustrate in the context of a higher-order polymorphic
$\lambda$-calculus with imperative constructs, chosen as a convenient
representative for languages with higher-order store. Moreover, by
leveraging the stratification between types and kinds, our design
shows that arguably advanced type-level features can be integrated
into a general purpose language without the need to fundamentally
alter the language's type system and its associated rules.


Just as refinement types support expressive type specifications by
comprehension principles expressed by {\em predicates over values} in
the type domains (typically implemented by SMT decidable Floyd-Hoare
assertions \cite{DBLP:journals/tse/RushbyOS98}), refinement kinds
support rich and flexible kind specifications by means of
comprehension principles expressed by {\em predicates over types} in
the kind domains.  They also naturally give rise to a notion of
subkinding by entailment in the refinement logic.  For example, we
introduce a least upper bound kind for each kind, from
which more concrete kinds and types may be defined by refinement,
adding an unusual degree of plasticity to subkinding.

Crucially, types in our language may be reflectively manipulated as
first-class (abstract-syntax) labelled trees (cf. XML data), both
statically and at runtime. Moreover, the deduction of relevant
structural properties of such tree representations of types is
amenable to rather efficient implementation, unlike properties on the
typical value domains (e.g., integers, arrays) manipulated by
mainstream languages, and easier to automate using off-the-shelf SMT
solvers (e.g. \cite{DBLP:conf/tacas/MouraB08,DBLP:conf/cav/BarrettCDHJKRT11}).
Remarkably, even if types in our system can essentially be manipulated
by type-level functions and operators as abstract-syntax trees, our
system statically ensures the sound inhabitation of the outcomes of
type-level computations by the associated program-level terms,
enforcing type safety. This allows our language to express challenging
reflection idioms in a type-safe way, that we have no clear
perspective on how to cleanly and effectively embed in extant
type theories in a fully automated way.

To make the design of our framework more concrete, we briefly detail
our treatment of record types.  Usually, a record type is represented
by a tuple of label-and-type pairs, subject to the constraint that all
the labels must be pairwise distinct
(e.g. see~\cite{DBLP:conf/popl/HarperP91}). In order to support more
effective manipulation of record types by type-level functions, record
types in our theory are represented by values of a list-like data
structure: the record type constructors are the type of empty records
$\trecnil$ and the ``cons'' cell $\treccons{L}{T}{R}$, which
constructs the record type obtained by adding a field declaration
$\trecord{\mathit{L}:T}$ to the record type $R$.

The record type destructors are functions $\hdlabel{R}$, $\hdtype{R}$
and $\tl{R}$, which apply to any non-empty record type $R$. As will be
shown later, the more usual record field projection operator $r.L$
and record type field projection operator $T.L$ are
definable in our language using suitable meta-programs. In our system,
record labels (cf. names) are type and term-level first-class values
of kind $\kname$.  Record types also have their own kind, dubbed
$\krecord$. As we will see, our theory provides a range of {\em basic} kinds
that specialize the kind of all types $\ktype$ via subkinding,
which can be further specialized via kind refinement.

For example, we may define the record type
$\basetype{Person}\triangleq
\treccons{name}{\basetype{String}}{\treccons{age}{\basetype{Int}}{\trecnil}}$,
which we conveniently abbreviate by
$\trecord{\mathit{name}:\basetype{String};\mathit{age}:\basetype{Int}}$.
We
then have that $\hdlabel{\basetype{Person}}=\mathit{name}$,
$\hdtype{\basetype{Person}}=\basetype{String}$ and
$\tl{\basetype{Person}}=\treccons{\mathit{age}}{\basetype{Int}}{\trecnil}$.
The kinding
of the $\treccons{L}{T}{R}$ type constructor may be clarified in the
following type-level function $\m{addFieldType}$:
\[
\begin{array}{ll}
    \m{addFieldType} :: \kforall {l{::}\kname}\kforall {t{::}\ktype} \kforall{r{::}\kref{s{::}\krecord}{l\not\in \lb{s}}}
  \krecord\\
  \m{addFieldType} \triangleq \lambda l{::}\kname.\lambda t{::}\ktype.\lambda
  r{::}\kref{s{::}\krecord}{l\not\in \lb{s}}.\treccons{l}{t}{r}\\
\end{array}
\]
The $\m{addFieldType}$ {\em type-level} function takes a label $l$, a
type $t$ and any record type $r$ that does not contain label $l$, and
returns the expected extended record type of kind $\krecord$. Notice
that the {\em kind} of all record types that do not contain label $l$ is
represented by the refinement kind
$\kref{s{::}\krecord}{l\not\in \lb{s}}$.

A refinement kind in our system is noted
$\kref{t{::}\smallkind}{\varphi(t)}$, where $\smallkind$ is a bas
kind, and the logical formula $\varphi(t)$ expresses a constraint on
the type $t$ that inhabits $\smallkind$.  As in 
refinement type systems
\cite{DBLP:journals/toplas/BengtsonBFGM11,DBLP:conf/icfp/SwamyCFSBY11,DBLP:conf/icfp/VazouSJVJ14},
our underlying logic of
refinements includes a (decidable) theory for the various finite
tree-like data types used to schematically represent type
specifications, as is the case of our record-types-as-lists,
function-types-as-pairs (i.e.~a pair of a domain and an image type),
and so on. The kind refinement rule is thus expressed by
  \[
  \inferrule*[right=(kref)]
    {\Ga \models \varphi\{T/t\} \quad \Ga \vdash T :: \smallkind }
    {\Ga \vdash T :: \kref{t{::}\smallkind}{\varphi}}
  \]
where $\Gamma\models\varphi$ denotes entailment in the refinement
logic. Basic formulas of our refinement logic include propositional
logic, equality, and some useful predicates and functions on types,
including the primitive type constructors and destructors, such as
$\lb{R}$ (record label set), $L\in S$ (label set membership), $S\#S'$
(label set apartness), $R@S$ (concatenation), $\dom{F}$ (function
domain selector). Interestingly,
given the presence of equality in refinements, it is
always possible to define for any type $T$ of kind $\smallkind$ a
precise singleton kind of the form
$\kref{t::\smallkind}{t= T}$. As another simple
example, consider the kind $\m{Auto}$ of automorphisms, defined as
$\kref{t::\kfun}{\dom{t} = \codom{t}}$.

A use of the type-level function $\m{addFieldType}$ given above is,
for instance, the definition of the following {\em term-level}
polymorphic record extension function
\[
\begin{array}{ll}
  \m{addField} : \forall {l{::}\kname}.\forall {t{::}\ktype}.
  \forall{r{::}\kref{s{::}\krecord}{l\not\in \lb{s}}}.t \tfun  r\tfun
  \m{addFieldType} \;l\; t\; r \\
   \m{addField} \triangleq \Lambda l{::}\kname.\Lambda t{::}\ktype.\Lambda
  r{::}\kref{s{::}\krecord}{l\not\in \lb{s}}. \lambda x{:}t.\lambda y{:}r.\reccons{l}{x}{y}\\
\end{array}
\]
The $\m{addField}$ function takes a label $l$, a type $t$,
a record type $r$ that does not contain label $l$, and values
of types $t$ and $r$, respectively, returning a record of type
$\m{addFieldType} \;l\; t\; r$.

The type-level and term-level functions $\m{addFieldType}$ and
$\m{addField}$ respectively illustrate some of the key insights of our
type theory, namely the use of types and their refined kinds as
specifications that can be manipulated as tree-like structures by
programs in a fully type-safe way. For instance, the following judgment,
expressing the correspondence between the term-level computation
$\m{addField}\;l\;t\;\;r\;x\;y$ and the type-level computation
$\m{addFieldType} \;l\; t\; r $,
is derivable:
\[
  l{:}\kname, t{:}\ktype,
r{:}\kref{s{::}\krecord}{l\not\in \lb{s}}, x{:}t,y{:}r
\vdash \m{addField}\;l\;t\;r\;x\;y : \m{addFieldType} \;l\; t\; r 
\]
An instance of this judgement yields:
\[
\vdash
\m{addField}\;\mathit{name}\;\basetype{String}\;\record{age :
  \basetype{Int}}\;\mbox{``jack''}\;\record{age = 20} : \m{addFieldType}
\;\mathit{name}\; \basetype{String}\; \record{age : \basetype{Int}}
\]
Noting that
$\record{\mathit{age} : \basetype{Int}}::
\kref{s{::}\krecord}{\mathit{name} \not\in \lb{s}} $ is derivable since
$\mathit{name} \not\in \lb{\record{\mathit{age} :
    \basetype{Int}}}$ is provable in the refinement logic, we
have the following term and type-level evaluations:
$$
\begin{array}{ll}
  (\m{addField}\;\mathit{name}\;\basetype{String}\;\record{age : \basetype{Int}}\;\mbox{``jack''}\;\record{age = 20}) \to^* \record{\mathit{name}=\mbox{``jack''};\mathit{age} = 20}\\
  (\m{addFieldType} \;\mathit{name}\; \basetype{String}\; \record{age : \basetype{Int}})
  \equiv \record{\mathit{name}:\basetype{String};\mathit{age}: \basetype{Int}}
\end{array}
$$
Using the available refinement principles, our system can also derive the
following more precise kinding for the type $\m{addFieldType} \;l\; t\; r $:
\[
 l{:}\kname, t{:}\ktype,
  r{:}\kref{s{::}\krecord}{l\not\in \lb{s}} \vdash \m{addFieldType} \;l\; t\; r :: 
  \kref{ s{::}\krecord}{s = \treccons{l}{t}{r} }
\]

\mypara{Contributions}
We summarise the main contributions of this work:
\begin{itemize}
  \item We illustrate the concept of refinement kinds,
showing how it supports the flexible and clean definition of
statically typed meta-programs through several examples
(Section~\ref{sec:introexamples}).
\item We technically develop our refinement kind system
(Section~\ref{sec:system}), targeting a
polymorphic $\lambda$-calculus (Section~\ref{sec:pl}) with records,
references, collections and supporting type-level computation over
types of all kinds, thus capturing the essence of an ML-like language.
\item We establish the key meta-theoretical result (Section~\ref{sec:metatheory})
of type safety through 
type unicity, type preservation and progress
(Theorems~\ref{thm:unicity},~\ref{thm:preserv} and~\ref{thm:progress}, respectively).

\item We report on our implementation of a prototype kind and
type-checker for our theory (Section~\ref{sec:eqreasoning}), which
validates the examples of our paper 
and the overall feasibility of our approach.
%
%

\item We give a detailed overview of related work
(Section~\ref{sec:rw}), and offer some concluding remarks and
discussion of future work
(Section~\ref{sec:conc}).

\end{itemize}

Appendices~\ref{app:rules},~\ref{app:opsemrules}
  and~\ref{app:proofs} list omitted definitions of the type theory,
  its semantics and proof outlines, respectively.


\section{Programming with Refinement Kinds}
\label{sec:introexamples}

Before delving into the technical intricacies of our theory in
Section~\ref{sec:system} and beyond, we illustrate the various
features and expressiveness of our theory through a series of examples
that showcase how our language supports
challenging (from a static typing perspective)
meta-programming idioms.

\mypara{Generating Mutable Records}
We begin with a simple higher-order meta-program that computes a
``generator'' for mutable records from a specification of its
representation type, expressed as an arbitrary record type.  Consider
the following definition of the (recursive) function $\m{genConstr}$:
$$
\begin{array}{ll}
  \m{genConstr} \triangleq 
  \Lambda S{::}\kref{r{::}\krecord}{\neg\m{empty}(r)}.
 \Lambda V{::}\kref{v{::}\krecord}{\lb{v} \# \lb{S}}.\lambda v{:}V.\\
 \hspace{0.5cm}  \lambda x{:}\hdtype{S}.
   \ite{\nonemp{\tl{S}}}{\\
  \hspace{1.5cm}\m{genConstr}\;\tl{S}\;\treccons{\hdlabel{S}}{\tref{\hdtype{S}}}{V}\;
   \reccons{\hdlabel{S}}{\mkref{x}}{v}\\\hspace{2.9cm}}{

 \reccons{\hdlabel{S}}{\mkref{x}}{v}}\\
 
\end{array}
$$
Given a non-empty record type $S$, function
$\m{genConstr}$ returns a constructor {\em function} for a mutable
record whose fields are specified by
$S$.  We use a pragmatic notation to express recursive definitions
(coinciding with that of our implementation),
which in our formal core language is represented by an explicit
structural recursion construct.  Parameters $V$ and
$v$ are accumulating parameters that track intermediate types, values
and a disjointness invariant on those types during computation (for
simplicity, we generate the record fields in reverse order). 

Intuitively, and recovering the record type $\basetype{Person}$
from above, 
$\m{genConstr}\;\basetype{Person}\;\trecord{}\;\record{}$ evaluates to
a value equivalent to $\lambda x{:}\basetype{String}.
\lambda y{:}\basetype{Int}.
\record{\mathit{age} = \mkref{y};\mathit{name} = \mkref{x}}$.

Notice that function
$\m{genConstr}$ accepts any non-empty record type
$S$, and proceeds by recursion on the structure of 
$S$, as a list of label-type pairs. The parameter
$S$ holds the types of the fields still pending for addition to the
final record type, parameter
$V$ holds the types of the fields already added to the final record
type, and
$v$ holds the already built mutable record value.  To properly call
$\m{genConstr}$, we ``initialize'' $V$ with
$\trecord{}$ (i.e.~the empty record {\em type}), and
$v$ to $\record{}$.
Moreover, the refined kind of
$V$ specifies the label apartness constraint needed to type check the
recursive call of
$\m{genConstr}$, in particular, given
$\lb{V}\#\lb{S}$, we can automatically deduce that $\hdlabel{S}\not\in
\lb{V}$, needed to kind check
$\treccons{\hdlabel{S}}{\tref{\hdtype{S}}}{V}$; and
$\lb{\treccons{\hdlabel{S}}{\tref{\hdtype{S}}}{V}}\#\lb{\tl{S}}$,
required to kind and type check the recursive call. In our language,
$\m{genConstr}$ can be typed as follows:
$$
\m{genConstr}:\forall S{::}\kref{r{::}\krecord}{\nonemp{r}}.
\forall V{::}\kref{v{::}\krecord}{\lb{v}\#\lb{S}}.
V \rightarrow (\m{GType}\,S\,V) 
$$
where $\m{GType}$ is the (recursive) type-level function such that
$$
\begin{array}{ll}
\m{GType} :: \kforall {S{::}\kref{r{::}\krecord}{\nonemp{r}}}\kforall {V{::}\kref{v{::}\krecord}{\lb{v}\#\lb{S}}}
\kfun\\
  \m{GType}\triangleq \\
\  \hspace{0.5cm}\lambda {S{::}\kref{r{::}\krecord}{\nonemp{r}}}.\\
\  \hspace{0.5cm}\lambda {V{::}\kref{v{::}\krecord}{\lb{v}\#\lb{S}}}. \\
\  \hspace{0.5cm}\hdtype{S}\tfun\ite{\nonemp{\tl{S}}}{ \\
  \hspace{3.5cm}\m{GType}\;\tl{S}\;\treccons{\hdlabel{S}}{\tref{\hdtype{S}}}{V}}
  {\\
  \hspace{4cm} \treccons{\hdlabel{S}}{\tref{\hdtype{S}}}{V}}\\
\\
\end{array}
$$
We can see that, in general, the type-level application
$\m{GType}\;\trecord{L_1:T_1; ...; L_n:T_n}\;\trecord{}$ computes 
the type
$T_1 \tfun ... \tfun T_n \tfun \trecord{L_n:\tref{T_n}; ...;
  L_1:\tref{T_1}}$.  In particular, we have
\[
\m{genConstr}\;\basetype{Person}\;\trecord{}\;\record{}
:\basetype{String}\tfun \basetype{Int}\tfun \trecord{\mathit{age} = \tref{\basetype{Int}};\mathit{name} = \tref{\basetype{String}}} 
\]

\mypara{From Record Types to XML Tables}
As a second example, we develop a generic function
$\m{MkTable}$ that generates and formats an XML table for any record
type, inspired by the example in Section 2.2 of
\cite{DBLP:conf/pldi/Chlipala10}, but where refinement kinds allow for
extreme simplicity. We start by
introducing an auxiliary type-level
$\m{Map}$ function, that computes the record type obtained from a
record type $R$ by applying a type transformation
$G$ (of higher-order kind) to the type of each field of $R$.
\[
  \begin{array}{ll}
    \m{Map} :: \kforall {G{::}(\kforall {X::\ktype}\ktype)}
\kforall {R{::}\krecord{}}
\kref{r::\krecord{}}{\lb{r} = \lb{R}}\\
\m{Map}\triangleq \lambda {G{::}(\kforall {X::\ktype}\ktype)}.\lambda {R{::}\krecord}.\\
   \hspace{1.0cm}\ite{\nonemp{R}}{\treccons{\hdlabel{R}}{G\;\hdtype{R}}{(\m{Map}\;G\;\tl{R}})}{\trecord{}}\\

\end{array}
\]

The logical constraint $\lb{r} = \lb{R}$ expresses that the result of
$\m{Map}\; G \; R$ has exactly the same labels as record type
$R$. This implies that $\hdlabel{R}\not\in\lb{\m{Map}\;G\;\tl{R}}$ in
the recursive call, thus allowing the ``cons'' to be well-kinded.  We
now define:
\[
\begin{array}{ll}
\m{XForm}&:: \kforall{t::\ktype} \ktype{}\\
  \m{XForm}&\triangleq \lambda {t{::}\ktype}.\trecord{\mathit{tag}:\basetype{String};
\mathit{toStr}:t\tfun\basetype{String}}\\

  \\
\m{MkTableType}&::\Pi {r{::}\krecord{}}.\kref{r::\krecord{}}{\lb{r} = \lb{R}}\\  
\m{MkTableType}&\triangleq \lambda {r{::}\krecord{}}.\m{Map}\;\m{XForm}\;r\\

  \\
  \m{MkTable}&:\forall {R{::}\krecord{}}.(\m{MkTableType}\;R) \tfun R \tfun \basetype{String}\\
  \m{MkTable}&\triangleq \Lambda {R{::}\krecord}.
               \lambda M{:}\m{MkTableType}\;R.\lambda r{:}R.\\
& \hspace{0.3cm} \ite{\nonemp{R}}{\\
& \hspace{0.6cm}\mathtt{``{<}tr{>}{<}th{>}"}+M.\tmhdlabel{M}.\mathit{tag}+\mathtt{``{<}/th{>}{<}td{>}"}+\\
& \hspace{0.6cm}M.\tmhdlabel{M}.\mathit{toStr}\;r.\tmhdlabel{M} + \mathtt{``{<}/td{>}{<}/tr{>}"}+\\
& \hspace{0.6cm}\m{MkTable}\;\tl{R}\;\tmtl{M}\;\tmtl{r}
\hspace{0.3cm} }
{\mathtt{``"}}

\end{array}
\]
It is instructive to discuss why and how this code is well-typed,
witnessing the expressiveness of refinement kinds, despite their
conceptual simplicity (which can be judged by the arguably
parsimonious nature of the definitions above).  Let us first consider
the expression $M.\tmhdlabel{M}.\mathit{tag}$. Notice that, by
declaration, $R{::}\krecord$ and $r{:}R$.  However, the expression
under consideration is to be typed under the assumption that
$\nonemp{R}$, which is added to the current set of refinement
assumptions while typing the $\kwd{then}$ branch. Using
$TT$ for the type of $M$, since
$\m{MkTableType}\;R ::\kref{r{::}\krecord{}}{\lb{r} = \lb{R}}$, by
refinement we have that $\lb{TT}=\lb{R}$ and thus
$\nonemp{TT}$, allowing $\tmhdlabel{M}$ to be defined.
Since $M:\m{MkTableType}\;R$ we have
\[
\begin{array}{ll}
(\m{MkTableType}\;R) \equiv
(\m{Map}\;\m{XForm}\;R) \equiv 
\treccons{\hdlabel{R}}{\m{XForm}\;\hdtype{R}}{(\m{Map}\;G\;\tl{R})}
\end{array}
\]
We thus derive $\hdlabel{TT}\equiv\hdlabel{R}$. 
Then 
\[
\begin{array}{ll}
 \!\!\!\hdtype{\m{MkTableType}\;R} \!\equiv\! 
  \m{XForm}\;\hdtype{R} \!\equiv\! \trecord{\mathit{tag}:\basetype{String}; \mathit{toStr}:\hdtype{R}\tfun\basetype{String}}
\end{array}
\]
Hence $M.\hdlabel{M}.\mathit{tag} : \basetype{String}$.
By a similar reasoning, we conclude $r.\tmhdlabel{M} : \hdtype{R}$. In
Section~\ref{sec:system} we show how refinements and equalities
derived therein are integrated into typing and kinding. Moreover, in
Section~\ref{sec:eqreasoning} we detail how refinements can be
represented and discharged via SMT solvers in order to make fully precise
the reasoning sketched above.

\mypara{Generating Getters and Setters}
As a final introductory example, we develop a generic function
$\m{MkMut}$ that generates a getter/setter wrapper for any mutable
record (i.e. a record where all its fields are of reference type).
We first define the auxiliary type-level $\m{MutableRec}$ function,
that returns the mutable record type obtained from a record type $R$
in terms of $\m{Map}$:
\[
\begin{array}{ll}
\m{MutableRec}&:: \kforall {R::\krecord}\kref{r::\krecord}{\lb{r} = \lb{R}}\\

  \m{MutableRec}&\triangleq \m{Map}\;(\lambda {r{::}\ktype}.\tref{r})\\
\end{array}
\]

We then define the auxiliary type-level $\m{SetGet}$ function, that
returns the record type that exposes the getter/setter interface
generated from record type $R$:
$$
\begin{array}{ll}
\m{SetGetRec}&:: \kforall {R::\krecord}\kref{r::\krecord}{\lb{r} = set{++}\lb{R}\cup get{++}\lb{R}}\\

  \m{SetGetRec}&\triangleq \lambda {R{::}\krecord{}}.\\
& \hspace{0.3cm} \ite{\nonemp{R}}{\\ 
& \hspace{0.6cm}\treccons{get{++}\hdlabel{R}}{\one\tfun\hdtype{R}}{}\\
& \hspace{0.6cm}\treccons{set{++}\hdlabel{R}}{\hdtype{R}\tfun \one}{}\\
& \hspace{0.6cm}\m{SetGetRec}\;\tl{R}\\
& \hspace{0.2cm}}{\trecord{}}\\
\end{array}
$$
Here, $n{++}m$ denotes the name obtained by appending $n$ to $m$, and
$n{++}S$ denotes the {\em label set} obtained from $S$ by prefixing
every label in S with name $n$.  The function $\m{SetGet}$ is well
kinded since the refinement kind constraints imply that the resulting
getter/setter interface type is well formed (i.e. all labels
distinct).  We can finally depict the type and code of the $\m{MkMut}$
function:
\[
\begin{array}{ll}
\m{MkMut}&:: \forall {R::\krecord}.\m{MutableRec}\;R \tfun \m{SetGetRec}\;R\\
  \m{MkMut}&\triangleq \Lambda {R{::}\krecord{}}.\\
& \hspace{0.3cm} \lambda {r{:}\m{MutableRec}\;R}.\\
& \hspace{0.3cm} \ite{\nonemp{R}}{\\ 
& \hspace{0.6cm}\reccons{get{++}\hdlabel{R}}{\lambda x{:}\one. !(r.\tmhdlabel{R})}{}\\
& \hspace{0.6cm}\reccons{set{++}\hdlabel{R}}{\lambda x{:}\hdtype{R}.r.\tmhdlabel{R} := x}{}\\
& \hspace{0.6cm}\m{MkMut}\;\tl{R}\;\tmtl{r}\\
& \hspace{0.2cm}}{\record{}}\\

\end{array}
\]
For example, assuming $r:\m{MutableRec}\;\basetype{Person}$ we have that 
$\m{MkMut}\;\basetype{Person}\;r$ computes a record equivalent to:
\[
  \begin{array}{ll}
&\record{
getname=\lambda x{:}\one.!(r.name);\\
&setname=\lambda x{:}\basetype{String}.r.name := x ;\\
&getage=\lambda x{:}\one.!(r.age);\\
&setage=\lambda x{:}\basetype{Int}.r.age := x}
\end{array}
\]
where $(\m{MkMut}\;\basetype{Person}\;r): \m{SetGetRec}\;\basetype{Person}$.





\section{A Type Theory with Kind Refinements}
\label{sec:system}

Having given an informal overview of the various features and
expressiveness of our theory, we now formally develop  our theory of
refinement kinds, targeting 
an ML-like functional language with a higher-order store and the
appropriate reference types, collections (i.e. lists) and records.
The typing and kinding systems rely on type-level functions (from
types to types) and a novel form of \emph{subkinding} and \emph{kind
  refinements}.  We first address our particular form of (sub)kinding,
types and the type-level operations enabled by this fine-grained view of
kinds, addressing kind refinements and their interaction with types and
type-level functions in Section~\ref{sec:krefs}.

Given that kinds are classifiers for types, we introduce
a separate kind for each of the key type constructs
of the language. Thus, we have a kind for records, $\krecord$, which
classifies record types; a kind $\kcollection$, for collection types;
a kind $\kfunction$, for function types; a kind $\kreference$, for
reference types; a kind $\kpolyfun_K$ for polymorphic function types
(whose type parameter is of kind $K$); and, a kind $\kname$ for
labels in record types (and records). All of these are specialisations
(i.e. subkinds) of the kind $\ktype$. We write
$\smallkind$ for any such kind.  
The language of {\em types} (a type-level $\lambda$-calculus) provides
constructors for the types described above, but crucially
also introduces type {\em destructors} that allow us to inspect the
structure of types of a given kind and, in combination with type-level
functions and structural type-recursion, enable a form of typed
meta-programming. Indeed, our type language is essentially one of
(inductive) structures and their various constructors and
destructors (and basic data types such as $\tbool$ and $\one$). The syntax of types and kinds is given in
Figure~\ref{fig:ktsyntax}.

\begin{figure}[t]
\[
   \begin{array}{llcll}
      \mbox{Kinds} & K,K' & ::= & \smallkind  \mid
                                  \kref{t{::}\smallkind}{\varphi} \mid \Pi t{:}K.K' & \mbox{Refinement Kinds} \\
                  &\smallkind & ::= & \krecord \mid \kcollection \mid \kfunction \mid
                                      \kreference \mid \kname & \mbox{Basic Kinds}\\
                   & & \mid & \ktype \mid \kpolyfun_K \\[1em]

      \mbox{Types} & T,S,R & ::= & t \mid \lambda t{::}K.T \mid T\,S & \mbox{Type-level Functions}\\
      
                 & & \mid &   \tfix{F}{K}{K'}{t}T    & \mbox{Structural Recursion}\\
                   & & \mid & \forall t{::}K.T & \mbox{Polymorphism}\\
               
                   & & \mid & L \mid \trecnil \mid \treccons{L}{T}{S} & \mbox{Record Type constructors}\\
                   & & \mid & \hdlabel{T} \mid \hdtype{T} & \mbox{Record Type destructors}\\
& & \mid &  \tl{T} \\
                   & & \mid & \tcol{T} \mid \colof{T} & \mbox{Collection Types} \\
                 & & \mid & \tref{T} \mid \refof{T} & \mbox{Reference Types} \\
                 & & \mid & T \rightarrow S \mid  \dom{T} \mid \codom{T} & \mbox{Function Types} \\
                   & & \mid & \tkindofP{T}{\smallkind}{S}{U} & \mbox{Kind Case}\\
                   & & \mid & \ite{\varphi}{T}{S} & \mbox{Property Test}\\ 

                   & & \mid & \tbool \mid \one \mid \dots  & \mbox{Basic Data Types}\\[1em]

\mbox{Extended Types} & \mathcal{T},\mathcal{S} & ::= & T \mid \lb{T}
                                                        \mid \mathcal{T} {++}\,  \mathcal{S}  & 
                                                                          \\[1em]
     
      \mbox{Refinements} & \varphi,\psi & ::= &
                   \varphi \supset \psi \mid
                              \varphi \wedge \psi \mid \dots &
                                                               \mbox{Propositional Logic}\\
          & & \mid &   \m{empty}(\mathcal{T} ) & \mbox{Empty Record Test}\\
                  & & \mid & \mathcal{T} = \mathcal{S} &
                                                         \mbox{Equality}\\
          & & \mid & \mathcal{T} \in \mathcal{S} & \mbox{Label Set
                                                   Membership}\\
        & & \mid & \mathcal{T} \# \,\mathcal{S} & \mbox{Label Set Apartness}
                                                           \end{array}
  \]
\caption{Syntax of Kinds, Types and Refinements\label{fig:ktsyntax}}
\end{figure}

\mypara{Record Types}
Our notion of record type, as introduced in Section~\ref{sec:introexamples},
is essentially a type-level list of pairs of
labels and types which maintains the invariant that all labels in a
record must be distinct. We thus have the type of empty records
$\trecord{}$, and the constructor $\treccons{L}{T}{R}$, which given a
record type $R$ that does not contain the label $L$, generates a
record type that is an extension of $R$ with the label $L$ associated
with type $T$.  Record types are associated with three destructors:
$\hdlabel{T}$, which projects the label of the head of the record $T$
(when seen as a list); $\hdtype{T}$ which produces the type at the
head of the record $T$; and $\tl{T}$ which produces the tail of the
record $T$ (i.e. drops its first label and type pair). As we will see
(Example~\ref{ex:typelabelsel}), since our type-level
$\lambda$-calculus allows for structural recursion, we can {\em
  define} a suitable record projection type construct in terms of
these lower-level primitives.

\mypara{Function Types and Polymorphism}
Functions between terms of type $T$ and $S$ are typed by the usual
$T\rightarrow S$. Given a function type $T$, we can inspect its domain
and image via the destructors $\dom{T}$ and $\codom{T}$, respectively.

Polymorphic function types are represented by $\forall t{::}K.T$ (with
$t$ bound in $T$, as usual). Note that the kind annotation for the
type variable $t$ allows us to express not only general parametric
polymorphic functions (by specifying the kind as $\ktype$) but also
a form of sub-kind polymorphism, since we can restrict the kind
of $t$ to a specific kind such as $\kreference$ or
$\kfunction$, or to a refined kind.
For instance, we
can specify the type
$\forall t{::}\kfunction.t \rightarrow \dom{t} \rightarrow \codom{t}$
of functions that, given a \emph{function type} $t$, a function of
such a type and a value in its domain produce a value in its image
(i.e.  the type of function application).
%

\mypara{{\bf Collections and References}}
The type of collections of elements of type $T$ is written as
$\tcol{T}$, with the associated type destructor $\colof{T}$, which
projects out the type of the collection elements. Similarly, reference
types $\tref{T}$ are bundled with a destructor $\refof{T}$ which
determines the type of the referenced elements.

\mypara{{\bf Kind Test}}
Just as many programming languages have a type case construct \cite{DBLP:journals/toplas/AbadiCPP91} that
allows for the runtime testing of the type of a given expression, our
$\lambda$-calculus of types has a \emph{kind case} construct,
$\tkindofP{T}{\smallkind}{S}{U}$, which checks the kind of type $T$
against kind $\smallkind$, computing to type $S$ if the kinds match
and to $U$ otherwise. Coupled with a term-level analogue, this
enables {\em ad-hoc polymorphism}, allowing us to
express non-parametric polymorphic functions.

\subsection{Type-level Functions and Refinements}
\label{sec:krefs}


The language of types that we have introduced up to this point
essentially consists of tree-like structures with their
various constructors and destructors. As we have mentioned, our type
language is actually a $\lambda$-calculus for the manipulation of such
structures and so includes functions from types to types,
$\lambda t{::}K.T$, and their respective application, written
$T\,S$. We also include a type-level structural recursion operator
$\tfix{F}{K}{K'}{t}T$, which allows us to define recursive type
functions from kind $K$ to $K'$. While written as a fixpoint operator,
we syntactically enforce that recursive calls must always take
structurally smaller arguments to ensure well-foundedness.

Type-level functions are {\em dependently kinded}, with kind
$\Pi t{:}K.K'$ (i.e. the kind of the image type in a type
$\lambda$-abstraction can refer to its {\em type} argument), where the
dependencies manifest themselves in {\em kind refinements}.  Just as
the concept of type refinements allow for rich type specifications
through the integration of predicates over values of a given type in
the type structure, our notion of kind refinements integrate
predicates over {\em types} in the kind structure, enabling the
kinding system to specify and enforce logical constraints on the
structure of types.

A kind refinement, written
$\kref{t{::}\smallkind}{\varphi}$, where $\smallkind$ is a {\em basic}
kind, and $\varphi$ is a logical formula (with $t$ bound in
$\varphi$), characterises types $T$ of kind $\smallkind$ such that the
property $\varphi$ holds of $T$ (i.e.~$\varphi\{T/t\}$ is true). The
language of properties $\varphi$ can refer to the syntax of types,
extended with a refinement-level notion of label set of a (record)
type, $\lb{T}$, and a notion of label set concatenation,
$\mathcal{T} {++}\, \mathcal{S}$, where $\mathcal{T}$ is such an {\em
  extended} type. Refinements $\varphi,\psi$ consist of propositional
logic formulae, (logical) equality, $\mathcal{T} = \mathcal{S}$, an
empty record predicate $\m{empty}(\mathcal{T})$, and basic label set
predicates and such as label inclusion ($\mathcal{T} \in \mathcal{S}$)
and set apartness ($\mathcal{T} \# \,\mathcal{S}$). The intended
target logic is a typed first-order logic with uninterpreted
functions, finite sets, inductive datatypes and equality
\cite{DBLP:conf/cav/BarrettCDHJKRT11}. While such theories are in
general undecidable, the state-of-the-art in SMT solving
\cite{DBLP:journals/lmcs/BansalBRT18,DBLP:conf/cade/ReynoldsTGKDB13}
procedures can be applied to effectively cover the automated reasoning
needed in our work.

Such an extension already provides a significant boost
in expressiveness: By using logical equality in the refinement
formula we can immediately represent singleton kinds such as
$\kref{t{::}\kfunction}{\codom{t} = \tbool}$, the kind
of function types whose image is of $\tbool$ type. Moreover, by combining
kind refinements and type-level functions, we can express non-trivial
type transformations in a fully typed (or kinded) way.
For instance consider the following:
\[
  \begin{array}{c}
  \m{dropField}\triangleq \lambda l {::}\kname.
  \tfix{F}{\kref{r{::}\krecord}{l\in \lb{r}}}{\kref{r{::}\krecord}{l\not\in \lb{r}}}{t}\\
  \ite{\hdlabel{t} = l }
    {\tl{t}}{\treccons{\hdlabel{t}}{\hdtype{t}}{(F\,(\tl{t}))}}
  \end{array}
\]
The function $\m{dropField}$ above takes label $l$ and a record type
with a field labelled by $l$ and removes the corresponding field and type pair
from the record type (recall that $\lb{r}$ denotes the refinement-level
set of labels of $r$). Such a function
combines structural recursion (where $\tl{t}$ is correctly deemed as
structurally smaller than $t$) with our type-level refinement test,
$\ite{\varphi}{T}{S}$. We note that the well-kindedness of such a
function relies on the ability to derive that, when the
record label $\hdlabel{t}$ is not $l$, since we know that $l$ must be
in $t$, $\tl{t}$ is a record type containing $l$. This kind of
reasoning is easily decided using SMT-based techniques
\cite{DBLP:conf/cav/BarrettCDHJKRT11}. 

\subsection{Kinding and Type Equality}
\label{sec:kindtypeeq}

Having introduced the key components of our kind and type language, we
now detail the kinding and type equality rules of our theory, making precise
the various intuitions of previous sections.

The kinding judgment is written $\Ga \vdash T :: K$, denoting that
type $T$ has kind $K$ under the assumptions in the context
$\Ga$. Contexts contain assumptions of the form $t{:}K$, $x{:}T$ and
$\varphi$ -- $t$ stands for a type of kind $K$, $x$ stands for a term
of type $T$ and refinement $\varphi$ is assumed to hold, respectively.
Kinding relies on a context well-formedness judgment, written
$\Ga \vdash$, a kind well-formedness judgment $\Ga \vdash K$,
subkinding judgment $\Ga \vdash K \leq K'$ and the refinement
well-formedness and entailment judgments, $\Ga \vdash \varphi$ and
$\Ga \models \varphi$. Context well-formedness
simply checks that all types, kinds and refinements in $\Ga$ are
well-formed. Kind well-formedness is defined in the standard way,
relying on refinement well-formedness (see
Appendix~\ref{app:kandting}), which requires that formulae and types
in refinements be well-formed.  Subkinding codifies the informal
reasoning from the start of this section, specifying that all
basic kinds are a specialization of $\ktype$; and captures equality of
kinds. Kind equality, written $\Ga \vdash K \equiv K'$, identifies
definitionally equal kinds, which due to the presence of kind
refinements requires reasoning about logically equivalent
refinements. We define
equality between $K$ and $K'$ by requiring $K\leq K'$ and $K' \leq K$.



  We now introduce the key kinding rules for the various types in our
  theory and their associated definitional equality rules. The type
  equality judgment is written $\Ga \models T \equiv S :: K$, denoting
  that $T$ and $S$ are equal types of kind $K$.

  \mypara{{\bf Refinements and Type Properties}}
A kind refinement is introduced by the rule {\sc (kref)} below. Given a type $T$ of kind $\smallkind$ and a {\em valid} property
$\varphi$ of $T$, we are justified in stating that $T$ is of kind
$\kref{t{::}\smallkind}{\varphi}$. 
  \[
    \inferrule*[right=(kref)]
    {\Ga \models \varphi\{T/t\} \quad \Ga \vdash T :: \smallkind }
    {\Ga \vdash T :: \kref{t{::}\smallkind}{\varphi}}
    \quad
    \inferrule*[right=(entails)]
  {\Ga \vdash \varphi \quad \m{Valid}(\llb \Ga \rrb \Rightarrow \llb \varphi\rrb)}
  {\Ga \models \varphi}
  \]
 Rule {\sc (entails)} specifies that a refinement formula is satisfiable if it is well-formed (i.e.,
   a syntactically well-formed boolean expression which may include
   equalities on terms of basic kind) and if the representation
of the context $\Ga$ and the refinement $\varphi$ as an implicational formula is
SMT-valid. The context and refinement representation is discussed in
Section~\ref{sec:eqreasoning}.
 
Crucially, 
since we rely on an underlying logic with inductive types (which
includes constructor and destructor equality reasoning), refinements
can specify the shape of the refined types.
For instance,
the expected $\beta$-equivalence reasoning for records allows us
to derive
$\treccons{\ell}{\tbool \rightarrow \tbool}{\trecnil{}} ::
\kref{t{::}\krecord} {\hdtype{t} = \tbool \rightarrow \tbool}$.
In general, we provide an equality elimination
rule for refinements {\sc (r-eqelim)}, internalizing such equalities
in definitional equality of our theory:

\vspace{-0.25cm}
\[
  \inferrule*[right=(r-eqelim)]
   {\Ga \vdash T :: \kref{t{::}\smallkind}{t = S} \quad \Ga \vdash S ::
    \smallkind}
    {\Ga \vdash T \equiv S :: \smallkind}
  \]

 These principles become particularly interesting when reasoning from
 refinements that appear in type variables. For instance, the type
 $\forall t{::}\kref{f{:}\kfunction}{\dom{f} = \tbool \wedge \codom{f}
   = \tbool } . t\rightarrow \tbool $ can be used to type the term
 $\Lambda t{::}\kref{f{:}\kfunction}{\dom{f} = \tbool \wedge \codom{f}
   = \tbool }. \lambda f{:}t. (f\,\true)$, where $\Lambda$ is the
 binder for polymorphic functions, as usual. Crucially, typing (and
 kinding) exploits not only the fact that we know that the type
 variable $t$ stands for a function type, but also that the domain and
 image are the type $\tbool$, which then warrants the application of
 $f$ to a boolean in order to produce a boolean, despite the basic
 kinding information only specifying that $f$ is of function
 kind. This style of reasoning, which is not readily available even in
 powerful type theories such as that of Coq \cite{Coq:manual}, is akin
 to that of a setting with singleton
 kinds~\cite{DBLP:journals/tocl/StoneH06}.

  As we have shown in Section~\ref{sec:introexamples}, properties can
  also be tested in types through a conditional construct
  $\ite{\varphi}{T}{S}$. Provided that the property $\varphi$ is
  well-formed, if $T$ is of kind $K$ assuming $\varphi$ and $S$ of
  kind $K$ assuming $\neg\varphi$, then the conditional test is
  well-kinded, as specified by the rule \textsc{(k-ite)}.  The
  equality principles for the property test rely on validity of the
  specified property (with a degenerate case where both
  branches are equal types).
  We note that $\Ga , \varphi \vdash T :: K$ can effectively be
  represented as $\Ga , x : \kref{\_}{ \varphi} \vdash T :: K$ where $x$ is
  fresh. This representation encodes $\varphi$ in the context through
  a ``dummy'' refinement that simply asserts the property.

  {\small\[
  \begin{array}{c}
  \inferrule*[right=(k-ite)]
    {\Ga \vdash \varphi \quad \Ga , \varphi \vdash T :: K 
    \quad \Ga , \neg\varphi \vdash S :: K}
  {\Ga \vdash \ite{\varphi}{T}{S} :: K    }\quad
   \inferrule*[right=(eq-iteT)]
    {\Ga \models \varphi \quad
    \Ga ,\varphi \vdash T_1 :: K
    \quad \Ga ,\neg\varphi \vdash T_2 :: K}
    {\Ga \models \ite{\varphi}{T_1}{T_2} \equiv T_1 :: K}
\\[1em]
\inferrule*[right=(eq-iteE)]
    {\Ga \models \neg\varphi  \quad
    \Ga ,\varphi \vdash T_1 :: K
    \quad \Ga ,\neg\varphi \vdash T_2 :: K}
    {\Ga \models \ite{\varphi}{T_1}{T_2} \equiv T_2 :: K}
    \quad
\inferrule*[right=(eq-iteEq)]
    {\Ga \vdash \varphi \quad \Ga ,\varphi \vdash T :: K
    \quad \Ga ,\neg\varphi \vdash T :: K }
    {\Ga \models \ite{\varphi}{T}{T} \equiv T :: K}
    \end{array}
  \]}


\mypara{{\bf Type Functions and Function Types}}
The rules that govern kinding and equality of type-level functions consist of
the standard rules plus the extensionality principles of
\cite{DBLP:journals/tocl/StoneH06}
 (to streamline the
presentation, we omit the congruence rules for equality):

{\small\[
  \begin{array}{c}
   \inferrule*[right=(k-fun)]
    {\Ga \vdash K \quad \Ga , t{:}K \vdash T :: K'}
    {\Ga \vdash \lambda t {::}K.T :: \Pi t{:}K.K'}
    \quad
   \inferrule*[right=(k-app)]
    {\Ga \vdash T :: \Pi t{:}K.K' \quad
     \Ga \vdash S :: K}
   {\Ga \vdash T\, S :: K'\{S/t\}}
  \end{array}
\]
\[
\begin{array}{c}
   \inferrule*[right=(k-ext)]
    {\begin{array}{c}
       \Ga \vdash T :: \Pi t {:}K_1.K_3 \\
       \Ga , t {:}K_1 \vdash
       T\,t :: K_2 \quad x \not\in fv(T)
       \end{array}}
    {\Ga \vdash T :: \Pi t {:}K_1.K_2}\quad
    \inferrule*[right=(eq-funext)]
{\begin{array}{c}
       \Ga \vdash S :: \Pi t {:}K_1.K_3 \quad \Ga \vdash T :: \Pi
       t{:}K_1.K_4 \\
       \Ga , t{:} K_1 \vdash S\,t \equiv T\,t :: K_2
       \end{array}}
    {\Ga \vdash S \equiv T :: \Pi t {:}K_1.K_2}\\[1em]
    \inferrule*[right=(eq-funapp)]
  {\Ga , t{:}K \vdash T :: K'
   \quad \Ga \vdash S :: K}
  {\Ga \models (\lambda t{::}K.T)\,S \equiv T\{S/t\} :: K'\{S/t\} }
    \end{array}
  \]}
Rules {\sc (k-ext)} and {\sc (eq-funext)} allow for basic
extensionality principles on type-level functions. The former states
that an $\eta$-like typing rule, where a type $T$ that is a type-level
function from $K_1$ to $K_3$ can be seen as a type-level function from
$K_1$ to $K_2$ if $T$ applied to a fresh variable of type $K_1$ can
derive a type of kind $K_2$. Rule {\sc (eq-funext)} is the analogous
rule for type equality. We note that such rules, although they allow
us to equate types such as $\lambda t{::}\kref{s{:}\ktype}{t=\tbool
  \rightarrow \tbool }.t$ and $\lambda t{::}\kref{s{:}\ktype}{t=\tbool
  \rightarrow \tbool }.\tbool \rightarrow \tbool$, they do not disturb the decidability
of kinding or equality \cite{DBLP:journals/tocl/StoneH06}.

  Structural recursive functions, defined via a fixpoint construct,
  are defined by:
  \[
    \begin{array}{c}
   \inferrule*[right=(k-fix)]
    {\Ga , F {:} \Pi t{:}K.K' , t{:}K \vdash T :: K'\quad
    \mathsf{structural}(T,F,t)}
    {\Ga \vdash \tfix{F}{K}{K'}{t}T :: \Pi t{:}K.K' }\\[1em]
   \inferrule[(eq-fixunf)]
    {\Ga ,t{:}K_1 \vdash K_2 \quad
    \Ga , F{:}\Pi t{:}K_1.K_2,t{:}K_1 \vdash T :: K_2\quad
    \Ga \vdash S :: K_1 \quad \m{structural}(T,F,t)}
    {\Ga \models (\tfix{F}{K_1}{K_2}{t}T)\,S \equiv
      T\{S/t\}\{(\tfix{F}{K_1}{K_2}{t}T)/F\} :: K_2\{S/t\}}
   \end{array}
 \]
 The predicate $\m{structural}(T,F,t)$ enforces that calls of $F$ in
 $T$ must take arguments that are structurally smaller than $t$
 (i.e.~the arguments must be syntactically equal to $t$ applied to a
 destructor). More precisely, the predicate  $\m{structural}(T,F,t)$ holds iff
 all occurrences of $F$ in $T$ are applied to terms smaller than $t$, where
 the notion of size is given by $\elim{}{t} < t$, where $\elim{}{t}$
 stands for an appropriate destructor applied to $t$
 (e.g.,~ if $t$ is of kind $\kfunction$ then $\dom{t} < t$).
 The equality rule allows for the appropriate unfolding
 of the recursion to take place. Naturally, the implementation of this
 rule follows the standard lazy unfolding approach to recursive
 definitions.

 Polymorphic function types are
 assigned kind $\kpolyfun_K$:
\[\begin{array}{c}
 \inferrule*[right=(k-$\forall$)]
    {\Ga \vdash K \quad \Ga , t{:}K \vdash T :: \smallkind}
    {\Ga \vdash \forall t {::} K. T :: \kpolyfun_K}
    \end{array}
\]
Our manipulation of function types as essentially a pair of types (a
domain type and an image type) gives rise to the following kinding and
equalities:
\[
  \begin{array}{c}
 \inferrule*[right=(k-fun)]
    {\Ga \vdash T :: \smallkind \quad \Ga \vdash S :: \smallkind'}
    {\Ga \vdash T \rightarrow S :: \kfunction}
    \quad
    \inferrule*[right=(k-dom)]
    {\Ga \vdash T :: \kfunction}
    {\Ga \vdash \dom{T} ::\ktype}
    \quad
    \inferrule*[right=(k-codom)]
    {\Ga \vdash T :: \kfunction}
    {\Ga \vdash \codom{T} ::\ktype}
    \\[1em]
  \inferrule*[right=(eq-dom)]
  {\Ga \vdash T :: \smallkind \quad \Ga \vdash S :: \smallkind'}
  {\Ga \models \dom{T\rightarrow S} \equiv T :: \ktype}
  \quad
  \inferrule*[right=(eq-img)]
  {\Ga \vdash T :: \smallkind \quad \Ga \vdash S :: \smallkind'}
    {\Ga \models \codom{T\rightarrow S} \equiv S :: \ktype}
    \end{array}
  \]
  
\mypara{{\bf Records and Labels}}
The kinding rules that govern record type constructors and field labels
are:
\[
\begin{array}{c}
  \inferrule[(k-recnil)]
    {\Ga \vdash}
    {\Ga \vdash \trecnil :: \krecord}
    \quad
    \inferrule[(k-reccons)]
    {\Ga \vdash L :: \kname \quad
    \Ga \vdash T :: \smallkind \quad \Ga \vdash S :: \kref{t : \krecord}{L \not\in \lb{t} }}
    {\Ga \vdash \treccons{L}{T}{S} ::
      \krecord }
    \quad
    \inferrule[(k-label)]
    {\Ga \vdash  \ell \in \nameset}
  {\Ga \vdash \ell :: \kname }\\[1.5em]
    \inferrule[(k-hdt)]
    {\Ga \vdash T :: \kref{t{::}\krecord}{\neg\m{empty}(t)}}
    {\Ga \vdash \hdtype{T} :: \ktype }
    \quad
    \inferrule[(k-hdl)]
    {\Ga \vdash T :: \kref{t{::}\krecord}{\neg\m{empty}(t)}}
    {\Ga \vdash \hdlabel{T} :: \kname}
    \quad
    \inferrule[(k-tail)]
    {\Ga \vdash T :: \kref{t{::}\krecord}{\neg\m{empty}(t)}}
    {\Ga \vdash \tl{T} :: \krecord}
  \end{array}
\]

The rule for non-empty records requires that the tail $S$ of
the record type must {\em not} contain the field label $L$. The rules
for the various destructors require that the record be non-empty,
projecting out the appropriate data. The
equality principles for the three destructors are fairly
straightforward, projecting out the appropriate record type component,
provided the record is well-kinded.

{\small\[
\begin{array}{c}
 \inferrule[(eq-headlabel)]
  {\begin{array}{c}\Ga \vdash L :: \kname \quad
  \Ga \vdash T :: \smallkind \\
     \Ga \vdash S :: \kref{t : \krecord}{L\not\in \lb{t}}
     \end{array}}
  {\Ga \models \hdlabel{\treccons{L}{T}{S}} \equiv L :: \kname}
  \quad
  \inferrule[(eq-headtype)]
  {\begin{array}{c}\Ga \vdash L :: \kname \quad
  \Ga \vdash T ::\smallkind \\
     \Ga \vdash S :: \kref{t : \krecord}{L\not\in \lb{t}}
     \end{array}}
  {\Ga \models \hdtype{\treccons{L}{T}{S}} \equiv T :: \ktype}
\\[1.5em]
  \inferrule[(eq-tail)]
  {\Ga \vdash L :: \kname \quad
  \Ga \vdash T :: \smallkind \quad
  \Ga \vdash S :: \kref{t : \krecord}{L\not\in \lb{t}}}
  {\Ga \models \tl{\treccons{L}{T}{S}} \equiv S :: \krecord}
 \end{array}
\]}

\mypara{{\bf Collections and Reference Types}}
At the level of kinding, there is virtually no difference between a
collection and a reference type. They both denote a structure
that ``wraps'' a single type (the type of the collection elements for
the former and the type of the referenced values in the latter). Thus,
the respective destructor simply unwraps the underlying type.
\[
\begin{array}{c}
\inferrule[(k-col)]
    {\Ga \vdash T :: \smallkind}
    {\Ga \vdash \tcol{T} :: \kcollection}
    \quad
    \inferrule[(k-ref)]
    {\Ga \vdash T :: \kcollection}
  {\Ga \vdash \refof{T} :: \smallkind}
  \quad
   \inferrule[(eq-col)]
  {\Ga \vdash T :: \smallkind}
  {\Ga \models \colof{\tcol{T}} \equiv T :: \ktype }
  \quad
  \inferrule[(eq-ref)]{\Ga \vdash T :: \smallkind}
  {\Ga \models \refof{\tref{T}} \equiv T :: \ktype}
\end{array}
\]

\mypara{{\bf Conversion and Subkinding}}
As we have informally described earlier, our theory of kinds is
predicated on the idea that we can distinguish between the different
specialized types at the kind level. For instance, 
the kind of record types $\krecord$ is a specialisation
of $\ktype$, the kind of all types,
and similarly for the other type-level base constructs of
the theory. We formalise this via a subkinding relation,
which also internalises kind equality, and the corresponding
subsumption rule:
\[
  \begin{array}{c}
  \inferrule*[right=(K-sub)]
  {\Ga \vdash T :: K \quad
   \Ga \vdash K \leq K'}
 {\Ga \vdash T :: K'}
 \quad
  \inferrule*[right=(sub-eq)]
  {\Ga \vdash K \leq K' \quad \Ga \vdash K' \leq K }
  {\Ga \vdash K \equiv K'}
  \quad
  \inferrule*[right=(sub-type)]
  {\Ga \vdash
    }
  {\Ga \vdash \smallkind \leq \ktype}
    \\[1.5em]
    \inferrule*[right=(sub-refkind)]
  {\Ga \vdash \smallkind \quad \Ga , t{:}\smallkind \vdash \varphi }
    {\Ga \vdash \kref{t{::}\smallkind }{\varphi} \leq \smallkind}
    \quad
    \inferrule*[right=(sub-ref)]
    {\Ga \vdash \smallkind \leq \smallkind'
    \quad \Ga , t{:}\smallkind' \models \varphi \Rightarrow \varphi' }
    {\Ga \vdash \kref{t{::}\smallkind}{\varphi} \leq \kref{t{::}\smallkind'}{\varphi'} }
\end{array} 
\]
Rule (\textsc{sub-refkind}) specifies that a refined kind is always a
subkind of its unrefined variant. Rule (\textsc{sub-ref}) allows for
subkinding between refined kinds, by requiring that the basic kind respects
subkinding and that the refinement of the more precise kind implies
that of the more general one.

\mypara{{\bf Kind Case and Bottom}}
The kind case type-level mechanism is kinded in a natural way (rule
({\sc k-kcase})), accounting for the case where the kind of type $T$
matches the specified kind $\smallkind'$ with type $S$ and with type
$U$ otherwise.
\[
  \inferrule*[right=(k-kcase)]
    {\begin{array}{c}
    \Ga \vdash \smallkind \quad
    \Ga \vdash T :: \smallkind''
    \quad
    \Ga , t {:}\smallkind  \vdash S :: K'
    \quad
       \Ga \vdash U :: K'
     \end{array}}
   {\Ga \vdash \tkindofP{T}{\smallkind}{S}{U} :: K'}
   \quad
    \inferrule*[right=(k-bot)]
    {\Ga \models \bot \quad \Ga \vdash K}
    {\Ga \vdash \bot :: K  }
 \]
 Our treatment of $\bot$ allows for $\bot$ to be of {\em any}
 (well-formed) kind, provided one can conclude $\bot$ is valid. The
 associated equality principles implement the kind case by testing the
 specified kind against the derivable kind of type $T$. When $\bot$ is
 provable from $\Ga$ then we can derive any equality via rule
 {\sc (eq-bot)}.

\[
\begin{array}{c}
  \inferrule*[right=(eq-kcaseT)]
  {\begin{array}{c}
     \Ga \vdash T :: \smallkind \quad
  \Ga ,t{:}\smallkind\vdash S :: K'\quad
     \Ga \vdash U :: K'
   \end{array}}
  {\Ga \models  \tkindofP{T}{\smallkind}{S}{U} \equiv S\{T/t\} :: K'}
  
  \quad
      \inferrule*[right=(eq-bot)]
  {\Ga \models \bot \quad \Ga \vdash T :: \smallkind}
  {\Ga \models \bot \equiv T :: \smallkind}

\\[1.5em]

 \inferrule*[right=(eq-kcaseF)]
  {\begin{array}{c}
     \Ga \vdash T :: \smallkind_0 \quad \Ga \vdash \smallkind_0\not\equiv \smallkind\quad
  \Ga ,t{:}\smallkind\vdash S :: K'\quad
     \Ga \vdash U :: K'
   \end{array}}
  {\Ga \models  \tkindofP{T}{K}{S}{U} \equiv U :: K'}
\end{array}
\]

\begin{example}[Representing Record Field Selection in types and values]
  \label{ex:typelabelsel}

  With the development presented up to this point we can 
  implement the more usual record selection operator $T.L$, where
  $T$ is a record type and $L$ is a field label of $T$.
  We represent such a construct as a type-level function that given some
  $L :: \kname$ produces a recursive type-function that essentially iterates
  over a type record of kind $\kref{r{::}\krecord}{\ell \in \lb{r}}$:
  \[
  \begin{array}{c}  
  \lambda L {::} \kname.\tfix{F}{\kref{r{::}\krecord}{L \in \lb{r}}}{\ktype}{t}\\
  \ite{\hdlabel{t} = L}{\hdtype{t}}
    {F(\tl{t})} :: \Pi L : \kname.\Pi t : \kref{r{::}\krecord}{L \in
    \lb{r}} . \ktype
  \end{array}
\]

The function iteratively tests the label at the head of the record
against $L$, producing the type at the head of the record on a match
and recurring otherwise.  It is instructive to consider
the kinding for the property test construct (let $\Ga_0$ be
$L {:}\kname, F{:}\Pi t{:}\kref{r{::}\krecord}{L \in \lb{r}}.\ktype,
t{:}\kref{r{:}\krecord}{L \in \lb{r}}$):

\[
  \inferrule*[right=(k-ite)]
  {\inferrule[]{}{\Ga_0 \vdash \hdlabel{t} = L}
    \qquad
    \mathcal{D}
    \qquad
    \mathcal{E}
  }
  {\Ga_0 \vdash \ite{\hdlabel{t} = L}{\hdtype{t}}
    {F(\tl{t})} :: \ktype}
\]
where $\mathcal{D}$ is a derivation of
$\inferrule[]{}{\Ga_0,\hdlabel{t} = L \vdash \hdtype{t}
  :: \ktype}$ and $\mathcal{E}$ is a derivation of
$\inferrule[]{}{\Ga_0,\neg(\hdlabel{t} = L) \vdash
  F(\tl{t}) :: \ktype}$. To show that $\hdlabel{t} = L$ is
well-formed we must be able to derive 
$t :: \kref{r{::}\krecord}{\neg\m{empty}(r)}$ from
$t :: \kref{r{::}\krecord}{L \in \lb{r}}$, which is achieved via
subkinding, by appealing to entailment in our underlying theory
(see Section~\ref{sec:eqreasoning}).  Similarly, the derivation
$\mathcal{E}$ requires the ability to conclude that
$\tl{t} :: \kref{r{::}\krecord}{L \in \lb{r}}$, using the information that
$t :: \kref{r{::}\krecord}{L \in \lb{r}}$ and
$\neg(\hdlabel{t} = L)$, which is also a valid entailment.
\end{example}

{
  \begin{example}[Generic Pairing of Objects]
    \label{ex:genobjtyp}

The following example consists of an object (implemented as a record
of methods) combinator $\m{Pairer}$
which takes two object types $X$ and $Y$ and for every method that $X$
and $Y$ have in common, $\m{Pairer}\,X\,Y$ contains a method with the
same name and domain types, but where the return type is a pair of the
two original return types. This practical  example is inspired by an
example found in \citet{DBLP:conf/pldi/HuangS08}, which 
uses pattern-based nested reflection in the context of Java.

  We first define a $\m{Pair}$ type constructor as a type-level
  function that takes two types $X$ and $Y$ as argument and produces a two-element
  record, where the label $\m{fst}$ denotes the first element of
  the pair (of type $X$) and the label $\m{snd}$ denotes the second
  element of the pair (of type $Y$):
  \[
    \begin{array}{ll}
\m{Pair} &  :: \Pi X {::} \ktype . \Pi Y {::} \ktype.\kref{ r ::
      \krecord}{\varphi}\\
      \m{Pair} & \triangleq \lambda X {::} \ktype. \lambda Y {::}
                 \ktype. \treccons{\m{fst}}{X}{\treccons{\m{snd}}{Y}{\trecnil}}
    \end{array}
  \]
\[\varphi \triangleq \hdlabel{r} = \m{fst} \,\wedge\, \hdtype{r} = X
\,\wedge\, \hdlabel{\tl{r}} = \m{snd} \,\wedge\, \hdtype{\tl{r}} = Y\]

For the sake of conciseness, we make use of a predicate $\m{isObj}$ on
record types which holds if a record is a record of functions
(i.e. object methods). For simplicity we assume that methods are
functions of exactly one argument.

We now define a $\m{Pairer}$ type-level function which takes two
object types $X$ and $Y$
to produce a new object type which contains the same methods of $X$
and $Y$, but where the methods that $X$ and $Y$ have in common
(i.e.~methods with the same name and same argument types) have as
result type the pairing of the two original return types.

\[
  \begin{array}{ll}
    \m{Pairer} & :: \Pi X :: \kref{r :: \krecord}{\m{isObj}(r)} .
                           \Pi Y :: \kref{r :: \krecord}{\m{isObj}(r)}
                 . \kref{r :: \krecord}{\m{isObj}(r)}\\
    \m{Pairer} & \triangleq \lambda X :: \kref{r ::
                 \krecord}{\m{isObj}(r)} .
                 \lambda Y :: \kref{r :: \krecord}{\m{isObj}(r)} .\\
               & \quad
                 \ite{\neg \m{empty}(X) \,\wedge\, \neg\m{empty}(Y)}
                 {\\ &\quad\quad
                       (\ite{\hdlabel{X} \in \lb{Y} \,\wedge\,
                       \dom{Y.\hdlabel{X} } = \dom{\hdtype{X} }}
                       {\\ &\quad\quad \treccons{\hdlabel{X}}{\dom{\hdtype{X}}
                             \rightarrow \\
             &
               \qquad\qquad\qquad\qquad\quad\m{Pair}\, (\codom{\hdtype{X}})\,
               (\codom{Y.\hdlabel{X}}) }
                             {\m{Helper}}\\ & \qquad }
                { \treccons{\hdlabel{X}}{\hdtype{X}}{\m{Helper})} \\
               & \quad } }
               {\\ & \qquad \ite{\neg\m{empty}(X)}{X}{Y}   }\\

    \m{Helper} & = \ite{(\hdlabel{Y} \in \lb{X} \,\wedge\,
                       \dom{X.\hdlabel{Y} } = \dom{\hdtype{Y}}
                 \,\wedge\,\\
               & \quad \hdlabel{X} \neq \hdlabel{Y}) }
                 {\\ & \qquad \quad
                       \treccons
                       {\hdlabel{Y}}
                       {\dom{\hdtype{Y}} \rightarrow\\ &\qquad\qquad
                       \m{Pair}\, (\codom{X.\hdlabel{Y}}) \, (\codom{\hdtype{Y}}) }
                       {\\ & \qquad\qquad
                             \m{dropField}(\hdlabel{Y},
            \m{Pairer}\,(\tl{X})\,(\m{dropField}(\hdlabel{X},\tl{Y})))\\
               & \quad 
    } }
            {
               \,  \m{Pairer}\,(\tl{X})\,(\m{dropField}(\hdlabel{X},Y))                 }

   \end{array}
 \]
 The $\m{Pairer}$ function above proceeds recursively over the records
 $X$ and $Y$. When $Y$ is empty, the function returns $X$ since there
 is nothing left to pair, and similarly for when $X$ is empty. When
 neither $X$ or $Y$ are empty, we test whether the head label of $X$
 is in the label set of $Y$ with a matching domain type, if not, then
 there is no pairing to be done with the method at the head of $X$ and
 the resulting record copies the method signature from $X$.  If the
 conditional holds, then we produce a function type with the
 appropriate domain and where the image is the pairing of the two
 image types. In both cases (to ease with the formating) the tail of
 the record is defined by a $\m{Helper}$ definition.

 The $\m{Helper}$ definition tests whether the head label of $Y$ is in
 $X$ with matching domain types, but is not the first label of $X$
 (which is handled in the previous test). If the condition holds, then
 we must include the head method of $Y$ with the appropriately paired
 image type. The recursive call to $\m{Pairer}$ makes use of the $\m{dropField}$
type-level function, which removes a record entry, to ensure that the
head label of $X$ is removed from the tail of $Y$ and that the head
label of $Y$ is removed from the result of the recursive call. When
the condition does not hold we simply recurse on the tail of $X$ and
on $Y$ with the method labelled by $\hdlabel{X}$ removed.
\end{example}}

\section{A Programming Language with Kind Refinements}
\label{sec:pl}

\begin{figure}
  \[
    \begin{array}{llcll}
  \mbox{Terms} & M,N & ::= & x \mid \lambda x{:}T.M \mid M\,N & \mbox{Functions}\\
               & & \mid & \typeabs{t{::}K}M \mid \typeapp{M}{T} &  \mbox{Type Abstraction and Application}\\
                 & & \mid & \record{} \mid \reccons{\ell}{M}{N} \mid \tmtl{M}  \\
               & & \mid & \tmhdlabel{M} \mid \tmhdterm{M}& \mbox{Records} \\
                 & & \mid & \unitelem & \mbox{Unit Element}\\
                 & & \mid & \ite{M}{N_1}{N_2}\\
                 & & \mid & \true \mid \false & \mbox{Booleans}\\
                 & & \mid & \ite{\varphi}{M}{N} & \mbox{Property Test}\\
                 & & \mid & \tkindofP{T}{K}{M}{N} & \mbox{Kind Case}\\
                 & & \mid & \emptycol{T} \mid M :: N \\
                 & & \mid & \colcase{M}{N_1}{x {::} xs \Rightarrow N_2} & \mbox{Collections}\\
                 & & \mid & \mkref{M} \mid \deref{M} \mid M := N \mid l & \mbox{References}\\
                 & & \mid & \vrec{F}{T}{M} & \mbox{Recursion}
    \end{array}
\]
\caption{Syntax of Terms\label{fig:synterms}}
\end{figure}

Having covered the key details of kinding and type equality, we
introduce the syntax and typing for
our programming language
{\em per se}, {capturing the essence of} an ML-style functional language with a higher-order
store, 
the syntax of which is given in Figure~\ref{fig:synterms}.  Most
constructs are standard.

We highlight the treatment of records, mirroring that of record {\em types},
as heterogeneous lists of (pairs of) field labels and terms
equipped with the appropriate destructors. Collections are built from
the empty collection $\emptycol{}$ and the concatenation of an element
$M$ with a collection $N$, $M :: N$, with the usual case analysis
$\colcase{M}{N_1}{x {::} xs \Rightarrow N_2}$ that reduces to $N_1$
when $M$ evaluates to the empty collection and to $N_2$ otherwise, where $x$
is instantiated with the head of the collection and $xs$ with its tail.
We allow for recursive terms via a
fixpoint construct $\vrec{F}{T}{M}$,
noting that since there are no type dependencies,
non-termination in the term language does not affect the overall
soundness of the development. We also mirror the type-level property
test and kind case constructs in the term language as
$\ite{\varphi}{M}{N}$ and $\tkindofP{T}{K}{M}{N}$, respectively.  As
we have initially stated, our language has general higher-order
references, represented with the constructs $\mkref{M}$, $\deref{M}$
and $M := N$, which create a reference to $M$, dereference a reference
$M$ and assign $N$ to the reference $M$, respectively.  As usual in
languages with a store, we use $l$ to stand for the runtime values of
memory locations.

\begin{figure}
{\small
\[
  \begin{array}{c}
    \inferrule[(var)]
    {(x{:}T)\in\Ga \quad \Ga ; S \vdash \quad \Ga \vdash }
    {\Ga \vdash_S x : T}
    \quad
    \inferrule[($\one$I)]{ \Ga \vdash}{ \Ga \vdash \unitelem : \one}
    \quad
    \inferrule[($\rightarrow$I)]
    {\Ga \vdash_S T :: \ktype \quad \Ga , x{:}T\vdash_S M : U}
    {\Ga \vdash_S \lambda x{:}T.M : T \rightarrow U}\\[1.5em]
    \inferrule[($\rightarrow$E)]
    {\Ga \vdash_S M : T \rightarrow S\quad \Ga \vdash_S N : T}
    {\Ga \vdash_S M\, N : S}
    \quad
    \inferrule[($\forall$I)]
    {\Ga \vdash K \quad \Ga , t{:}K \vdash_S M : T}
    {\Ga \vdash_S \typeabs{t{::}K}M : \forall t{::}K.T }
    \quad
    \inferrule[($\forall$E)]
    {\Ga \vdash_S M : \forall t :: K.S \quad \Ga \vdash T :: K }
    {\Ga \vdash_S \typeapp{M}{T} : S\{T/t\}}\\[1.5em]
    \inferrule[($\trecnil I_1$)]
    {\Ga \vdash \quad \Ga ; S \vdash }
    {\Ga \vdash_S \record{} : \trecnil }
    \quad
    \inferrule[($\trecnil I_2$)]
    {\Ga \vdash_S L :: \kname \quad  \Ga \vdash_S M : T_1
    \quad \Ga \vdash T_2 :: \kref{t {::} \krecord}{ L \not\in \lb{t}} \quad
    \Ga \vdash_S N : T_2}
    {\Ga \vdash_S \reccons{L}{M}{N} : \treccons{L}{T}{U}}\\[1.5em]
    \inferrule[(reclabel)]
    {\Ga \vdash_S M : \treccons{L}{T}{U}    }
    {\Ga \vdash_S \tmhdlabel{M} : L}\quad
    \inferrule[(recterm)]
    {\Ga \vdash_S M : \treccons{L}{T}{U}   }
    {\Ga \vdash_S \tmhdterm{M} : T}\quad

     \inferrule[(rectail)]
    {\Ga \vdash_S M : \treccons{L}{T}{U}   }
    {\Ga \vdash_S \tl{M} : U}
    \\[1.5em]

    \inferrule[(true)]
    {\Ga \vdash\quad \Ga ; S \vdash}
    {\Ga \vdash_S \true: \tbool}
    \quad
    \inferrule[(false)]
    {\Ga \vdash \quad \Ga ; S \vdash}
    {\Ga \vdash_S \false : \tbool}\quad
    
    \inferrule[(bool-ite)]
    {\Ga \vdash_S M : \tbool \quad
    \Ga \vdash_S N_1 : T \quad
    \Ga \vdash_S N_2 :T }
    {\Ga \vdash_S \ite{M}{N_1}{N_2} : T}\\[1.5em]

    \inferrule[(emp)]
    {\Ga \vdash T :: \ktype \quad \Ga ; S \vdash}
    {\Ga \vdash_S \emptycol{T} : \tcol{T}}
    \quad
    \inferrule[(cons)]
    { \Ga \vdash_S M : T \quad \Ga \vdash_S N :\tcol{T} }
       {\Ga \vdash_S M :: N : \tcol{T}}
    \quad

    \inferrule[(case)]
    {
    \Ga \vdash_S M : \tcol{T} \quad \Ga \vdash N_1 : S \quad
    \Ga ,x{:}T, xs{:}\tcol{T}\vdash N_2 : S
    }
    {\Ga \vdash_S \colcase{M}{N_1}{x {::} xs \Rightarrow N_2} : S }
    \\[1.5em]
 
    \inferrule[(loc)]
    {\Ga \vdash \quad \Ga;S\vdash \quad S(l) = T}{\Ga \vdash_S l : \tref{T}}
    \quad
    \inferrule[(ref)]
    {\Ga \vdash_S M : T}
    {\Ga \vdash_S \mkref{M} : \tref{T}}
    \quad
     \inferrule[(deref)]
    {\Ga \vdash_S M : \tref{T}}
       {\Ga \vdash_S \deref{M} : T}
    \quad
    \inferrule[(assign)]
    {\Ga \vdash_S M : \tref{T} \quad \Ga \vdash_S N : T}
    {\Ga \vdash_S M := N : \one }\\[1.5em]

    \inferrule[(prop-ite)]
    {\Ga \vdash \varphi \quad
    \Ga , \varphi \vdash_S M : T_1 \quad
    \Ga , \neg\varphi \vdash_S N : T_2}
    {\Ga \vdash_S \ite{\varphi}{M}{N} : \ite{\varphi}{T_1}{T_2}}
    \quad
    \inferrule[(kindcase)]
    {\Ga \vdash T :: \smallkind' \quad \Ga \vdash \smallkind \quad
     \Ga , t{:}\smallkind \vdash_S M : U \quad \Ga \vdash_S N : U}
    {\Ga \vdash_S  \tkindofP{T}{\smallkind}{M}{N} : U}\\[1.5em]
    \inferrule[(conv)]
    {\begin{array}{c}
       \Ga \vdash_S M : U \quad \Ga \models U \equiv T :: \ktype
     \end{array}}
    {\Ga \vdash_S M : T}
    \quad
    \inferrule[(fix)]
    {\Ga ,F: T \vdash_S M : T} 
    {\Ga \vdash_S \vrec{F}{T}{M} : T}
    
  \end{array}
\]}
\caption{Typing Rules\label{fig:typing}}
\end{figure}

The typing rules for the language are given in
Figure~\ref{fig:typing}. The typing judgment is written as
$\Ga\vdash_S M :T$, where $S$ is a location typing environment. We
write $\Ga ; S \vdash$ to state that $S$ is a valid mapping from
locations to well-kinded types, according to the typing context $\Ga$.
Notably, despite the fairly advanced type-level features, the typing
rules are virtually unchanged when compared to those of a language in
the ML family.

In fact, the advanced kinding and type equality features
manifest themselves in typing via the {\sc (conv)} 
conversion rule, {\sc (kindcase)} and the {\sc ($\trecnil I_2$)}
record formation rule -- this further reveals a potential strength of
our approach, since it allows for a clean integration of powerful
type-level reasoning and meta-programming without dramatically changing the
surface-level language.
%
%
%
For instance, the following term is well-typed:
\[\begin{array}{c}
  \vdash \Lambda s{:}\ktype.
  \Lambda t{:}\kref{f{::}\kfunction}{\dom{f} = s
    \wedge \codom{f} = \tbool}.\\
    \lambda x{:}t.\lambda y{:}s.(x\,y) :
    \forall s{:}\ktype.\forall t{::}\kref{f{::}\kfunction}{\dom{f} = s
    \wedge \codom{f} = \tbool }.t \rightarrow s
    \rightarrow \tbool
\end{array}
\]
Despite not knowing the exact form of the function type that is to be
instantiated for $t$, by refining its domain and image types we can
derive that $t = s \rightarrow \tbool$ and give a type to applications of
terms of type $t$ correctly.
Note that this is in contrast with what happens in dependent type
theories such as Agda \cite{norell:thesis} or that of Coq
\cite{Coq:manual}), where the leveraging of dependent types, explicit
equality proofs and equality elimination would be needed to provide an
``equivalently'' typed term. 
%
%

We also highlight the typing of the property test term construct,
\[
    \inferrule*[right=(prop-ite)]
    {\Ga \vdash \varphi \quad
    \Ga , \varphi \vdash_S M : T_1 \quad
    \Ga , \neg\varphi \vdash_S N : T_2}
    {\Ga \vdash_S \ite{\varphi}{M}{N} : \ite{\varphi}{T_1}{T_2}}
\]
which types the term $\ite{\varphi}{M}{N}$ with the {\em type}
$\ite{\varphi}{T_1}{T_2}$ and thus allows for a conditional branching
where the types of the branches differ.
Rule ({\sc kindcase}) mirrors
the equivalent rule for the type-level kind case, typing the term
$\tkindofP{T}{\smallkind}{M}{N}$ with the type $U$ of both $M$ and $N$
but testing the kind of type $T$ against $\smallkind$.
Such a construct enables us to define non-parametric polymorphic functions,
and introduce forms of ad-hoc polymorphism. 
For instance, we can derive the following:
\[
  \begin{array}{c}
    \Lambda s{::}\ktype.\lambda x{:}s.\tkindofP{s}{\kreference}
    {(\ite{\refof{t} = \m{Int} }{\deref{x}}{0})}
    {0} : \forall s{::}\ktype.s\rightarrow \m{Int}
    \end{array}
  \]
  The function above takes a type $s$, a term $x$ of that type and, if
  $s$ is of kind $\kreference$ such that $s$ is a reference type for
  integers (note the use of reflection using destructor $\refof{-}$ on type $s$), returns $\deref{x}$, otherwise simply returns $0$. The typing
  exploits the equality rule for the property test where both branches
  are the same type.

  Finally, the type conversion rule ({\sc conv})
  allows us to coerce between equal types, allowing
  for type-level computation to manifest itself in the typing of
  terms. 
  
  \begin{example}[Record Selection]
    \label{ex:termlabelsel}
  Using the record selection type of Example~\ref{ex:typelabelsel} we can
  construct a term-level analogue of record selection.  Given a label
  $L$ and a term $M$ of type $T$ of kind
  $ \kref{r{::}\krecord}{L \in \lb{r}}$, we define the record selection
  construct $M.L$ as (for conciseness, let $\mathcal{R} = \kref{r{::}\krecord}{L \in \lb{r}}$):
  \[
      \begin{array}{c}
    M.L \triangleq
    \Lambda L :: \kname.
    \vrec{F}{\forall t::\mathcal{R}. t\rightarrow (t.L)}
       \Lambda t :: \mathcal{R}.\lambda x{:}t.\\
        \qquad\ite{\hdlabel{t} = L }{\tmhdterm{x}}
    {F[\tl{t}](\tl{x})})[L][T]\,M
      \end{array}
    \]
    such that $M.L : T.L$.
    The typing requires crucial use of type conversion to allow for the unfolding of the recursive type function to take place (let $\Ga_0$ be $L:\kname,
    F{:}\forall t::\mathcal{R}. t\rightarrow (t.L) ,
      x{:}T$):
    \[
      \inferrule[(conv)]
      {\mathcal{D} \qquad
      \Ga_0 \models  (\ite{\hdlabel{T} = L}{\hdtype{T}}{\tl{T}.L})\equiv T.L  :: \ktype }
      {\Ga_0 \vdash \ite{\hdlabel{T} = L}{\tmhdterm{x}}
    {F[\tl{T}](\tl{x})} : T.L  }
    \]
    with $\mathcal{D}$ a derivation of
    \[
\begin{array}{l}
\Ga_0 \vdash \ite{(\hdlabel{T} = L)}{\tmhdterm{x}}
    {F[\tl{T}](\tl{x})} : T_0\\
\end{array}
\] 
where $T_0$ is $\ite{(\hdlabel{T} = L)}{\hdtype{T}}{\tl{T}.L}$,
requiring a similar appeal to logical entailment to that of
Example~\ref{ex:typelabelsel}. Specifically, in the $\kwd{then}$
branch we must show that
$\Ga_0,\hdlabel{T} = L\vdash \tmhdterm{x} :
\hdtype{T}$, which is derivable from $x{:}T$ and
$x : \treccons{\hdlabel{T}}{\hdtype{T}}{\tl{T}}$ --
the latter following from type conversion due to the refinement $L\in
T$ allowing us to establish $\neg\m{empty}(T)$ -- via
typing rule {\sc (recterm)}. 

The $\kwd{else}$ branch requires showing
that
$\Ga_0, \neg\hdlabel{T} = L \vdash F[\tl{T}](\tl{x}) :
\tl{T}.L$, which is derivable from
$F : \forall t::\mathcal{R}. t\rightarrow (t.L)$ and $x{:}T$ as
follows: $\tl{T} :: \mathcal{R}$ follows from
$\neg\hdlabel{T} = L$ and $T :: \mathcal{R}$ (see
Section~\ref{sec:eqreasoning}), thus
$F[\tl{T}] : \tl{T} \rightarrow \tl{T}.L$. Since $\tl{x} : \tl{T}$
from $x : T$ and
$x : \treccons{\hdlabel{T}}{\hdtype{T}}{\tl{T}}$ via rule
{\sc (rectail)}, we conclude using the application rule.
Thus, combining the type and term-level record projection constructs we have that the following is admissible:
\[
  \inferrule[]
  {\Ga \vdash L :: \kname \quad
    \Ga \vdash M : T \quad \Ga \vdash T :: \kref{r{::}\krecord}{L \in \lb{r}} }
  {\Ga \vdash M.L : T.L}
\]
   
\end{example}

{
\begin{example}[Generic Object Pairing]

We now produce the term-level implementation of
Example~\ref{ex:genobjtyp}, which takes two objects $x$ and $y$ of
types $X$ and $Y$ and produces a
new object of type $\m{Pairer}\,X\,Y$. We first define a
constructor for pairs, $\m{PCons}$:
\[
  \begin{array}{ll}
\m{PCons} & : \forall X :: \ktype . \forall Y :: \ktype. X \rightarrow
            Y \rightarrow \m{Pair}\,X\,Y\\
    \m{PCons} & \triangleq
                \Lambda X :: \ktype . \Lambda Y :: \ktype. \lambda
                x{:}X.\lambda y{:}Y. \record{\m{fst} = x , \m{snd} = y}
%
  \end{array}
\]

We now define the pair-object constructor, which makes use of
$\m{Pairer}$ in its typing and of term-level record projection and
record field removal  in its
definition: 
\[
  \begin{array}{ll}
    \m{ObjPair} & : \forall X :: \kref{r :: \krecord}{\m{isObj}(r)} .
                  \forall Y :: \kref{r :: \krecord}{\m{isObj}(r)} .
                  X \rightarrow Y \rightarrow \m{Pairer}\,X\,Y\\
    \m{ObjPair} & \triangleq  \Lambda X :: \kref{r :: \krecord}{\m{isObj}(r)} .
                                            \Lambda Y :: \kref{r :: \krecord}{\m{isObj}(r)} .
                                            \lambda x {:} X . \lambda y {:} Y.\\
& \quad    \ite{\neg \m{empty}(X) \,\wedge\, \neg\m{empty}(Y)}
                 {\\ &\quad\quad
                       (\ite{\hdlabel{X} \in \lb{Y} \,\wedge\,
                       \dom{Y.\hdlabel{X} } = \dom{\hdtype{X} }}
                       {\\ &\quad\quad \reccons{\tmhdlabel{x}}{\lambda z{:}\dom{\hdtype{X}}.
                              \\
             &
               \qquad\qquad\qquad\qquad\quad\m{PCons}\, (\tmhdterm{x}\,z)\,
               (y.\tmhdlabel{x}\,z) }
                             {\m{PHelper}}\\ & \qquad }
                { \reccons{\tmhdlabel{x}}{\tmhdterm{x}}{\m{PHelper})} \\
               & \quad } }
                 {\\ & \qquad \ite{\neg\m{empty}(X)}{x}{y}   }
  \end{array}
\]
\[\begin{array}{ll}

    \m{PHelper} & = \ite{(\hdlabel{Y} \in \lb{X} \,\wedge\,
                       \dom{X.\hdlabel{Y} } = \dom{\hdtype{Y}}
                 \,\wedge\,\\
               & \quad \hdlabel{X} \neq \hdlabel{Y}) }
                 {\\ & \qquad
                       \reccons
                       {\tmhdlabel{y}}
                       {\lambda z{:}\dom{\hdtype{Y}} .\\ &\qquad\quad
                       \m{PCons}\, (x.\tmhdlabel{y}\,z) \, (\tmhdterm{y}\,z) }
                       {\\ & \quad
                             \m{dropField}(\tmhdlabel{y},
            \m{ObjPair}\,(\tl{x})\,(\m{dropField}(\tmhdlabel{x},\tl{y})))\\
               & \quad 
    } }
            {
               \,  \m{ObjPair}\,(\tl{x})\,(\m{dropField}(\tmhdlabel{x},y))                 }
  \end{array}
\]

The structure of the code follows that of the $\m{Pairer}$
definition. The key point is the new method construction, where we
define a function that takes a value $z$ in the domain of the head type of one of the
records and pairs up the result of applying the corresponding methods
of $x$ and $y$ to $z$.

  \end{example}
  }


\section{Operational Semantics and
  Metatheory}\label{sec:metatheory}

We now formulate the operational semantics of our language and
develop the standard type safety results in terms of uniqueness
of types, type preservation and progress.

Since the programming language includes a higher-order store, we
formulate its semantics in a (small-step) store-based reduction
semantics. Recalling that the syntax of the language includes the
runtime representation of store locations $l$, we represent the store
($H,H'$) as a finite map from labels $l$ to values $v$. Given that
kinding and refinement information is needed at runtime for the
property and kind test constructs, we tacitly thread a typing
environment in the reduction semantics.

Moreover, since types in our language are themselves structured
objects with computational significance, we make use of a type
reduction relation, written $T \rightarrow T'$, defined as a
call-by-value reduction semantics on types when seen as a
$\lambda$-calculus.
It is convenient to define a notion of {\em type value},
denoted by $T_v, S_v$ and given by the following grammar:
\[
  \begin{array}{rcl}
    T_v,S_v & ::= & \lambda t{::}K.T \mid \forall t{::}K.T \mid
                    \ell \mid \trecnil \mid \treccons{\ell}{T_v}{S_v} \mid
                    \tcol{T_v} \mid \tref{T_v} \mid T_v \rightarrow S_v \mid
                    \bot \mid \tbool \mid \one \mid t
                   
  \end{array}
\]

We note that it follows from the literature {on $F_\omega$ and related
systems} that type reduction is strongly
normalizing {
\cite{Pierce:2002:TPL:509043,DBLP:conf/popl/StoneH00,norell:thesis,DBLP:conf/icalp/Gimenez98}}. 
The values of the {\em term} language are defined by the grammar:
\[
        \begin{array}{rcl}
          v, v' & ::= & \true \bnfor \false \bnfor \record{} \bnfor \reccons{\ell}{v}{v'} \bnfor \abs{x {:} T_v} M \bnfor \typeabs{t {::} K} M 
                        \bnfor v :: v' \bnfor \emptycol{T}
                        \bnfor l
        \end{array}
\]
Values consist of the booleans $\true$ and $\false$ (extensions to
other basic data types are straightforward as usual); the empty record
$\record{}$; the non-empty record that assigns fields to values,
$\reccons{\ell}{v}{v'}$; the empty collection, $\emptycol{}$, and the
non-empty collection of values, $v :: v'$; as well as type
and $\lambda$-abstraction. For convenience of notation we write
$\trecord{\ell_1:T_1,\dots ,\ell_n:T_n}$ for
$\trecord{\ell_1:T_1}@ \dots @\trecord{\ell_n:T_n}@\trecord{}$, and
similarly $\record{\ell_1 = M_1, \dots , \ell_n = M_n}$ for
$\record{\ell_1 = M_1}@ \cdots @\record{\ell_n = M_n}@\record{}$.

The operational semantics is defined in terms of the judgment
$\stateconf{H}{M} \red \stateconf{H'}{M'}$, indicating that term $M$
with store $H$ reduces to $M'$, resulting in the store $H'$.
For conciseness, we omit congruence rules such as:
\[
  \inferrule*[right=({R-RecConsL})]
    { \stateconf{H}{M} \red \stateconf{H'}{M'} }
    { \stateconf{H}{ \reccons{\ell }{M}{N} } \red
      \stateconf{H'}{\reccons{\ell }{M'}{N}} }
\]
where the record field labelled by $\ell$ is evaluated (and the
resulting modifications in store $H$ to $H'$ are propagated
accordingly). 
The reduction rules enforce a call-by-value,
left-to-right evaluation order and are listed in
Figure~\ref{fig:opsem} (note that we require types occurring in an active
position to be first reduced to a type value, following the call-by-value discipline).
We refer the reader to
Appendix~\ref{app:opsemrules} for the complete set of rules.

\begin{figure}[t]

{\small
  \[
  \begin{array}{c}
    \inferrule*[lab=({R-RecHdLabV})]
    { }
    {\stateconf{H}{\tmhdlabel{\reccons{\ell}{v}{v'}}} \red \stateconf{H}{\ell}  }
\quad
    \inferrule*[lab=({R-RecHdValV})]
    { }
    {\stateconf{H}{\tmhdterm{\reccons{\ell}{v}{v'}}} \red \stateconf{H}{v}  }
\\[1em]
    \inferrule*[lab=({R-RecTailV})]
    { }
    {\stateconf{H}{\tmtl{\reccons{\ell}{v}{v'}}} \red \stateconf{H}{v'}  }
    \quad
    \inferrule*[lab=({R-RefV})]{ l \not\in \m{dom}(H) }{ \stateconf{H}{ \mkref{v}} \red \stateconf{H[l \mapsto v]}{l} }
    \quad
    \inferrule*[lab=({R-DerefV})]{ H(l) = v }{ \stateconf{H}{\deref{l}} \red \stateconf{H}{v} }\\[1em]

    \inferrule*[lab=({R-AssignV})]{ \,}
    { \stateconf{H}{ l := v } \red \stateconf{H[l \mapsto v]}{\unitelem} }
    \quad
    \inferrule*[lab=({R-PropT})]{  \Gamma  \models \varphi }{ \stateconfH{\ite{ \varphi }{M}{N} }  \red \stateconfH{M}}\\[1em]
    \inferrule*[lab=({R-PropF})]{ \Gamma \models \neg \varphi }{ \stateconfH{ \ite{ \varphi }{M}{N} }  \red \stateconfH{N}}
    \quad
      \inferrule*[lab=({R-IfT})]{ }{ \stateconf{H}{ \ite{\true}{M}{N}} \red \stateconf{H}{M}}\\[1em]
    \inferrule*[lab=({R-IfF})]{ }{ \stateconf{H}{\ite{\false}{M}{N}} \red \stateconf{H}{N}}
    \quad
      \inferrule*[lab=({R-Fix})]{  }{ \stateconfH{\vrec{F}{T}{M}}
    \red \stateconf{H}{M\{\vrec{F}{T}{M}/{F}\}}}\\[1em]
    \inferrule*[lab=({R-TAppTRed})]
    {T \rightarrow T'}
    {\stateconf{H}{(\Lambda t{::}K.M)[T]} \red
    \stateconf{H}{(\Lambda t{::}K.M)[T']}}\\[1em]
    \inferrule*[lab=({R-TApp})]{   }{ \stateconfH{(\Lambda t {::} K. \, M) [T_v]}   \red  \stateconfH{M\{T_v/t\}} }
    \quad
    \inferrule*[lab=({R-AppV})]{ 
    }{ \stateconf{H}{(\lambda x: {T_v}. M) \, v} \red \stateconf{H}{M\{v/x\}} }\\[1em]
     \inferrule*[lab=({R-ColCaseEmp})]
    {\, }
    {\stateconf{H}{\colcase{\emptycol{}}{N_1}{x {::} xs \Rightarrow
    N_2}} \red
    \stateconf{H}{N_1}}
\\[1em]
       \inferrule*[lab=({R-ColCaseCons})]
    {\, }
    {\stateconf{H}{\colcase{v :: vs}{N_1}{x {::} xs \Rightarrow
    N_2}} \red
    \stateconf{H}{N_2\{v/x,vs/xs\}}}\\[1em]

    \inferrule*[lab=(R-KindTRed)]
    {T\red T'}
    {\stateconf{H}{\tkindofP{T}{\smallkind}{M}{N} }  \red
    \stateconf{H}{\tkindofP{T'}{\smallkind}{M}{N} }}\\[1em]
    
     \inferrule*[lab=({R-KindL})]{\Ga \vdash T_v :: \smallkind }{
      \stateconf{H}{ \tkindofP{T_v}{\smallkind}{M}{N} }
        \red
        \stateconf{H}{M\{T/t\}}
      } \quad
    \inferrule*[lab=({R-KindR})]{\Ga \vdash T_v :: K_0 \quad \Ga \vdash K_0 \not\equiv \smallkind }{
    \stateconf{H}{ \tkindofP{T_v}{\smallkind}{M}{N} }
        \red
        \stateconf{H}{N}
      } 
  \end{array}
\]
}

\caption{Operational Semantics (Excerpt)~\label{fig:opsem}}
\end{figure}

The three rules for the record destructors project the appropriate
record element as needed. Treatment of references is also standard,
with rule {\sc (R-RefV)} creating a new location $l$ in the store
which then stores value $v$; rule {\sc (R-DerefV)} querying the store
for the contents of location $l$; and rule for {\sc (R-AssignV)}
replacing the contents of location $l$ with $v$ and returning
$v$. Rules {\sc (R-PropT)} and {\sc (R-PropF)} are the only ones that
appeal to the entailment relation for refinements, making use of the
running environment $\Ga$ which is threaded through the reduction
rules straightforwardly. Similarly, rules {\sc (R-KindL)} and {\sc
  (R-KindR)} mimic the equality rules of the kind case construct,
testing the kind of type $T$ against $\mathcal{K}$.

\subsection{Metatheory}

We now develop the main metatheoretical results of type preservation,
progress and uniqueness of kinding and typing.  We begin by noting
that types and their kinding system are essentially as complex as a
type theory with singletons
\cite{DBLP:conf/popl/StoneH00,DBLP:journals/tocl/StoneH06}.  Theories
of singleton kinds essentially amount to $F_\omega$
\cite{GIRARD1986159} with kind dependencies and a fairly powerful but
decidable definitional equality. This is analogous to our development,
but where singletons are replaced by kind refinements and the
additional logical reasoning on said refinements, and the type
language includes additional primitives to manipulate types as data.
Notably, when we consider terms and their typing there is no
significant added complexity since our typing rules are
essentially those of an ML-style, quotiented by a more intricate
notion of type equality.

In the remainder of this section we write $\Ga \vdash \mathcal{J}$ to
stand for a typing, kinding, entailment or equality judgment as
appropriate.  
Since entailment is defined by appealing to SMT-validity, we require some
basic soundness assumptions on the entailment relation, which we list below.

\begin{postulate}[Assumed Properties of Entailment]
  ~\label{prop:refinementprops}
  
  \begin{description}
  \item[Substitution:] If $\Ga \vdash T :: K$ and $\Ga , t{:}K,\Ga' \models \varphi$ then
    $\Ga ,\Ga'\{T/k\} \models \varphi\{T/t\}$;
  \item[Weakening:] If $\Ga \models \varphi$ then $\Ga' \models \varphi$ where $\Ga \subseteq \Ga'$;
  \item[Functionality:] If $\Ga \models T \equiv S :: K$ and
    $\Ga ,t:K,\Ga' \vdash \varphi$ then $\Ga \models \varphi\{T/t\} \Leftrightarrow \varphi\{S/t\}$.
  \item[Soundness:] If $\m{Valid}(\llb\Ga\rrb \Rightarrow
    \llb\varphi\rrb)$, then $\llb\Ga\rrb
    \Rightarrow \llb\varphi\rrb$ is valid; If $\,\m{Valid}(\llb\Ga\rrb \Rightarrow
    \llb\varphi\rrb)$ answers negatively, then it is not the case that
    $\neg(\llb\Ga\rrb \Rightarrow \llb\varphi\rrb)$ is valid.
  \end{description}
\end{postulate}

The general structure of the development is as follows: we first
establish basic structural properties of substitution
(Lemma~\ref{lem:subst}) and weakening, which we can then use to show
that we can apply type and kind conversion inside contexts
(Lemma~\ref{lem:ctxtconv}), which then can be used to show a so-called
{\em validity} property for equality (Theorem~\ref{lem:valideq}),
stating that equality derivations only manipulate well-formed objects (from
which kind preservation -- Lemma~\ref{cor:kindpres} -- follows).

\begin{restatable}[Substitution]{lemma}{subst}
 \label{lem:subst} ~
 \begin{enumerate}
 \item[(a)] If $\Ga \vdash T :: K$ and
 $\Ga , t{:}K , \Ga' \vdash \mathcal{J}$ then $\Ga , \Ga'\{T/t\}
 \vdash \mathcal{J}\{T/t\}$.
\item[(b)] If $\Ga \vdash M : T$ and $\Ga , x{:}T , \Ga' \vdash N : S$
  then $\Ga ,\Ga' \vdash N\{M/x\} : S$.
 \end{enumerate}
\end{restatable}

\begin{restatable}[Context Conversion]{lemma}{ctxtconv}
  \label{lem:ctxtconv}~
  \begin{enumerate}
\item[(a)]  Let $\Ga , x{:}T \vdash$ and $\Ga \vdash T' :: K$.
  If $\Ga , x{:}T \vdash \mathcal{J}$ and $\Ga \models T \equiv T' :: K$ then
  $\Ga,x{:}T' \vdash \mathcal{J}$.

\item[(b)] Let $\Ga , t{:}K \vdash$ and $\Ga \vdash K'$.
  If $\Ga , t{:}K \vdash \mathcal{J}$ and $\Ga \vdash K \leq K'$ then
  $\Ga , t{:}K' \vdash \mathcal{J}$.

  \end{enumerate}
\end{restatable}

\begin{restatable}[Validity for Equality]{theorem}{valideq}
\label{lem:valideq}~
  \begin{enumerate}
  \item[(a)] If $\Ga \vdash K \leq K'$ and $\Ga \vdash$ then $\Ga \vdash K$ and $\Ga \vdash K'$.
  \item[(b)] If $\Ga \models T \equiv T' :: K$ and $\Ga \vdash$ then $\Ga\vdash K$, $\Ga \vdash T :: K$ and
    $\Ga \vdash T' :: K$.
  \item[(c)] If $\Ga \models \psi\Leftrightarrow \varphi$ and $\Ga \vdash$ then $\Ga
    \vdash \psi$ and $\Ga \vdash \varphi$
  \end{enumerate}
\end{restatable}

\begin{restatable}[Kind Preservation]{lemma}{kindpres}
  \label{cor:kindpres}

 If $\Ga \vdash T :: K$ and $T \rightarrow T'$ then $\Ga \vdash T' :: K$.  
\end{restatable}

This setup then allows us to show so-called functionality properties
of kinding and equality (see Appendix~\ref{app:proofs}), stating that
substitution is consistent with our theory's definitional equality and
that definitional equality is compatible with substitution of
definitionally equal terms.

  



With functionality and the previous properties we can then establish
the so-called validity theorem for our
theory, which is a general well-formedness property of the judgments
of the language. Validity is crucial in establishing the various type
and kind inversion
principles (note that the inversion principles become non-trivial due to the
closure of typing and kinding under equality)
necessary to show uniqueness of types and kinds
(Theorem~\ref{thm:unicity}) and type preservation
(Theorem~\ref{thm:preserv}). Moreover, kinding crucially ensures that
all types of refinement kind are such that the corresponding
refinement is SMT-valid.

\begin{restatable}[Unicity of Types and Kinds]{theorem}{unicity}
\label{thm:unicity}~
\begin{enumerate}
\item  If \,$\Ga \vdash M : T$ and $\Ga \vdash M : S$ then $\Ga \vdash T \equiv S :: K$ and $\Ga \vdash K \leq \ktype$.
\item If $\Ga \vdash T :: K$ and $\Ga \vdash T :: K'$ then
    $\Ga \vdash K \leq K'$ or $\Ga \vdash K' \leq K$.
\end{enumerate}
\end{restatable}

In order to state type preservation we first define the usual
notion of well-typed store, written $\Ga \vdash_S H$, denoting that
for every $l$ in $\m{dom}(H)$ we have that $\Ga \vdash_S l : \tref{T}$
with $\cdot \vdash H(l) : T$. We write $S \subseteq S'$ to denote that
$S'$ is an extension of $S$ (i.e. it preserves the location typings of
$S$).

\begin{restatable}[Type Preservation]{theorem}{preserv}
\label{thm:preserv}
  Let $\Ga \vdash_S M : T$ and $\Ga\vdash_s H$. If
  $\stateconf{H}{M} \red \stateconf{H'}{M'}$ then there exists
  $S'$ such that $S \subseteq S'$, $\Ga\vdash_{S'} H'$ and
  $\Ga \vdash_{S'} M' : T$.

\end{restatable}

Finally, progress can be established in a fairly direct manner
(relying on a straightforward notion of progress for the type
reduction relation). The main interesting aspect is that progress
relies crucially on the decidability of entailment due to the
term-level and type-level predicate test construct.

\begin{restatable}[Type Progress]{lemma}{typprog}
If $\cdot \vdash T :: K$ then either $T$ is a type value or $T\rightarrow T'$, for some $T'$.
\end{restatable}

\begin{restatable}[Progress]{theorem}{progress}
\label{thm:progress}
  Let $\cdot\vdash_S M : T$ and $\cdot\vdash_S H$. Then either
  $M$ is a value or there exists $S'$ and $M'$ such that $\stateconf{H}{M}\red\stateconf{H'}{M'}$.

\end{restatable}


\section{Algorithmic Type Checking and Implementation}
\label{sec:eqreasoning}

This section provides a general description of our practical design
choices and OCaml implementation of the type theory of the previous
sections. While a detailed description of the formulation of our typing and
kinding algorithm is not given for the sake of conciseness, we describe
the representation and entailment of refinements and the
implementation strategy for typing, kinding and equality.
%
From a conceptual point of view, type theories either have a very
powerful and undecidable definitional equality (i.e.~extensional type
theories) or a limited but decidable definitional equality
(i.e.~intensional type theories) \cite{DBLP:books/daglib/0088162}.
For instance, the theories underlying Coq and Agda fall under the
latter category, whereas the theory underlying a system such as
NuPRL~\cite{DBLP:books/daglib/0068834} is of the former
variety.
Languages with refinement types such as Liquid
Haskell~\cite{DBLP:conf/icfp/VazouSJVJ14} and
F-Star~\cite{DBLP:conf/icfp/SwamyCFSBY11} (or with limited forms
of dependent types such as Dependent ML~\cite{DBLP:journals/jfp/Xi07})
live somewhere in the middle of the spectrum, effectively equipping
types with a richer notion of definitional equality through refinement
predicates but disallowing the full power of extensional theories
(i.e.~allowing arbitrary properties to be used as refinements). The
goal of such languages is to allow for non-trivial equalities on types
while preserving decidability of type-checking, typically off-loading
the non-trivial reasoning about entailment of refinement predicates to
some external solver.

\mypara{{Kind Refinements through SMT Solving}}
Our approach follows in this tradition, and our system is
implemented by offloading validity checks of refinement predicates
to the SMT solver CVC4 \cite{DBLP:conf/cav/BarrettCDHJKRT11}, embodied
by the rule for refinement entailment (and for subkinding between two
refinement kinds):
\[
    \inferrule*[right=(entails)]
  {\Ga \vdash \varphi \quad \m{Valid}(\llb \Ga \rrb \Rightarrow \llb \varphi\rrb)}
  {\Ga \models \varphi}
\]

The solver includes first-order theories (with equality) on strings,
finite sets and inductive types (with their associated constructors,
destructors and congruence principles), and so allows us to represent our refinement language
in a fairly direct manner. Crucially, since our theory maintains the
 distinction between types and terms, we need only represent the {\em
  type-level} constructs of our theory in the solver.

Types of basic kind are encoded using an inductive type with a
constructor and destructor for each type constructor and destructor in
our language, respectively. Labels are represented by strings
(i.e.~finite sequences).
In this representation, the ``type of all types'' is named $\m{Types}$.
Types of higher-kind are encoded as first-order terms, so they can
be encoded adequately in the theory of the solver. To do this in a
general way, we add a function symbol $\m{appTyp}$ to the theory
that is used to encode type-level application, effectively
implementing defunctionalization
\cite{Reynolds:1972:DIH:800194.805852}. {For instance, if $f$
  is a variable of some higher-kinded type, such that some equality on $f$ is
  present in refinement formulae, e.g.~$\kref{x :: \smallkind}{f(x) =
    t}$, an equation corresponding to $\m{appTyp}(f,x) = \llb t\rrb$
  will be added to the SMT proof context.}

Refinements are encoded as logical formulae that make use of the
theory of finite sets in order to represent reasoning about record
label set membership and apartness. We add two auxiliary functional
symbols to the theory: $\m{isRec} : \m{Types} \rightarrow \m{Bool}$
and $\m{lab} : \m{Types} \rightarrow \m{Set}\,\m{of}\,\m{String}$,
whose meaning is given through appropriate defining axioms. The 
$\m{isRec}$ predicate codifies that a given term (representing a type) is a
well-formed record, specifying that it is either the representation of
the empty record or a cons-cell, such that the label at the head of the
record does not occur in the label set of its tail. $\m{lab}$ encodes
the label set of a record representation, essentially projecting out
its labels accordingly. We can then define apartness of two label
sets (formally, $\m{apart} : (\m{Set}\,\m{of}\,\m{String},
\m{Set}\,\m{of}\,\m{String}) \rightarrow \m{Bool}$) as the formula
that holds iff the intersection of the two sets is empty. Label
concatenation and its lifting to label sets is defined in terms of
string concatenation. The empty record test and its negation is
encoded via an equality test to the empty record and the appropriate
negation. 

To map types to their representation in the SMT solver we make use of a
representation function $\llb {-} \rrb$ on contexts which collects
variable names (which will be universally quantified in the resulting
formula)
and assumed refinements from the context as a
conjunction. Without loss of generality, we assume that all basic kinds
appear at the top level in the context as a refinement, all context variables are
distinct and all bound occurrences of variables are distinct. 

\[
  \begin{array}{lcl}
    \llb \emptyset \rrb & \triangleq & \m{True}\\
    \llb \Ga , t : \kref{x :: \smallkind}{\phi(x)} \rrb & \triangleq &
                                                                     \llb \Ga
                                                                     \rrb
                                                                     \wedge
                                                                       t
                                                                       :
                                                                       \llb
                                                                       \smallkind \rrb
                                                                        
                                                                       \wedge  
                                                                     \llb\phi(t) \rrb
    \\
    \llb \Ga , t : \Pi s : K.K' \rrb & \triangleq  & \llb \Ga \rrb
                                                     \wedge t :
                                                     \m{Types} \\
    
    \llb \smallkind \rrb & \triangleq & \m{Types}\\
    \llb \kref{ x :: \smallkind}{\phi(x)} \rrb & \triangleq & \llb
                                                             \smallkind \rrb
    \end{array}
\]

To simplify the presentation, we overload the $\llb {-} \rrb$ notation
on contexts, types and kinds. All basic kinds are translated to the
representation type $\m{Types}$. At the level of contexts, type
variables of basic kind are translated to a declaration of a variable
of the appropriate target type and the refinement is translated
straightforwardly making use of the auxiliary predicates defined
above. \btnote{To } 
represent type variables $t$ of higher-kind
{we encode them as variables of representation type
  $\m{Types}$, such that occurrences of $t$ in refinements are
  defunctionalized using the technique described above.}

{ \mypara{Outline of the Algorithm} The main function of our
  checker takes a term (i.e.~an expression or a type), an expected type or kind (and other
  auxiliary parameters explained below) and either raises an exception
  if the term is ill-typed or ill-kinded, or returns the type or kind
  of the term. The returned value is a lower bound of the expected type or
  kind. 
  
  The OCaml type signature of the function is:
}

{\small
\begin{verbatim}
val typecheck: termenv -> termenv -> term -> term -> term -> bool -> ihenv -> term
\end{verbatim}
}

{ The parameters of (OCaml) type \texttt{termenv} 
  respectively hold the typing and evaluation contexts. The evaluation
  context holds bindings between variables and corresponding
  value bindings, necessary to normalize (open) types during type
  and kind-checking.  The parameters of OCaml type \texttt{term} are
  respectively the term to be typed, its expected type, and the
  expected kind of such type.
  The \texttt{typecheck} function is also used to (kind) check
  types, in which case the three terms hold the type to be kinded,
  its expected kind, and the (conventional) well-kinded classifier
  \texttt{KindOK} constant. The parameter of type \texttt{bool} is
  used to approximate whether the typing environment is known to be
  inconsistent (used to implement the kinding and typing rules for
  $\bot$), and the parameter of type \texttt{ihenv} holds the
  induction environment for recursive types.

  The algorithm crucially relies on auxiliary functions to compute
  normal forms of types using an evaluation strategy that is confluent
  and normalizing, and relies on the SMT solver to decide conditional
  predicates and equality tests. In particular, unfolding of
  recursive types is only performed when the recursion argument is
  equal (up to refinements) to a constructor (see
  \cite{DBLP:conf/icalp/Gimenez98,Coq:manual}), so that the condition
  imposed on the rules for recursive types (decreasing size on the
  argument on recursive calls) suffices to ensure termination.

  We highlight our adopted solution for interleaving type-level
  computation with type checking of value terms. When considering a
  value term for type checking, the algorithm first considers the
  structure of the (weak-head normalized) expected type. It then
  focuses on the structure of the type if its head term takes the
  form of a conditional, an application, or a recursive type, and
  applies the appropriate typing rule before recursing on the
  structure of the value term. Recursive types are handled using the
  induction environment mentioned above, allowing typing of recursive
  calls to be discharged using the appropriate kind declared in the
  recursive definition, as specified in the corresponding kinding
  rule. We illustrate a run of our type checker (and interpreter) on
  the concrete syntax for Examples~\ref{ex:typelabelsel}
  and~\ref{ex:termlabelsel}, implementing record field selection in types and values.}

{\small
\begin{verbatim}
# type
let Proj:: Pi L:: Lab.Pi t:: { r::Rec | L inl labSet(r) }.Type =
  fun L::Lab -> 
    letrec G :: Pi t :: { r :: Rec | L inl labSet(r) }.Type = 
      fun t :: { r :: Rec | L inl labSet(r) } -> 
        if (headlb(t) == L) then head(t) else (G (tail(t)))
    in G end
in letrec proj : All L::Lab. All t:: { r::Rec | L inl labSet(r) }.(t -> (Proj L t)) =
       fun  L::Lab -> fun t :: { r::Rec | L inl labSet(r) } -> fun r:t ->
         if (headlb(t) == L) then head(r) else (proj L (tail(t)) tail(r))
   in 
       ( proj `a [|`b : int, `a:bool|] [`b=5, `a=false] )  
   end
end;;
 - : bool = false
\end{verbatim}}
\noindent
The type {\small\verb|(Proj L t)|} defines the projection of the type associated with label {\small\verb|L|} in (record type) {\small\verb|t|},
and {\small\verb|(proj L t r)|} defines the projection of value associated with label with label {\small\verb|L|} in record {\small\verb|r|} (of type {\small\verb|t|}). Notice the declared kind of {\small\verb|Proj|} and the declared type of {\small\verb|proj|}.

\mypara{Kinding Algorithm} The implementation of kind checking
follows a standard algorithm for type-checking a $\lambda$-calculus
with lists, pairs, subtyping and structurally recursive function
definitions \cite{Pierce:2002:TPL:509043}.  Kinding rules that make
use of refinements (e.g.,~those that manipulate records) and any
instance of subkinding or kind equality in the presence of refinements
is discharged via the encoding into CVC4. {
Kind-checking (of types) only requires type-level computation
to take place while handling refinements predicates in kinds: those are normalized prior to encoding.}

\mypara{Type Equality} As in most type theories, the crux
of our implementation lies in a suitable implementation of type
equality. Since our notion of type equality has flavours of
extensionality (recall the examples of Section~\ref{sec:kindtypeeq})
and is essentially kind sensitive, we make use of the now folklore
equivalence checking algorithms that exploit weak-head normalization
and type information \cite{Pierce:2004:ATT:1076265}. In our setting,
we use weak-head normalization of {\em types} and exploit {\em
  kinding} information
\cite{DBLP:journals/tocl/StoneH06,DBLP:conf/popl/StoneH00}. The
algorithm alternates between weak-head normalization and kind-oriented
equality checking phases. In the former phase, weak-head reduction of
types that form a $\lambda$-calculus is used. In the latter phase,
extensionality of type-level functions is implemented essentially by
the equivalent of rule {\sc (eq-funext)} read bottom up and
comparisons at basic kinds against variables of refined kind are
offloaded to the SMT solver, implementing extensionality for types of
basic kind (e.g.,~deriving that $t \equiv \m{Bool}\rightarrow\m{Bool}$
if
$t :: \kref{f{:}\kfunction}{\dom{f} = \tbool \wedge
  \codom{f} = \tbool  }$). The type checking algorithm
itself (which makes use of the type equality algorithm) is mostly straightforward,
since the typing rules of our language are essentially those of an
ML-like language (without modules).

{In terms of our overall approach to type and kind-checking, it
  follows closely known type-checking algorithms for related systems
  and so is clearly sound. Completeness holds relative to the
  underlying SMT theories, as is generally the case in related works
  on SMT-based
  refinement~\cite{DBLP:conf/icfp/VazouSJVJ14,DBLP:conf/esop/VazouRJ13}.}
Our prototype implementation consists of around 5000 lines of OCaml
code (not counting the generated lexer and parser code) which includes
parsuning, kind-checking, type-checking and an interpreter for our
system (using the operational semantics of
Section~\ref{sec:metatheory}). The implementation validates the
examples found in the paper.  {The interaction with the SMT
  solver to discharge refinements produces some overheads, but that is
  mostly due to the unoptimized nature of our proof-of-concept implementation. 
  }

\section{Related Work}
\label{sec:rw}

To the best of our knowledge, ours is the first work to explore the
concept of refinement kind and illustrate their expressiveness as a
practical language feature that 
integrates statically typed meta-programming features such as type
reflection, ad-hoc polymorphism, and type-level computation which
allows us to specify structural properties of function, collection and
record types.

The concept of
refinement kind is a natural adaptation of the well-known notion of
refinement type
\cite{DBLP:journals/toplas/BengtsonBFGM11,DBLP:conf/pldi/RondonKJ08,DBLP:conf/esop/VazouRJ13}, 
which effectively extends type specifications with (SMT decidable)
logical assertions.  Refinement types have been applied to various
verification domains such as security
\cite{DBLP:journals/toplas/BengtsonBFGM11} or the verification of
data-structures
\cite{DBLP:conf/pldi/XiP98,DBLP:conf/pldi/KawaguchiRJ09}, and are
being incorporated in full-fledged programming languages, e.g., ML
\cite{DBLP:conf/pldi/FreemanP91} Haskell
\cite{DBLP:conf/icfp/VazouSJVJ14}, F-Star
\cite{DBLP:conf/icfp/SwamyCFSBY11}, JavaScript
\cite{DBLP:conf/pldi/VekrisCJ16}.

With the aim of supporting common meta-programming idioms in the
domain of web programming, \citet{DBLP:conf/pldi/Chlipala10} develops
a type system to support type-level record computations with
similar aims as ours, avoiding type dependency.  In our case, we
generalize type-level computations to other types as data, and rely on
more amenable explicit type dependency, in the style of System-F
polymorphism. Therefore, we still avoid the need to pollute programs with
explicit proof terms, but through our development of a principled theory
of kind refinements.
The idea of expressing
constraints (e.g., disjointness) on record labels with predicates goes
back to \cite{DBLP:conf/popl/HarperP91}. We note that our system admits 
convenient predicates and operators in the refinement logic that are
applicable not just to record types, but also to other kinds of types
such as function and collection types.

The work of \citet{DBLP:conf/haskell/KiselyovLS04} implements a library
of strongly-typed heterogeneous collections in Haskell via an encoding
using the language extensions of multi-parameter type classes and
functional dependencies. Their library includes heterogeneous lists
and extensible records, with a semantics that is akin to that of our
record types.
Since their development is made on top of Haskell and its type-class
system, they explicitly encode all the necessary type manipulation
(type-level) functions through the type-class system. To do this, they
must also encode several auxiliary type-level data such as type-level
natural numbers, type-level booleans, type-level occurrence and
deletion predicates, to name but a few. To adequately manipulate these
types, they also reify type equality and type unification as explicit
type classes. This is in sharp contrast with our development, which
leverages the expressiveness of refinement kinds to produce the same
style of reasoning but with significantly less machinery.
We also highlight the work of~\citet{DBLP:conf/dsl/LeijenM99}, a domain
specific embedded compiler for SQL in Haskell by using so-called
phantom types, which follows a related approach.

\citet{DBLP:journals/pacmpl/MorrisM19} study a general framework of
extensible data types by introducing a notion of row theory which gives
a general account of record concatenation and projection. Their work
is based on a generalization of row types using qualified types that
can refer to some properties of row containment and combination. The
ability to express these properties at the type-level is similar to
our work, although we can leverage the more general concept of
refinement kind to easily express programs and structural properties
of records that are not definable in their work: the $\m{Map}$ and
$\m{SetGetRec}$ record transformations from
Section~\ref{sec:introexamples}, the ability to state that a record
\emph{does not} contain a given label~\cite{GJ-Nott-86}, or the
general case of a combinator that takes two records $R_1$ and $R_2$
and produces a record where each label $\ell$ is mapped to
$R_1.\ell \rightarrow R_2.\ell$. Their work develops an encoding
of row theories into System F satisfying coherence.
It would be interesting to explore a similar encoding of our work into
an appropriate $\lambda$-calculus such as $F_\omega$ with product
types. 

\citet{DBLP:conf/icfp/WeirichHE13} study an extension to
the core language (System FC) of the Glasgow Haskell Compiler (GHC)
with a notion of kind equality proofs, in order to allow type-level
computation in Haskell to refer to kind-level functions. Their
development, being based on System FC, is designed to manipulate
explicit type and kind coercions as part of the core language
itself, which have a non-trivial structure (as required by the various
type features and extensions of GHC), and so differs significantly
from our work which is designed to keep type and kind conversion as
implicit as possible. However, their work can be seen as a stepping
stone towards the integration of refinement kinds and related
constructs in a general purpose language with an advanced
typing system such as Haskell.

Our extension of the concept of refinements to kinds, together with
the introduction of primitives to reflectively manipulate types as
data (cf. ASTs) and express constraints on those data also highlights how
kind refinements match fairly well with the programming practice of
our time (e.g., interface reflection in Java-like languages),
contrasting the focus of our work with the goals of other approaches
to meta-programming such as
\citet{DBLP:conf/ifip2-1/AltenkirchM02,DBLP:journals/jfp/CalcagnoMS03}.
{The work of Altenkirch and McBride takes a dual approach to ours:
While we take the stance of not having a dependently typed language,
their work starts from a dependent type theory with universes and
so-called large eliminations and shows how one can encode generic
programming (i.e., the ability to define functions by recursion on the
structure of datatypes) by defining an appropriate universe of types
and a coding function. Thus, their general framework is naturally more
expressive than ours, but lacks the general purpose practical
programming setting of ours.

The work of Calcagno et al. focuses on multi-staging in the presence
of effects. Staged computation is a form of
meta-programming where program fragments can be safely quoted and
executed in different phases. This form of metaprogramming is
fundamentally different from that formulated in our work, being
targeted towards efficiency and optimizations based on {\em safe}
partial evaluation. }

%
%
%
%
The concept of a statically checked type-case construct was
introduced by \citet{DBLP:journals/toplas/AbadiCPP91}; however, our
refinement kind checking of dynamic type conditionals on types and
kinds $\ite{\varphi}{e_1}{e_2}$ and $\tkindof{T}{t}{K}{e_1}{e_2}$
greatly extends the precision of type and kind checking, and 
supports very flexible forms of statically checked ad-hoc
polymorphism, as we have shown.

Some works
\cite{DBLP:conf/aplas/SmaragdakisBKB15,DBLP:conf/pldi/HuangS08,DBLP:conf/gpce/FahndrichCL06}
have addressed the challenge of typing specific meta-programming
idioms in {real-world general purpose} languages such as Java
and C\# {(or significant fragments of those languages).
  By using the standard record-based encoding of objects (as done in
  the examples of Sections~\ref{sect:intro}
  and~\ref{sec:introexamples}), several of the meta-programming
  patterns found in their works are representable using our framework
  of refinement kinds (e.g., generating constructors, changing field
  types, generating accessor and modifier methods). However, since
  those works target object-oriented languages, they support
  OO-specific features that are out of the scope of our work
  (e.g. inheritance, method visibility), which does not deal with
  object orientation concepts but rather with a minimal ML-style
  language in order to illustrate the core ideas and their general
  expressiveness.}

  {We further highlight the recent work of
  \citet{DBLP:conf/pldi/KazerounianGVFH19}, which addresses arbitrary
  type-level computation in Ruby libraries and injects appropriate
  {\em run-time checks} to ensure that library methods abide by their
  computed type. Their work allows for
  arbitrary Ruby functions to be called during type-level computation
  and is thus more expressive than our more strictly stratified
  framework. Their typing discipline also exploits singleton-like
  types, that can be used in the context of database column and table
  names, to assign precise type signatures to database query methods
  (i.e.,~type-level computations can interact with a database to find
  the schema of a table that is then used as part of a type).
  While we can define types of the form $F(T_1) \rightarrow T_2$,
  where the domain type is the result of a computation $F$ on $T_1$,
  we are restricted to more limited type-level reasoning, whereas
  their work is closer to general dependent types in this sense. For
  instance, we can define:
  \[
    \begin{array}{l}
    \m{ColRecPreds} \triangleq \lambda s :: \ktype .\\
    \hspace{2.5cm}(\tkindofP{s}{\kcollection}{ (\tkindofP{\colof{t}}{\krecord}{t}{\one})  \\\hspace{3cm}}
    {\one}) \rightarrow \tbool
    %
    \end{array}
  \]
  where $\m{ColRecPreds}$ is a type-level
  function that given a type $s$, provided $s$ is a
  collection type of records of some record type,
produces the type of predicates on that record type (i.e.~a function
from that record type to $\tbool$) and
otherwise returns the  trivial predicate type (i.e.~$\one
\rightarrow \tbool$, where $\one$ is the unit type).
 We can use $\m{ColRecPreds}$ to type a program akin to a generic record
existence test in a table (i.e.~a collection of records):
\[
  \begin{array}{l}
    \m{exists}
      :    \forall C :: \kcollection . C \rightarrow \m{ColRecPreds} (C)\\
    
    \m{exists} \triangleq \Lambda C :: \kcollection .
    \tkindofP{\colof{C}}{\krecord}
    {\vrec{F}{C \rightarrow \m{ColRecPreds} (C)}
    {\\\hspace{3cm}\lambda c{:}C.\lambda x{:}t. \colcase{c}{\false\\\hspace{5.5cm}}{r {::} rs \Rightarrow
    \ite{r=x}{\true}{F(rs)(x)}}}\\\hspace{2.3cm}}
    {\lambda c{:}C.\lambda x{:}\one . \false}

\end{array}
\]
The example above contrasts with
\cite{DBLP:conf/pldi/KazerounianGVFH19}, where a related example is
formulated such that the equivalent of our $\m{ColRecPreds}$
type-level function actually queries a database for the appropriate
table schema, whereas $\m{ColRecPreds}$ can only inspect the ``shape''
of its type arguments to obtain the table schema (i.e.~the types of records
contained in the collection).

}

Our work shows
how the fundamental concept of refinement kinds suggests
itself as a general type-theoretic principle that accounts for statically
checked typeful \cite{Cardelli91} meta-programming, including programs
that manipulate types as data, or build types and programs from data
(e.g., as the type providers of
F\#~\cite{DBLP:conf/pldi/PetricekGS16}) which seems to be out of reach
of existing static type systems. Our language conveniently expresses
programs that automatically generate types and operations from data
specifications, while statically ensuring that generated types satisfy
the intended invariants expressed by refinements.


\section{Concluding Remarks}
\label{sec:conc}

This work introduces the concept of refinement kinds and develops its
associated type theory, in the context of higher-order polymorphic
$\lambda$-calculus with imperative constructs, several kinds of
datatypes, and type-level computation.  The resulting programming
language supports static typing of sophisticated features such
as type-level reflection with ad-hoc and parametric polymorphism,
which can be elegantly combined to implement non-trivial
meta-programming idioms, as we have illustrated with several
examples. Crucially, the typing system for our language is
essentially that of an ML-like language but with a more intricate
notion of type equality and kinding, which are defined independently
from typing.

We have validated our theory by establishing the standard type safety
results and by further developing a prototype implementation for our
theory, making use of the SMT solver
CVC4~\cite{DBLP:conf/cav/BarrettCDHJKRT11} to discharge the
verification of refinements. Our implementation
 demonstrates the
practicality and effectiveness of our approach, and validates all examples in the paper. 
Moreover, as discussed in Section~\ref{sec:eqreasoning}, apart from the peculiarities
specific to the refinement logic,
our implementation is not significantly more involved than standard algorithms
for type-checking system $F_\omega$ or those for singleton kinds
\cite{Pierce:2002:TPL:509043,Pierce:2004:ATT:1076265,DBLP:conf/popl/StoneH00}.

There are many interesting avenues of exploration that have been
opened by this work: From a theoretical point-of-view, it would be
instructive to study the tension imposed on shallow embeddings of our
system in general dependent type theories such as Coq.  After
including existential types, variant types and higher-type imperative
state (e.g., the ability to introduce references storing types at the
term-level), which have been left out of this presentation for the
sake of focus, it would be relevant to investigate limited
forms of dependent or refinement types. It would be also
interesting to investigate how refinement kinds and stateful types
(e.g., typestate or other forms of behavioral types) may be used to
express and type-check invariants on meta-programs with challenging
scenarios of strong updates, e.g., involving changes in representation
of abstract data types.

{Following the approach of
  \citet{DBLP:conf/pldi/KazerounianGVFH19}, it would be interesting to
  study a version of our theory of refinement kinds that is applied to
  a real-world dynamically typed language by inserting run-time checks
  to ensure methods follow their specified types.}

The relationship between our refinement kind
system and the notion of type class \cite{DBLP:conf/popl/WadlerB89},
popularised by Haskell~\cite{DBLP:journals/toplas/HallHJW96},
also warrants further investigation. Type classes integrate ad-hoc
polymorphism with parametric polymorphism by allowing
for the specification of functional interfaces that
quantified types must satisfy. In principle, type classes
can be realized by appropriate type-level records of functions and may
thus be representable in our general framework.
%
Finally, to
ease the burden on programmers, we plan to investigate how to
integrate our algorithmic system with partial type inference mechanisms.


\begin{acks}
{The authors would like to thank the anonymous referees for their
valuable comments and suggestions. This work is supported by
\grantsponsor{NL}{NOVA LINCS}{}
\grantnum{NL}{(Ref. UID/CEC/04516/2019)},
\grantsponsor{FCT}{Funda\c{c}\~{a}o para a Ci\^{e}ncia e Tecnologia
  project CLAY}{}
\grantnum{FCT}{(PTDC/EEICTP/4293/2014)}, and the \grantsponsor{NL}{NOVA
  LINCS \& OutSystems}{} \grantnum{NL}{FLEX-AGILE project}.

}
\end{acks}


\bibliography{references,refs2}

\newpage
\appendix

\hspace{0pt}
\vfill
\begin{center}
{\Huge \bf Appendix}\\[1em]

Refinement Kinds:\\
Type-safe Programming with Practical Type-level Computation\\[1em]

Additional definitions and proofs of the main materials.

\end{center}
\vfill
\hspace{0pt}

\newpage

\section{Full Syntax, Judgments and Rules}
\label{app:rules}

We define the syntax of kinds $K, K'$, refinements
$\varphi, \varphi'$, types $T,S,R$, and terms $M,N$ below. We assume
countably infinite sets of type variables $\typevarset$, names
$\nameset$ and term variables $\varset$. We range over type variables
with $t, t', s, s'$, name variables with $n, m$ and term variables
with $x, y, z$.

\[
   \begin{array}{llcll}
      \mbox{Kinds} & K,K' & ::= & \smallkind  \mid
                                  \kref{t{::}\smallkind}{\varphi} \mid \Pi t{:}K.K' & \mbox{Refined and Dependent Kinds} \\
                  &\smallkind & ::= & \krecord \mid \kcollection \mid \kfunction \mid
                                      \kreference \mid \kname & \mbox{Base Kinds}\\
                   & & \mid & \ktype \mid \kpolyfun_K  \\[1em]

      \mbox{Types} & T,S,R & ::= & t \mid \lambda t{::}K.T \mid T\,S & \mbox{Type-level Functions}\\
      
                 & & \mid &   \tfix{F}{K}{K'}{t}T    & \mbox{Structural Recursion}\\
                   & & \mid & \forall t{::}K.T 
                                                                                    & \mbox{Polymorphism}\\
               
                   & & \mid & L \mid \trecnil \mid \treccons{L}{T}{S} & \mbox{Record Type constructors}\\
                   & & \mid & \hdlabel{T} \mid \hdtype{T} \mid \tl{T} & \mbox{Record Type destructors}\\

                   & & \mid & \tcol{T} \mid \colof{T} & \mbox{Collection Types} \\
                 & & \mid & \tref{T} \mid \refof{T} & \mbox{Reference Types} \\
                 & & \mid & T \rightarrow S \mid  \dom{T} \mid \codom{T} & \mbox{Function Types} \\
                   & & \mid & \tkindofP{T}{\smallkind}{S}{U} & \mbox{Kind Case}\\
                   & & \mid & \ite{\varphi}{T}{S} & \mbox{Property Test}\\ 
                   & & \mid & \bot \mid \top  & \mbox{Empty and Top Types}\\

                   & & \mid & \tbool \mid \one \mid \dots  & \mbox{Basic Data Types}\\[1em]

     \mbox{Extended Types} & \mathcal{T},\mathcal{S} & ::= &T \mid \lb{T}
                                                        \mid \mathcal{T} {++}\,  \mathcal{S} 
                                                                          \\[1em]
     
      \mbox{Refinements} & \varphi,\psi & ::= &
                   \varphi \supset \psi \mid
                              \varphi \wedge \psi \mid \dots &
                                                               \mbox{Propositional Logic}\\
          & & \mid &   \m{empty}(\mathcal{T} ) & \mbox{Empty Record Test}\\
                  & & \mid & \mathcal{T} = \mathcal{S} &
                                                         \mbox{Equality}\\
          & & \mid & \mathcal{T} \in \mathcal{S} & \mbox{Label Set
                                                   Inclusion}\\
        & & \mid & \mathcal{T} \# \,\mathcal{S} & \mbox{Label Set Apartness}
                                                           \end{array}  
  \]
  \[
    \begin{array}{llcll}
  \mbox{Terms} & M,N & ::= & x \mid \lambda x{:}T.M \mid M\,N & \mbox{Functions}\\
               & & \mid & \typeabs{t{::}K}M \mid \typeapp{M}{T} &  \mbox{Type Abstraction and Application}\\
                 & & \mid & \record{} \mid \reccons{\ell}{M}{N} \mid \tmtl{M}  \\
               & & \mid & \tmhdlabel{M} \mid \tmhdterm{M}& \mbox{Records} \\
                 & & \mid & \unitelem & \mbox{Unit Element}\\
                 & & \mid & \ite{M}{N_1}{N_2}\\
                 & & \mid & \true \mid \false & \mbox{Booleans}\\
                 & & \mid & \ite{\varphi}{M}{N} & \mbox{Property Test}\\
                 & & \mid & \tkindofP{T}{K}{M}{N} & \mbox{Kind Case}\\
                 & & \mid & \emptycol{T} \mid M :: N \\
                 & & \mid & \colcase{M}{N_1}{x {::} xs \Rightarrow
                            N_2}
                                                              & \mbox{Collections}\\
                 & & \mid & \mkref{M} \mid \deref{M} \mid M := N \mid l & \mbox{References}\\
                 & & \mid & \vrec{F}{T}{M} & \mbox{Recursion}
    \end{array}
\]

\subsection{Kinding and Typing}
\label{app:kandting}

Our type theory is defined by the following judgments:
\[
  \begin{array}{ll}
    \Ga \vdash & \mbox{$\Ga$ is a well-formed context}\\
    \Ga \vdash K & \mbox{$K$ is a well-formed kind under the assumptions in $\Ga$}\\
    \Ga \vdash \varphi & \mbox{Refinement $\varphi$ is well-formed under the assumptions in $\Ga$}\\
    \Ga \vdash T :: K & \mbox{Type $T$ is a (well-formed) type of kind $K$ under the assumptions in $\Ga$}\\
    \Ga \vdash_S M : T & \mbox{Term $M$ has type $T$ under the assumptions in $\Ga$ and store typing $S$}\\

    \Ga \models \varphi & \mbox{Refinement $\varphi$ holds under the assumptions in $\Ga$}\\
   \Ga \vdash K \equiv K' & \mbox{Kinds $K$ and $K'$ are equal}\\
    \Ga \vdash K \leq K' & \mbox{Kind $K$ is a sub-kind of $K'$}\\
    \Ga \vdash T \equiv T' :: K & \mbox{Types $T$ and $T'$ of kind $K$ are equal}\\
  \end{array}
\]




\paragraph{Context Well-formedness}
\[
  \begin{array}{c}
    \inferrule[]
    {\Ga \vdash K \quad \Ga \vdash}
    {\Ga , t {:} K \vdash  }
    \quad
    \inferrule[]
    {\Ga \vdash T :: \ktype \quad \Ga \vdash}
    {\Ga , x {:} T \vdash }
    \quad
    \inferrule[]
    {\Ga \vdash \varphi \quad \Ga \vdash}
    {\Ga , \varphi \vdash }
    \quad
    \inferrule[]
    {\Ga;S\vdash \quad \Ga \vdash T :: \smallkind}
    {\Ga;S,l:T \vdash}
    \quad
    \inferrule[]
    {\, }
    {\cdot \vdash }\\[1.5ex]
    \inferrule[]
    {\, }
    {\Ga;\cdot \vdash }
  \end{array}
\]

\paragraph{Kind well-formedness}

\[
  \begin{array}{c}
    \inferrule[]
    {\Ga \vdash \quad K \in \{ \krecord,\kcollection ,\kfunction, \kreference,
                  \kname, \ktype \} }
     {\Ga \vdash K}
    \quad
    \inferrule[]
    {\Ga \vdash K \quad \Ga , t{:}K \vdash K'}
    {\Ga \vdash \Pi t{:}K.K'}
    \quad

    \inferrule[]{\Ga \vdash K }
    {\Ga \vdash \kpolyfun_K}
    \quad
    \inferrule[]
    { \Ga \vdash \smallkind \quad \Ga , t{:}\smallkind \vdash \varphi}
    {\Ga \vdash \kref{t {::} \smallkind}{\varphi}}
  \end{array}
\]

\paragraph{Refinement Well-formedness}

Refinement well-formedness simply requires context well-formedness and that all logical predicates
are well-sorted (i.e.~logical expressions of type $\tbool$). All types
occurring in refinements must be well-kinded (we write $p$ to stand
for any logical predicate or uninterpreted function of the theory with
the appropriate sort):
\[
 \inferrule[]
  {
  \forall i \in \{1,\dots,n\} . \Ga \vdash T_i :: \smallkind}
{\Ga \vdash p(T_1, \dots , T_n) }
\]

\paragraph{Refinement Satisfiability}A refinement is satisfiable if it is well-formed and if the representation
of the context $\Ga$ and the refinement $\varphi$ as an implicational formula is
SMT-valid.
\[
  \inferrule[]
  {\Ga \vdash \varphi \quad \m{Valid}(\llb \Ga \rrb \Rightarrow \llb \varphi\rrb)}
  {\Ga \models \varphi}
\]

\paragraph{Kinding}

\[
  \begin{array}{c}
    \inferrule[]
    { t{:}K \in \Ga \quad \Ga \vdash }
    { \Ga \vdash t :: K    }
    \quad
        \inferrule[]
    {\Ga \vdash T :: K \quad \Ga \vdash K \leq K'}
    {\Ga \vdash T :: K'}
    \quad
    \inferrule[]
    {\Ga \vdash}
    {\Ga \vdash \tbool :: \ktype}
    \quad
    \inferrule[]
    {\Ga \vdash}
    {\Ga \vdash \one :: \ktype}
    \quad
     \inferrule[]
    {\Ga \vdash  \ell \in \nameset}
    {\Ga \vdash \ell :: \kname }\\[1.5em]

     \inferrule[]
    {\Ga \vdash T :: \Pi t{:}K.K' \quad
     \Ga \vdash S :: K}
    {\Ga \vdash T\, S :: K'\{S/t\}}
    \quad
    \inferrule[]
    {\Ga \vdash K \quad \Ga , t{:}K \vdash T :: K'}
    {\Ga \vdash \lambda t {::}K.T :: \Pi t{:}K.K'}
    \quad
    \inferrule[]
    {\begin{array}{c}
       \Ga \vdash T :: \Pi t {:}K_1.K_3 \\
       \Ga , t {:}K_1 \vdash
       T\,t :: K_2 \quad x \not\in fv(T)
       \end{array}}
    {\Ga \vdash T :: \Pi t {:}K_1.K_2}\\[1.5em]

    \inferrule[]
    {\Ga \vdash K \quad \Ga , t{:}K \vdash T :: \ktype}
    {\Ga \vdash \forall t {::} K. T :: \kpolyfun_K}
    \\[1.5em]

      \inferrule[]
    {\Ga \vdash}
    {\Ga \vdash \trecnil :: \krecord}
    \quad
    \inferrule[]
    {\Ga \vdash L :: \kname \quad
    \Ga \vdash T :: \smallkind \quad \Ga \vdash S :: \kref{t {::} \krecord}{L \not\in \lb{t} }}
    {\Ga \vdash \treccons{L}{T}{S} ::
    \krecord }\\[1.5em]

    \inferrule[]
    {\Ga \vdash T :: \kref{t{::}\krecord}{\neg\m{empty}(t)}}
    {\Ga \vdash \hdtype{T} :: \ktype }
    \quad
    \inferrule[]
    {\Ga \vdash T :: \kref{t{::}\krecord}{\neg\m{empty} (t)}}
    {\Ga \vdash \hdlabel{T} :: \kname}
    \quad
    \inferrule[]
    {\Ga \vdash T :: \kref{t{::}\krecord}{\neg\m{empty} (t)}}
    {\Ga \vdash \tl{T} :: \krecord}\\[1.5em]


    \inferrule[]
    {\Ga \vdash T :: \ktype \quad \Ga \vdash S :: \ktype}
    {\Ga \vdash T \rightarrow S :: \kfunction}
    \quad
    \inferrule[]
    {\Ga \vdash T :: \kfunction}
    {\Ga \vdash \dom{T} ::\ktype}
    \quad
    \inferrule[]
    {\Ga \vdash T :: \kfunction}
    {\Ga \vdash \codom{T} ::\ktype}
\\[1.5em]

    \inferrule[]
    {\Ga \vdash T :: \ktype}
    {\Ga \vdash \tcol{T} :: \kcollection}
    \quad
    \inferrule[]
    {\Ga \vdash T :: \kcollection}
    {\Ga \vdash \colof{T} :: \ktype}
    \quad
    \inferrule[]
    {\Ga \colof{T} :: \ktype}
    {\Ga \vdash T :: \kcollection}\\[1.5em]
    
    \inferrule[]
    {\Ga \vdash T :: \ktype}
    {\Ga \vdash \tref{T} :: \kreference}
    \quad
    \inferrule[]
    {\Ga \vdash T :: \kreference}
    {\Ga \vdash \refof{T} :: \ktype}
    \quad
    \inferrule[]
    {\Ga \vdash \refof{T} :: \ktype}
    {\Ga \vdash T :: \kreference}
    \\[1.5em]

 \inferrule[]
    {\Ga \vdash \varphi \quad \Ga , \varphi \vdash T :: K \quad \Ga , \neg\varphi \vdash S :: K}
    {\Ga \vdash \ite{\varphi}{T}{S} :: K    }
    \quad
    \inferrule[]
    {\begin{array}{c}
    \Ga \vdash \smallkind \quad
    \Ga \vdash T :: \smallkind''
    \quad
    \Ga , t {:}\smallkind  \vdash S :: K'
    \quad
       \Ga \vdash U :: K'
     \end{array}}
    {\Ga \vdash \tkindofP{T}{\smallkind}{S}{U} :: K'}\\[1.5em]

    \inferrule[]
    {\Ga , F {:} \Pi t{:}K.K' , t{:}K \vdash T :: K'\quad
    \mathsf{structural}(T,F,t)}
    {\Ga \vdash \tfix{F}{K}{K'}{t}T :: \Pi t{:}K.K' }
    \quad
       \inferrule[]
    {\Ga \models \varphi\{T/t\} \quad \Ga \vdash T :: \smallkind }
    {\Ga \vdash T :: \kref{t{:}\smallkind}{\varphi}}\\[1.5em]

  \end{array}
\]

\paragraph{Kind Equality and Sub-kinding}

\[
  \begin{array}{c}
    \inferrule[]
    {\Ga \vdash}
    {\Ga \vdash \smallkind \leq \ktype }
    \quad
     \inferrule[]
    {\Ga \vdash \varphi}
    {\Ga \vdash \kref{t{::}\smallkind}{\varphi} \leq \smallkind}
    \quad
    \inferrule[]
    {\m{Valid}(\llbracket \Ga \rrbracket \wedge \llbracket \varphi
    \rrbracket \Rightarrow \llbracket \psi \rrbracket)}
    {\Ga \vdash \kref{t{::}\smallkind}{\varphi} \leq \kref{t{::}\smallkind}{\psi}}
    \\[1.5em]
    \inferrule[]
    {\Ga \vdash K_3 \leq K_1 \quad \Ga , t {:} K_3 \vdash K_2 \leq K_4}
    {\Ga \vdash \Pi t {::} K_1 . K_2 \leq \Pi t {::} K_3 . K_4  }
    \quad
    \inferrule[]
    {\Ga \vdash K_1 \leq K_2 \quad \Ga \vdash K_2 \leq K_1}
    {\Ga \vdash K_1 \equiv K_2}
  \end{array}
\]

\paragraph{Type Equality}

\[
  \begin{array}{c}
    \mbox{Reflexivity, Transitivity, Symmetry} +\\[1em]
    \inferrule[]
    {\Ga \vdash T :: \kref{t{::}\smallkind}{t = S} \quad \Ga \vdash S ::
    \smallkind}
    {\Ga \vdash T \equiv S :: \smallkind}
    \\[1.5em]
   \inferrule[]
  {\Ga \vdash T_1 \equiv S_1 :: \Pi t{:}K_1.K_2
    \quad
    \Ga \vdash T_2 \equiv S_2 :: K_1}
    {\Ga \vdash T_1\,T_2 \equiv S_1\,S_2 :: K_2\{T_2/t\}}
    \quad    \inferrule[]
    {\begin{array}{c}
       \Ga \vdash S :: \Pi t {:}K_1.K_3 \quad \Ga \vdash T :: \Pi
       t{:}K_1.K_4 \\
       \Ga , t{:} K_1 \vdash S\,t \equiv T\,t :: K_2
       \end{array}}
    {\Ga \vdash S \equiv T :: \Pi t {:}K_1.K_2}\\[1.5em]

     \inferrule[]
  {\Ga \vdash K_1 \equiv K_1'
   \quad \Ga , t{:}K_1\vdash T_1 \equiv T_2 :: K_2 }
  {\Ga \vdash \lambda t{::}K_1.T_1 \equiv \lambda t{::}K_1'.T_2 :: \Pi t{:}K_1.K_2}
  \quad
  \inferrule[]
  {\Ga , t{:}K \vdash T :: K'
   \quad \Ga \vdash S :: K}
  {\Ga \vdash (\lambda t{::}K.T)\,S \equiv T\{S/t\} :: K'\{S/t\}
    }\\[1.5em]

     \inferrule[]
  {\Ga \vdash K_1 \equiv K_2 \quad
   \Ga , t{:}K_1 \vdash T \equiv S :: \smallkind}
  {\Ga \vdash \forall t{::}K_1.T \equiv \forall t{:}K_2.S ::
  \kpolyfun_{K_1}}
    \\[1.5em]


      \inferrule[]
  {\Ga \vdash L \equiv L' :: \kname\quad
  \Ga \vdash T\equiv T' :: \smallkind\quad
  \Ga \vdash S \equiv S' :: \kref{t {::} \krecord}{L\not\in \lb{t} }}
  {\Ga \vdash \treccons{L}{T}{S} \equiv
  \treccons{L'}{T'}{S'} :: \krecord}
    
  \end{array}
\]
\[
  \begin{array}{c}
\inferrule[]
  {\Ga \vdash T \equiv S ::
  \kref{r{::}\krecord}{\neg\m{empty}(r)}  }
  {\Ga \vdash \hdlabel{T} \equiv \hdlabel{S} :: \kname}
  \quad
  \inferrule[]
  {\Ga \vdash T \equiv S ::  \kref{r{::}\krecord}{\neg\m{empty} (r)}  \quad  }
  {\Ga \vdash \hdtype{T} \equiv \hdtype{S} :: \ktype}
\\[1.5em]
  \inferrule[]
  {\Ga \vdash T \equiv S ::\kref{r{::}\krecord}{\neg\m{empty} (r)}  \quad
  }
    {\Ga \vdash \tl{T} \equiv \tl{S} :: \krecord}
    \quad
  \inferrule[]
  {\Ga \vdash L :: \kname \quad
  \Ga \vdash T :: \smallkind \quad
  \Ga \vdash S :: \kref{t {::} \krecord}{L\not\in \lb{t} }}
  {\Ga \vdash \hdlabel{\treccons{L}{T}{S}} \equiv L :: \kname}
  \\[1.5em]
  \inferrule[]
  {\Ga \vdash L :: \kname \quad
  \Ga \vdash T ::\smallkind \quad
  \Ga \vdash S :: \kref{t {::} \krecord}{L\not\in \lb{t}}}
  {\Ga \vdash \hdtype{\treccons{L}{T}{S}} \equiv T :: \ktype}
\\[1.5em]
  \inferrule[]
  {\Ga \vdash L :: \kname \quad
  \Ga \vdash T :: \smallkind \quad
  \Ga \vdash S :: \kref{t {::} \krecord}{L\not\in \lb{t}}}
  {\Ga \vdash \tl{\treccons{L}{T}{S}} \equiv S :: \krecord}\\[1.5em]


  \inferrule[]
  {\Ga \vdash T \equiv S :: \ktype}
  {\Ga \vdash \tcol{T} \equiv \tcol{S} :: \kcollection}\quad
  \inferrule[]
  {\Ga \vdash T \equiv S ::
  \kcollection}
  {\Ga \vdash \colof{T} \equiv \colof{S} :: \ktype}\quad
  \inferrule[]
  {\Ga \vdash T :: \ktype}
    {\Ga \vdash \colof{\tcol{T}} \equiv T :: \ktype }\\[1.5em]

     \inferrule[]
  {\Ga \vdash T \equiv S :: \ktype}
    {\Ga \vdash \tref{T} \equiv \tref{S} :: \kreference}\\[1.5em]
    
     \inferrule[]
  {\Ga \vdash T \equiv S ::
  \kreference}
  {\Ga \vdash \refof{T} \equiv \refof{S} :: \ktype}\quad
  \inferrule[]
  {\Ga \vdash T :: \ktype}
    {\Ga \vdash \refof{\tref{T}} \equiv T :: \ktype }\\[1.5em]

    \inferrule[]
  {\Ga \vdash T \equiv S :: \ktype
   \quad \Ga \vdash T' \equiv S' :: \ktype}
  {\Ga \vdash T \rightarrow T' \equiv S \rightarrow S' ::
    \kfunction}\\[1.5em]
    \inferrule[]
  {\Ga \vdash T \equiv S :: \kfunction}
  {\Ga \vdash \dom{T} \equiv \dom{S} ::\ktype}
  \quad
  \inferrule[]
  {\Ga \vdash T \equiv S :: \kfunction}
  {\Ga \vdash \codom{T} \equiv \codom{S} :: \ktype}\\[1.5em]
  \inferrule[]
  {\Ga \vdash T :: \smallkind \quad \Ga \vdash S :: \smallkind'}
  {\Ga \vdash \dom{T\rightarrow S} \equiv T :: \ktype}
  \quad
  \inferrule[]
  {\Ga \vdash T :: \smallkind \quad \Ga \vdash S :: \smallkind'}
    {\Ga \vdash \codom{T\rightarrow S} \equiv S :: \ktype}\\[1.5em]

  \end{array}
\]

\[
  \begin{array}{c}
 \inferrule[]
  {\begin{array}{c}
     \Ga \vdash T \equiv T' :: \smallkind_0 \quad
  \Ga \vdash \smallkind \equiv \smallkind' \quad
  \Ga , t{:}\smallkind \vdash S \equiv S' :: K''\quad
     \Ga \vdash U \equiv U' :: K''
   \end{array}}
  {\Ga \vdash \tkindofP{T}{\smallkind}{S}{U} \equiv
  \tkindofP{T'}{\smallkind'}{S'}{U'} :: K'' }\\[1.5em]

  \inferrule[]
  {\begin{array}{c}
     \Ga \vdash T :: \smallkind \quad
  \Ga ,t{:}\smallkind\vdash S :: K'\quad
     \Ga \vdash U :: K'
   \end{array}}
  {\Ga \vdash  \tkindofP{T}{\smallkind}{S}{U} \equiv S\{T/t\} :: K'}\\[1.5em]

 \inferrule[]
  {\begin{array}{c}
     \Ga \vdash T :: \smallkind_0 \quad \Ga \vdash \smallkind_0\not\equiv \smallkind\quad
  \Ga ,t{:}\smallkind\vdash S :: K'\quad
     \Ga \vdash U :: K'
   \end{array}}
  {\Ga \vdash  \tkindofP{T}{\smallkind}{S}{U} \equiv U :: K'}\\
    \end{array}
\]

\[
  \begin{array}{c}
\inferrule[]
  {\Ga \models \varphi \Leftrightarrow \psi \quad
  \Ga,\varphi \vdash T_1 \equiv S_1 :: K\quad
  \Ga , \neg\varphi \vdash T_2 \equiv S_2 :: K}
    {\Ga \vdash \ite{\varphi}{T_1}{T_2} \equiv \ite{\psi}{S_1}{S_2} :: K }\\[1.5em]

    \inferrule[]
    {\Ga \vdash \varphi \quad
    \Ga ,\varphi \vdash T_1 :: K
    \quad \Ga ,\neg\varphi \vdash T_2 :: K}
    {\Ga \vdash \ite{\varphi}{T_1}{T_2} \equiv T_1 :: K}
    \quad
    \inferrule[]
    {\Ga \vdash \neg\varphi  \quad
    \Ga ,\varphi \vdash T_1 :: K
    \quad \Ga ,\neg\varphi \vdash T_2 :: K}
    {\Ga \vdash \ite{\varphi}{T_1}{T_2} \equiv T_2 :: K}
      \\[1.5em]
    \inferrule[]
    {\Ga \vdash \varphi \quad \Ga,\varphi \vdash T ::K \quad
     \Ga , \neg\varphi \vdash T :: K}
    {\Ga \vdash \ite{\varphi}{T}{T} \equiv T :: K}
    \quad
    \inferrule[]
    {\Ga \vdash T \equiv S :: K \quad \Ga \vdash K \leq K'}
    {\Ga \vdash T\equiv S :: K'}
    \\[1.5em]

    \inferrule[]
  {\begin{array}{c}
     \m{structural}(T,F,t)\quad  \m{structural}(S,F,t)\\
     \Ga \vdash K_1 \equiv K_1'\quad
  \Ga \vdash K_2 \equiv K_2' \quad
     \Ga, F{:}\Pi t{:}K_1.K_2,t{:}K_1 \vdash T \equiv S :: K_2
     \end{array}}
  {\Ga\vdash \tfix{F}{K_1}{K_2}{t}T \equiv
    \tfix{F}{K_1'}{K_2'}{t}S :: \Pi t{:}K_1.K_2  }\\[1.5em]

    \inferrule[]
    {\Ga ,t{:}K_1 \vdash K_2 \quad
    \Ga , F{:}\Pi t{:}K_1.K_2,t{:}K_1 \vdash T :: K_2\quad
    \Ga \vdash S :: K_1 \quad \m{structural}(T,F,t)}
    {\Ga \vdash (\tfix{F}{K_1}{K_2}{t}T)\,S \equiv
      T\{S/t\}\{(\tfix{F}{K_1}{K_2}{t}T)/F\} :: K_2\{S/t\}}
\end{array}    
\]

\paragraph{Typing}

For readability we omit the store typing environment from all rules
except in the location typing rule. In all other rules the store
typing is just propagated unchanged.

{\small
\[
   \begin{array}{c}
    \inferrule[(var)]
    {(x{:}T)\in\Ga \quad \Ga ; S \vdash \quad \Ga \vdash }
    {\Ga \vdash_S x : T}
    \quad
    \inferrule[($\one$I)]{ \Ga \vdash}{ \Ga \vdash \unitelem : \one}
    \quad
    \inferrule[($\rightarrow$I)]
    {\Ga \vdash_S T :: \ktype \quad \Ga , x{:}T\vdash_S M : U}
    {\Ga \vdash_S \lambda x{:}T.M : T \rightarrow U}\\[1.5em]
    \inferrule[($\rightarrow$E)]
    {\Ga \vdash_S M : T \rightarrow S\quad \Ga \vdash_S N : T}
    {\Ga \vdash_S M\, N : S}
    \quad
    \inferrule[($\forall$I)]
    {\Ga \vdash K \quad \Ga , t{:}K \vdash_S M : T}
    {\Ga \vdash_S \typeabs{t{::}K}M : \forall t{::}K.T }
    \quad
    \inferrule[($\forall$E)]
    {\Ga \vdash_S M : \forall t :: K.S \quad \Ga \vdash T :: K }
    {\Ga \vdash_S \typeapp{M}{T} : S\{T/t\}}\\[1.5em]
    \inferrule[($\trecnil I_1$)]
    {\Ga \vdash \quad \Ga ; S \vdash }
    {\Ga \vdash_S \record{} : \trecnil }
    \quad
    \inferrule[($\trecnil I_2$)]
    {\Ga \vdash_S L :: \kname \quad  \Ga \vdash_S M : T_1
    \quad \Ga \vdash T_2 :: \kref{t {::} \krecord}{ L \not\in \lb{t}} \quad
    \Ga \vdash_S N : T_2}
    {\Ga \vdash_S \reccons{L}{M}{N} : \treccons{L}{T}{U}}\\[1.5em]
    \inferrule[(reclabel)]
    {\Ga \vdash_S M : \treccons{L}{T}{U}    }
    {\Ga \vdash_S \tmhdlabel{M} : L}\quad
    \inferrule[(recterm)]
    {\Ga \vdash_S M : \treccons{L}{T}{U}   }
    {\Ga \vdash_S \tmhdterm{M} : T}\quad

     \inferrule[(rectail)]
    {\Ga \vdash_S M : \treccons{L}{T}{U}   }
    {\Ga \vdash_S \tl{M} : U}
    \\[1.5em]

    \inferrule[(true)]
    {\Ga \vdash\quad \Ga ; S \vdash}
    {\Ga \vdash_S \true: \tbool}
    \quad
    \inferrule[(false)]
    {\Ga \vdash \quad \Ga ; S \vdash}
    {\Ga \vdash_S \false : \tbool}\quad
    
    \inferrule[(bool-ite)]
    {\Ga \vdash_S M : \tbool \quad
    \Ga \vdash_S N_1 : T \quad
    \Ga \vdash_S N_2 :T }
    {\Ga \vdash_S \ite{M}{N_1}{N_2} : T}\\[1.5em]

    \inferrule[(emp)]
    {\Ga \vdash T :: \ktype \quad \Ga ; S \vdash}
    {\Ga \vdash_S \emptycol{T} : \tcol{T}}
    \quad
    \inferrule[(cons)]
    { \Ga \vdash_S M : T \quad \Ga \vdash_S N :\tcol{T} }
       {\Ga \vdash_S M :: N : \tcol{T}}
    \quad

   \inferrule[(case)]
    {
    \Ga \vdash_S M : \tcol{T} \quad \Ga \vdash N_1 : S \quad
    \Ga ,x{:}T, xs{:}\tcol{T}\vdash N_2 : S
    }
    {\Ga \vdash_S \colcase{M}{N_1}{x {::} xs \Rightarrow N_2} : S }
     \\[1.5em]
 
    \inferrule[(loc)]
    {\Ga \vdash \quad \Ga;S\vdash \quad S(l) = T}{\Ga \vdash_S l : \tref{T}}
    \quad
    \inferrule[(ref)]
    {\Ga \vdash_S M : T}
    {\Ga \vdash_S \mkref{M} : \tref{T}}
    \quad
     \inferrule[(deref)]
    {\Ga \vdash_S M : \tref{T}}
       {\Ga \vdash_S \deref{M} : T}
    \quad
    \inferrule[(assign)]
    {\Ga \vdash_S M : \tref{T} \quad \Ga \vdash_S N : T}
    {\Ga \vdash_S M := N : \one }\\[1.5em]

    \inferrule[(prop-ite)]
    {\Ga \vdash \varphi \quad
    \Ga , \varphi \vdash_S M : T_1 \quad
    \Ga , \neg\varphi \vdash_S N : T_2}
    {\Ga \vdash_S \ite{\varphi}{M}{N} : \ite{\varphi}{T_1}{T_2}}
    \quad
    \inferrule[(kindcase)]
    {\Ga \vdash T :: \smallkind' \quad \Ga \vdash \smallkind \quad
     \Ga , t{:}\smallkind \vdash_S M : U \quad \Ga \vdash_S N : U}
    {\Ga \vdash_S  \tkindofP{T}{\smallkind}{M}{N} : U}\\[1.5em]
    \inferrule[(conv)]
    {\begin{array}{c}
       \Ga \vdash_S M : U \quad \Ga \models U \equiv T :: \smallkind
     \end{array}}
    {\Ga \vdash_S M : T}
    \quad
    \inferrule[(fix)]
    {\Ga ,F: T \vdash_S M : T \quad \m{structural}(F,M)}
    {\Ga \vdash_S \vrec{F}{T}{M} : T}
    
   \end{array}
 \]}



\section{Full Operational Semantics}
\label{app:opsemrules}

The type reduction relation, $T \rightarrow T'$ is defined as a call-by-value
reduction semantics on types $T$, obtained by orienting the computational rules of type
equality from left to right (thus excluding rule ({\sc eq-elim}))
and enforcing the call-by-value discipline.
Recalling that type values are denoted by $T_v, S_v$ and given by the following grammar:
\[
  \begin{array}{rcl}
    T_v,S_v & ::= & \lambda t{::}K.T \mid \forall t{::}K.T \mid
                    \ell \mid \trecnil \mid \treccons{\ell}{T_v}{S_v} \mid
                    \tcol{T_v} \mid \tref{T_v} \mid T_v \rightarrow S_v \mid
                    \bot \mid \tbool \mid \one \mid t
                   
  \end{array}
\]
The type reduction rules are:
\[
  \begin{array}{c}
    \inferrule[]
    {T\rightarrow T'}
    {T\,S \rightarrow T'\,S}
    \quad
    \inferrule[]
    {S \rightarrow S'}
    {(\lambda t{::}K.T)\,S \rightarrow (\lambda t{::}K.T)\,S'}
    \quad
    \inferrule[]
    {T\not\rightarrow }
    {(\lambda t{::}K.T)\,S_v \rightarrow T\{S_v/t\}}
    \quad
       \inferrule[]
    {T \rightarrow T'}
    {\lambda t {::} K.T \rightarrow \lambda t {::} K.T'} 
    \\[1.5em]
    
    \inferrule[]
    {T \rightarrow T'}
    {\forall t {::} K.T \rightarrow \forall t {::} K.T'}
    \quad
    \inferrule[]
    { }
    { (\tfix{F}{K}{K'}{t}T)\,S_v \rightarrow T\{S_v/t\}\{\tfix{F}{K}{K'}{t}T/F\} }\\[1.5em]

    \inferrule[]
    {L \rightarrow L'}
    {\treccons{L}{T}{S} \rightarrow \treccons{L'}{T}{S}}
    \quad
    \inferrule[]
    {T \rightarrow T'}
    {\treccons{\ell}{T}{S} \rightarrow \treccons{\ell}{T'}{S}}
    \quad
    \inferrule[]
    {S \rightarrow S'}
    {\treccons{\ell}{T_v}{S} \rightarrow \treccons{\ell}{T_v}{S'}}\\[1.5em]

    \inferrule[]
    {T\rightarrow T'}
    {\hdlabel{T} \rightarrow \hdlabel{T'}}
    \quad
    \inferrule[]
    {T\rightarrow T'}
    {\hdtype{T} \rightarrow \hdtype{T'}}
    \quad
    \inferrule[]
    {T\rightarrow T'}
    {\tl{T} \rightarrow \tl{T'}}\\[1.5em]

    \inferrule[]
    { }
    {\hdlabel{\treccons{\ell}{T_v}{S_v}} \rightarrow \ell }
    \quad
    \inferrule[]
    { }
    {\hdtype{\treccons{\ell}{T_v}{S_v}} \rightarrow T_v }
    \quad
    \inferrule[]
    { }
    {\tl{\treccons{\ell}{T_v}{S_v}} \rightarrow S_v }
  \end{array}\]

\[\begin{array}{c}

    \inferrule[]
    {T \rightarrow T'}
    {\tcol{T} \rightarrow \tcol{T'}}
    \quad
    \inferrule[]
    {T \rightarrow T'}
    {\colof{T} \rightarrow \colof{T'}}
    \quad
    \inferrule[]
    { }
    {\colof{\tcol{T_v}} \rightarrow T_v }\\[1.5em]

    \inferrule[]
    {T \rightarrow T'}
    {\tref{T} \rightarrow \tref{T'}}
    \quad
    \inferrule[]
    {T \rightarrow T'}
    {\refof{T} \rightarrow \refof{T'}}
    \quad
    \inferrule[]
    { }
    {\refof{\tref{T_v}} \rightarrow T_v }\\[1.5em]

    \inferrule[]
    {T\rightarrow T'}
    {(T\rightarrow S) \rightarrow (T'\rightarrow S)}
    \quad
    \inferrule[]
    {S\rightarrow S'}
    {(T_v\rightarrow S) \rightarrow (T_v\rightarrow S')}
    \quad
    \inferrule[]
    {T\rightarrow T'}
    {\dom{T} \rightarrow \dom{T'}}
    \quad
 \inferrule[]
    {T\rightarrow T'}
    {\codom{T} \rightarrow \codom{T'}}
    \\[1.5em]

    \inferrule[]{ }{\dom{T_v\rightarrow S_v} \rightarrow T_v}
    \quad
    \inferrule[]{ }{\codom{T_v\rightarrow S_v} \rightarrow S_v}\\[1.5em]

    \inferrule[]
    {\Ga \models \varphi}
    { \ite{\varphi}{T}{S} \rightarrow T}
    \quad 
    \inferrule[]
    {\Ga \models \neg\varphi}
    { \ite{\varphi}{T}{S} \rightarrow S}
    \quad
    \inferrule[]
    {T \rightarrow T'}
    {\tkindofP{T}{\smallkind}{S}{U} \rightarrow \tkindofP{T'}{\smallkind}{S}{U}}
    \\[1.5em]
    \inferrule[]
    { \Ga \vdash T_v :: \smallkind}
    {\tkindofP{T_v}{\smallkind}{S}{U} \rightarrow S\{T_v/t\}}
    \quad
    \inferrule[]
    { \Ga \vdash T_v :: \smallkind' \quad \Ga \vdash \smallkind'\not\equiv\smallkind}
    {\tkindofP{T_v}{\smallkind}{S}{U} \rightarrow U}

  \end{array}
\]

The rules of our operational semantics are as follows:
\[
  \begin{array}{c}
 \inferrule*[lab={R-RecConsLab}]
    { \stateconf{H}{L} \red \stateconf{H'}{L'} }
    { \stateconf{H}{ \reccons{L }{M}{N} } \red
    \stateconf{H'}{\reccons{L' }{M}{N}} }\\[1.5em]
    
    \inferrule*[lab={R-RecConsL}]
    { \stateconf{H}{M} \red \stateconf{H'}{M'} }
    { \stateconf{H}{ \reccons{\ell }{M}{N} } \red
    \stateconf{H'}{\reccons{\ell }{M'}{N}} }
    \quad
    \inferrule*[lab={R-RecConsR}]
    { \stateconf{H}{M} \red \stateconf{H'}{M'} }
    { \stateconf{H}{ \reccons{\ell }{v}{M} }
    \red \stateconf{H'}{\reccons{\ell }{v}{M'}} }
    \\[1.5em]

    \inferrule*[lab={R-RecHdLab}]
    {\stateconf{H}{M} \red \stateconf{H'}{M'} }
    {\stateconf{H}{\tmhdlabel{M}} \red \stateconf{H'}{\tmhdlabel{M'}}  }
    \quad
    \inferrule*[lab={R-RecHdLabV}]
    { }
    {\stateconf{H}{\tmhdlabel{\reccons{\ell}{v}{v'}}} \red \stateconf{H}{\ell}  }
    \\[1.5em]

    \inferrule*[lab={R-RecHdVal}]
    {\stateconf{H}{M} \red \stateconf{H'}{M'} }
    {\stateconf{H}{\tmhdterm{M}} \red \stateconf{H'}{\tmhdterm{M'}}  }
    \quad
    \inferrule*[lab={R-RecHdValV}]
    { }
    {\stateconf{H}{\tmhdterm{\reccons{\ell}{v}{v'}}} \red \stateconf{H}{v}  }
    \\[1.5em]

    \inferrule*[lab={R-RecTail}]
    {\stateconf{H}{M} \red \stateconf{H'}{M'} }
    {\stateconf{H}{\tmtl{M}} \red \stateconf{H'}{\tmtl{M'}}  }
    \quad
    \inferrule*[lab={R-RecTailV}]
    { }
    {\stateconf{H}{\tmtl{\reccons{\ell}{v}{v'}}} \red \stateconf{H}{v'}  }\\[1.5em]

    \inferrule*[lab={R-Ref}]{ \stateconf{H}{M} \red \stateconf{H'}{M'}}
    { \stateconf{H}{ \mkref{M}} \red \stateconf{H'}{ \mkref{M'} } } \quad
    \inferrule*[lab={R-RefV}]{ l \not\in \m{dom}(H) }{ \stateconf{H}{ \mkref{v}} \red \stateconf{H[l \mapsto v]}{l} }\\[1.5em]
    \inferrule*[lab={R-Deref}]{ \stateconf{H}{M}
    \red \stateconf{H'}{M'}  }{ \stateconf{H}{\deref{M}} \red \stateconf{H'}{\deref{M'}} } \quad
   
    \inferrule*[lab={R-DerefV}]{ H(l) = v }{ \stateconf{H}{\deref{l}} \red \stateconf{H}{v} } \\[1.5em]

    \inferrule*[lab={R-AssignL}]{ \stateconfH{M} \red \stateconf{H'}{M'} }{ \stateconfH{ M := N} \red \stateconf{H'}{M' := N} }
    \quad
    \inferrule*[lab={R-AssignR}]{ \stateconfH{M} \red \stateconf{H'}{M'} }{ \stateconfH{ l := M} \red \stateconf{H'}{l := M'} }  \end{array}\]

\[
  \begin{array}{c}
    \inferrule*[lab={R-AssignV}]{ \,}
    { \stateconf{H}{ l := v } \red \stateconf{H[l \mapsto v]}{v} } \quad
 \inferrule*[lab={R-PropT}]{ \Gamma  \models \varphi }{ \stateconfH{\ite{ \varphi }{M}{N} }  \red \stateconfH{M}} 
   \\[1.5em]
    \inferrule*[lab={R-PropF}]{ { \Gamma } \models \neg \varphi }{ \stateconfH{ \ite{ \varphi }{M}{N} }  \red \stateconfH{N}} \quad
   
    \inferrule*[lab={R-IfT}]{ }{ \stateconf{H}{ \ite{\true}{M}{N}} \red \stateconf{H}{M}} \\[1.5em]

  \end{array}
\]
\[
  \begin{array}{c}
      \inferrule*[lab={R-IfF}]{ }{ \stateconf{H}{\ite{\false}{M}{N}} \red \stateconf{H}{N}} \quad
    \inferrule*[lab={R-If}]{ \stateconf{H}{M} \red \stateconf{H'}{M'} }{ \stateconf{H}{\ite{M}{N_1}{N_2}} \red \stateconf{H'}{\ite{M'}{N_1}{N_2}}} \\[1.5em]

    \inferrule*[lab={R-TAppTRed}]{T\rightarrow T'}
    {\stateconfH{(\Lambda t {::} K. \, M) [T]} \red
    \stateconfH{(\Lambda t {::} K. \, M) [T']} }\\[1.5em]

    \inferrule*[lab={R-Fix}]{ \, }{ \stateconfH{\vrec{F}{T_v}{M}}
    \red \stateconf{H}{M\{\vrec{F}{T_v}{M}/{F}\}}} \quad

    \inferrule*[lab={R-TApp}]{ \, }{ \stateconfH{(\Lambda t {::} K. \, M) [T_v]}   \red  \stateconfH{M\{T_v/t\}} } \\[1.5em]
    \inferrule*[lab={R-TAppL}]
    { \stateconf{H}{M} \red \stateconf{H'}{M'} }
    { \stateconf{H}{M[T]} \red \stateconf{H'}{M'[T]}}
    \quad
    \inferrule*[lab={R-AppV}]{T\not\red }{ \stateconf{H}{(\lambda x{:} T. M) \, V} \red \stateconf{H}{M\{V/x\}} }
    \\[1.5em]

    \inferrule*[lab={R-AppL}]{ \stateconf{H}{M} \red \stateconf{H'}{M'}}
    { \stateconf{H}{M\, N} \red \stateconf{H'}{M'\, N}}\quad
     \inferrule*[lab={R-AppLT}]{ T\rightarrow T'}
    { \stateconf{H}{\lambda x{:}T.M} \red \stateconf{H'}{\lambda
    x{:}T'.M}}\\[1.5em]

    \inferrule*[lab={R-AppL}]{ \stateconf{H}{M} \red \stateconf{H'}{M'}}
    { \stateconf{H}{M\, N} \red \stateconf{H'}{M'\, N}}
    \quad
    \inferrule*[lab={R-AppR}]{T\not\red \quad \stateconf{H}{N} \red \stateconf{H'}{N'} }
    { \stateconf{H}{(\lambda x{:} T. M) \, N} \red \stateconf{H'}{ (\lambda x: T. M) \, N' } }\\[1.5em]

    \inferrule*[lab={R-ColConsL}]{ \stateconf{H}{M} \red \stateconf{H'}{M'} }{ \stateconf{H}{M :: N} \red \stateconf{H'}{M' :: N} } \quad
    \inferrule*[lab={R-ColConsR}]{ \stateconf{H}{N} \red \stateconf{H'}{N'} }{ \stateconf{H}{v :: N} \red \stateconf{H'}{v :: N'} }\\[1.5em]

    \inferrule*[lab={R-ColTlV}]
    {\stateconf{H}{M} \red \stateconf{H'}{M'}}
    {\stateconf{H}{\colcase{M}{N_1}{x {::} xs \Rightarrow
    N_2} \red \stateconf{H'}{\colcase{M'}{N_1}{x {::} xs \Rightarrow
    N_2}}}}\\[1.5em]

     \inferrule*[lab=({R-ColCaseEmp})]
    {\, }
    {\stateconf{H}{\colcase{\emptycol{}}{N_1}{x {::} xs \Rightarrow
    N_2}} \red
    \stateconf{H}{N_1}}
\\[1em]
       \inferrule*[lab=({R-ColCaseCons})]
    {\, }
    {\stateconf{H}{\colcase{v :: vs}{N_1}{x {::} xs \Rightarrow
    N_2}} \red
    \stateconf{H}{N_2\{v/x,vs/xs\}}}\\[1em]
   \end{array}
  \]

  \[
    \begin{array}{c}
      \inferrule*[lab={R-KindType}]
      {T\rightarrow T'}
      {\stateconfH{\tkindofP{T}{\smallkind}{M}{N} } \red
      \stateconfH{\tkindofP{T'}{\smallkind}{M}{N} }}

      \\[1.5em]
   \inferrule*[lab={R-KindL}]{\Ga \vdash T :: \smallkind }{
      \stateconf{H}{ \tkindofP{T}{\smallkind}{M}{N} }
        \red
        \stateconf{H}{M\{T/t\}}
      } \quad
   \inferrule*[lab={R-KindR}]{\Ga \vdash T :: \smallkind_0 \quad \Ga \vdash \smallkind_0 \not\equiv \smallkind }{
      \stateconf{H}{ \tkindofP{T}{\smallkind}{M}{N} }
        \red
        \stateconf{H}{N}
      } \\[1.5em]

    \end{array}
    \]


\section{Proofs}
\label{app:proofs}

\subst*
\begin{proof}
By induction on the derivation of the second given judgment.
We show some illustrative cases.

{\bf (a) }

\begin{description}

\item[Case:] $\inferrule*[right=(kref)]
    {\Ga, t{:}K,\Ga' \models \varphi\{S/s\} \quad \Ga , t{:}K,\Ga'  \vdash T :: \smallkind }
    {\Ga , t{:}K,\Ga' \vdash T :: \kref{t{::}\smallkind}{\varphi}}$
    
\begin{tabbing}
$\Ga , \Ga'\{T/t\} \models \varphi\{S/s\}\{T/t\}$ \` by i.h.\\
$\Ga , \Ga' \{T/t\}  \vdash S\{T/t\}  :: \smallkind\{T/t\} $ \` by
i.h.\\
$\Ga , \Ga'\{T/t\} \models \varphi \{T/t\}\{S \{T/t\}/s\}$ \` by
properties of substitution\\
$\Ga , \Ga' \{T/t\} \vdash S \{T/t\} :: \kref{s{::}\smallkind \{T/t\}
} { \varphi \{T/t\} }$ \` by rule
\end{tabbing}

\item[Case:]  
$\inferrule*[right=(entails)]
  {\Ga , t{:}K,\Ga' \vdash \varphi \quad \m{Valid}(\llb \Ga , t{:}K,\Ga'  \rrb \Rightarrow \llb \varphi\rrb)}
  {\Ga , t{:}K,\Ga' \models \varphi}$

\begin{tabbing}
$\Ga , \Ga'\vdash \varphi\{T/t\} $ \` by i.h.\\
$\m{Valid}(\llb \Ga , \Ga'\{T/t\}   \rrb \Rightarrow \llb
\varphi\{T/t\} \rrb)$ \` by logical substitution / congruence\\
$\Ga , \Ga' \{T/t\}\models \varphi\{T/t\}$ \` by rule
\end{tabbing}
  
\item[Case:]    $ \inferrule[]
    { \Ga,t{:}K,\Ga' \vdash \smallkind' \quad \Ga,t{:}K,\Ga' , s{:}\smallkind' \vdash \varphi}
    {\Ga,t{:}K,\Ga' \vdash \kref{s : \smallkind'}{\varphi}}$

    \begin{tabbing}
      $\Ga,\Ga'\{T/t\} \vdash \smallkind'\{T/t\}$ \` by i.h.\\
      $\Ga,\Ga'\{T/t\} , s{:}\smallkind'\{T/t\} \vdash \varphi\{T/t\}$ \` by i.h.\\
      $\Ga , \Ga'\{T/t\} \vdash \kref{s : \smallkind'\{T/t\}}{\varphi\{T/t\}}$ \` by rule
    \end{tabbing}

\item[Case:]  $\inferrule*[right=(r-eqelim)]
   {\Ga,t:K,\Ga' \vdash T' :: \kref{s{::}\smallkind}{s = S} \quad \Ga , t:K,\Ga' \vdash S ::
    \smallkind}
    {\Ga , t:K,\Ga' \vdash T' \equiv S :: \smallkind}$

\begin{tabbing}
$\Ga,\Ga'\{T/t\} \vdash T'\{T/t\} :: \kref{s{::}\smallkind\{T/t\}}{s =
  S\{T/t\}} $ \` by i.h.\\
$\Ga , \Ga'\{T/t\} \vdash S\{T/t\} :: \smallkind\{T/t\}$ \` by i.h.\\
$\Ga , \Ga'\{T/t\} \vdash T'\{T/t\} \equiv S\{T/t\} ::
\smallkind\{T/t\}$ \` by rule
\end{tabbing}






    

\item[Case:]   $\inferrule[]
    {\Ga,t{:}K,\Ga' \vdash K' \quad \Ga,t{:}K,\Ga' , s{:}K' \vdash T' :: \smallkind}
    {\Ga,t{:}K,\Ga' \vdash \forall s {:} K'. T' :: \kpolyfun_{K'}}$

    \begin{tabbing}
      $\Ga,\Ga'\{T/t\} \vdash K'\{T/t\}$ \` by i.h.\\
      $\Ga,\Ga'\{T/t\}, s{:}K'\{T/t\} \vdash T'\{T/t\} :: \smallkind$ \` by i.h.\\
      $\Ga,t{:}K,\Ga' \vdash \forall s {:} K'\{T/t\}.
      T'\{T/t\} :: \kpolyfun_{K'\{T/t\}}$ \` by rule
    \end{tabbing}

\item[Case:] $\inferrule[]
    {\Ga,t{:}K,\Ga' \vdash L :: \kname \quad
      \Ga,t{:}K,\Ga'
      \vdash T' :: \smallkind \quad
      \Ga,t{:}K,\Ga' \vdash S' :: \kref{t : \krecord}{L \# t }}
    {\Ga,t{:}K,\Ga' \vdash \treccons{L}{T'}{S'} ::
      \krecord }$

    \begin{tabbing}
      $\Ga,\Ga'\{T/t\} \vdash L\{T/t\} :: \kname$ \` by i.h.\\
      $\Ga,\Ga'\{T/t\} \vdash T'\{T/t\} :: \smallkind$ \` by i.h.\\
      $\Ga,\Ga'\{T/t\} \vdash S'\{T/t\} :: \kref{t : \krecord}{L\{T/t\} \# t }$ \` by i.h.\\
      $\Ga,\Ga'\{T/t\} \vdash \treccons{L\{T/t\}}{T'\{T/t\}}{S'\{T/t\}} :: \krecord$
      
      \` by rule
    \end{tabbing}




\item[Case:] $\inferrule[]
  {\Ga,t{:}K,\Ga' \vdash \varphi \quad
    \Ga,t{:}K,\Ga' , \varphi \vdash T' :: K' \quad
    \Ga,t{:}K,\Ga' , \neg\varphi \vdash S :: K'}
    {\Ga,t{:}K,\Ga' \vdash \ite{\varphi}{T'}{S} :: K'    }$

    \begin{tabbing}
      $\Ga,\Ga'\{T/t\} \vdash \varphi\{T/t\}$ \` by i.h.\\
      $\Ga,\Ga'\{T/t\} , \varphi\{T/t\} \vdash T'\{T/t\} :: K'\{T/t\}$ \` by i.h.\\
      $\Ga,\Ga'\{T/t\} , \neg\varphi\{T/t\} \vdash S\{T/t\} :: K'\{T/t\}$ \` by i.h.\\
      $\Ga,\Ga'\{T/t\} \vdash \ite{\varphi\{T/t\}}{T'\{T/t\}}{S\{T/t\}} :: K'\{T/t\}    $ \` by rule
    \end{tabbing}

  \item[Case:] $\inferrule[]
    {\Ga,t{:}K,\Ga' \vdash S :: \kref{t : \krecord}{ \ell \not\in t} \quad
      \Ga,t{:}K,\Ga' \vdash M : T' \quad
    \Ga,t{:}K,\Ga' \vdash N : S}
    {\Ga,t{:}K,\Ga'  \vdash \reccons{\ell}{M}{N} : \treccons{\ell}{T'}{S}}$

    \begin{tabbing}
      $\Ga,\Ga'\{T/t\} \vdash S :: \kref{t : \krecord}{ \ell \not\in t}$ \` by i.h.\\
      $\Ga,\Ga'\{T/t\} \vdash M\{T/t\} : T'\{T/t\}$ \` by i.h.\\
      $\Ga,\Ga'\{T/t\} \vdash N\{T/t\} : S\{T/t\}$ \` by i.h.\\
      $\Ga,\Ga'\{T/t\}  \vdash
      \reccons{\ell}{M\{T/t\}}{N\{T/t\}} : \treccons{\ell}{T'\{T/t\}}{S\{T/t\}}$ \` by rule
    \end{tabbing}

 \item[Case:] $    \inferrule[]
    {\Ga,t{:}K,\Ga' \vdash M : \treccons{L}{S}{U} }
    {\Ga,t{:}K,\Ga' \vdash \tmhdlabel{M} : L }$

    \begin{tabbing}
      $\Ga,\Ga'\{T/t\} \vdash M\{T/t\} : \treccons{L\{T/t\} }{S\{T/t\} }{U\{T/t\} } $ \` by i.h.\\
           $\Ga,\Ga'\{T/t\} \vdash \tmhdlabel{M\{T/t\}} : L\{T/t\}$ \` by rule
    \end{tabbing}

\item[Case:] $\inferrule[]
    {\Ga , t{:}K,\Ga' \vdash M : \treccons{L}{S}{U}  }
    {\Ga,t{:}K,\Ga'  \vdash \tmhdterm{M} : S}$
    
    \begin{tabbing}
      $\Ga,\Ga'\{T/t\}  \vdash M\{T/t\} : \treccons{L \{T/t\}}{S \{T/t\}}{U \{T/t\}}$  \` by i.h.\\
      $\Ga,\Ga'\{T/t\}  \vdash \tmhdterm{M\{T/t\}} : S \{T/t\}$ \` by rule
    \end{tabbing}

\item[Case:]
    $\inferrule[]
    {\begin{array}{c}
       \Ga,t'{:}K,\Ga' ,t{:}K_1 \vdash K_2 \quad \Ga,t'{:}K,\Ga' \vdash S :: K_1\\
    \Ga,t'{:}K,\Ga' , F{:}\Pi t{:}K_1.K_2,s{:}K_1 \vdash T' :: K_2
        \quad \m{structural}(T',F,t)
       \end{array}}
    {\Ga,t'{:}K,\Ga' \models (\tfix{F}{K_1}{K_2}{t}T')\,S \equiv
      T'\{S/t\}\{(\tfix{F}{K_1}{K_2}{t}T')/F\} :: K_2\{S/t\}}$    

    \begin{tabbing}
      $\Ga,\Ga'\{T/t'\} ,t{:}K_1\{T/t'\} \vdash K_2\{T/t'\} $ \` by i.h.\\
      $\Ga,\Ga'\{T/t'\} \vdash S\{T/t'\} :: K_1\{T/t'\}$ \` by i.h.\\
      $\Ga,\Ga'\{T/t'\} , F{:}\Pi t{:}K_1\{T/t'\}.K_2\{T/t'\},s{:}K_1\{T/t'\} \vdash T'\{T/t'\} :: K_2\{T/t'\}$ \` by i.h.\\
      $\m{structural}(T'\{T/t'\},F,t)$ \` \\
      $\Ga,\Ga'\{T/t'\} \models (\tfix{F}{K_1\{T/t'\}}{K_2\{T/t'\}}{t}T'\{T/t'\})\,S\{T/t'\} \equiv$\\
     \qquad$
     T'\{T/t'\}\{S\{T/t'\}/t\}\{(\tfix{F}{K_1\{T/t'\}}{K_2\{T/t'\}}{t}T'\{T/t'\})/F\}$\\
     \qquad\qquad$:: K_2\{T/t'\}\{S\{T/t'\}/t\}$ \` by rule
    \end{tabbing}
    
\end{description}
The remaining cases follow by similar reasoning, relying on type- and kind-preserving substitution in the language of refinements.
\end{proof}

\ctxtconv*
\begin{proof}
  Follows by weakening and substitution.

  {\bf(a)}
  \begin{tabbing}
    $\Ga , x:T' \vdash x:T'$ \` by variable rule\\
    $\Ga \vdash T' \equiv T :: K$ \` by symmetry\\
    $\Ga ,x{:}T'\vdash x:T$ \` by conversion\\
    $\Ga,x':T \vdash \mathcal{J}\{x'/x\}$ \` alpha conversion, for fresh $x'$\\
    $\Ga,x:T',x'{:}T\vdash \mathcal{J}\{x'/x\}$ \` by  weakening\\
    $\Ga,x'{:}T\vdash \mathcal{J}\{x'/x\}\{x/x'\}$ \` by substitution\\
    $\Ga , x{:}T'\vdash \mathcal{J}$ by definition
  \end{tabbing}
Statement {\bf (b)} follows by the same reasoning.
\end{proof}

\begin{restatable}[Functionality of Kinding and Refinements]{lemma}{funckinding}
  \label{lem:funckind}~
  
  Assume $\Ga \models T \equiv S :: K$, $\Ga \vdash T :: K$ and $\Ga
  \vdash S :: K$:

  \begin{enumerate}
  \item[(a)] If $\Ga , t{:}K ,\Ga' \vdash T' :: K'$ then
    $\Ga , \Ga'\{T/t\} \models T'\{T/t\} \equiv T'\{S/t\} :: K'\{T/t\}$
  \item[(b)] If $\Ga ,   t{:}K ,\Ga' \vdash K'$ then
    $\Ga , \Ga'\{T/t\} \vdash K\{T/t\} \equiv K\{S/t\}$.
  \item[(c)] If $\Ga , t{:}K,\Ga' \models \varphi$ then
    $\Ga , \Ga'\{T/t\} \models \varphi\{T/t\} \Leftrightarrow \varphi \{S/t\}$
  \end{enumerate}

\end{restatable}

\begin{proof}
  By induction on the given kinding/kind well-formedness and entailment judgments.
  Functionality follows by substitution and the congruence rules of
  definitional equality.

  \begin{description}
  \item[Case:] $\inferrule[]
    { \Ga,t:K,\Ga' \vdash \smallkind' \quad \Ga,t:K,\Ga' , t'{:}\smallkind' \vdash \varphi}
    {\Ga,t:K,\Ga' \vdash \kref{t' : \smallkind'}{\varphi}}$

    \begin{tabbing}
      $\Ga,\Ga'\{T/t\} \vdash \smallkind'\{T/t\}\equiv \smallkind'\{S/t\}$ \` by i.h.\\
      $\Ga,\Ga'\{T/t\} , t'{:}\smallkind'\{T/t\} \vdash \varphi\{T/t\} \equiv \varphi\{S/t\}$ \` by i.h.\\
      $\Ga,\Ga'\{T/t\} \vdash \kref{t' : \smallkind'\{T/t\}}{\varphi\{T/t\}}
      \equiv\kref{t' : \smallkind'\{S/t\}}{\varphi\{S/t\}}$ \` by kind ref. equality
    \end{tabbing}

  \item[Case:] $\inferrule[]
    {\Ga,t:K,\Ga' \vdash K \quad \Ga,t:K,\Ga , s{:}K' \vdash T' :: \smallkind}
    {\Ga,t:K,\Ga \vdash \forall s {:} K'. T' :: \kpolyfun_K}$

    \begin{tabbing}
      $\Ga,\Ga'\{T/t\} \vdash K'\{T/t\} \equiv K'\{S/t\}$ \` by i.h.\\
      $\Ga,\Ga'\{T/t\} , t'{:}K'\{T/t\} \vdash T'\{T/t\} \equiv T'\{S/t\} :: \smallkind$ \` by i.h.\\
      $\Ga , \Ga'\{T/t\} \vdash \forall s:K'\{T/t\}.T'\{T/t\} \equiv 
         \forall s:K'\{S/t\}.T'\{S/t\} :: \kpolyfun_{K'\{T/t\}}$ \` by $\forall$ Eq.
       \end{tabbing}

  \item[Case:] $\inferrule[]
    {\Ga,t:K,\Ga' \vdash L :: \kname \quad
      \Ga,t:K,\Ga' \vdash T' :: \smallkind \quad
      \Ga,t:K,\Ga' \vdash S' :: \kref{t : \krecord}{L \not\in t }}
    {\Ga,t:K,\Ga' \vdash \treccons{L}{T'}{S'} ::
      \krecord }$

    \begin{tabbing}
      $\Ga , \Ga'\{T/t\} \vdash L\{T/t\} \equiv L\{S/t\} :: \kname$ \` by i.h.\\
      $\Ga,\Ga'\{T/t\} \vdash T'\{T/t\} \equiv T'\{S/t\} :: \smallkind$ \` by i.h.\\
      $\Ga,\Ga'\{T/t\} \vdash S'\{T/t\} \equiv S'\{S/t\} ::
      \kref{t : \krecord}{L\{T/t\} \not\in t }$ \` by i.h.\\
      $\Ga , \Ga'\{T/t\} \vdash \treccons{L\{T/t\}}{T'\{T/t\}}{S'\{T/t\}} \equiv
      \treccons{L\{S/t\}}{T'\{S/t\}}{S'\{S/t\}} :: \krecord$ \` by Rec Eq.\\
    \end{tabbing}

  \item[Case:]$\inferrule*[right=(k-dom)]
    {\Ga ,t:K,\Ga'\vdash T' :: \kfunction}
    {\Ga , t:K,\Ga'\vdash \dom{T'} ::\ktype}$

    \begin{tabbing}
$\Ga , \Ga'\{T/t\} \vdash T'\{T/t\} \equiv T'\{S/t\} :: \kfunction$ \`
by i.h.\\
$\Ga , \Ga'\{T/t\} \vdash \dom{T\{T/t\}}\equiv \dom{T'\{S/t\}} ::
\ktype$ \` by congruence rule 
     \end{tabbing}

 \item[Case:]
     $\inferrule[(k-reccons)]
    {\Ga t:K,\Ga'\vdash L :: \kname \quad
    \Ga t:K,\Ga'\vdash U :: \smallkind \quad \Ga t:K,\Ga'\vdash W :: \kref{s : \krecord}{L \not\in \lb{s} }}
    {\Ga t:K,\Ga'\vdash \treccons{L}{U}{W} ::
      \krecord }$
   
    \begin{tabbing}
$\Ga , \Ga'\{T/t\} \vdash L\{T/t\} \equiv L\{S/t\} :: \kname$ \` by
i.h.\\
$\Ga,\Ga'\{T/t\} \vdash U \{T/t\}  \equiv U\{S/t\} :: \smallkind
\{T/t\}$ \` by i.h.\\
$\Ga ,\Ga'\{T/t\}\vdash W\{T/t\} \equiv W\{S/t\} :: \kref{s :
  \krecord}{L\{T/t\} \not\in \lb{s} }$ \` by i.h.\\
$\Ga , \Ga'\{T/t\} \vdash \treccons{L\{T/t\}}{U\{T/t\}}{W\{T/t\}}
\equiv \treccons{L\{S/t\}}{U\{S/t\}}{W\{S/t\}} :: \krecord$ \\\` by
congruence rule
   \end{tabbing}


\item[Case:]
  $\inferrule[]
  {\Ga,t{:}K,\Ga' \vdash \varphi \quad \Ga,t{:}K,\Ga' , \varphi \vdash T' :: K' \quad
    \Ga,t{:}K,\Ga' , \neg\varphi \vdash S' :: K'}
    {\Ga,t{:}K,\Ga' \vdash \ite{\varphi}{T'}{S'} :: K'    }$

    \begin{tabbing}
      $\Ga,\Ga'\{T/t\}, \varphi\{T/t\}
      \vdash T'\{T/t\} \equiv T'\{S/t\} :: K'\{T/t\}$ \` by i.h.\\
      $\Ga,\Ga'\{T/t\} , \neg\varphi\{T/t\} \vdash S'\{T/t\} \equiv S'\{S/t\} :: K'\{T/t\}$
      \` by i.h.\\

      $\Ga , t:K , \Ga' , \varphi \models \varphi$ \` tautology\\
      $\Ga , \Ga'\{T/t\} , \varphi\{T/t\} \models \varphi\{T/t\}$ \` by substitution\\
      $\Ga , \Ga'\{T/t\} , \varphi\{S/t\} \models \varphi\{T/t\}$ \` by ctxt. conversion\\
      $\Ga , \Ga'\{T/t\} \models\varphi\{S/t\} \supset\varphi\{T/t\}$ \` by $\supset$I\\
      $\Ga , \Ga'\{S/t\} , \varphi\{S/t\} \models \varphi\{S/t\}$ \` by substitution\\
      $\Ga , \Ga'\{T/t\} , \varphi\{T/t\} \models \varphi\{S/t\}$ \` by ctxt. conversion\\
      $\Ga , \Ga'\{T/t\} \models\varphi\{T/t\} \supset\varphi\{S/t\}$ \` by $\supset$I\\
      $\Ga,\Ga'\{T/t\} \vdash \varphi\{T/t\} \equiv \varphi\{S/t\}$ \` by definition\\
      $\Ga , \Ga'\{T/t\} \models\ite{\varphi\{T/t\}}{T'\{T/t\}}{S'\{T/t\}} \equiv$\\
      \qquad\qquad\qquad $\ite{\varphi\{S/t\}}{T'\{S/t\}}{S'\{S/t\}} :: K'\{T/t\}$
      \` by rule

    \end{tabbing}

  \item[Case:] $\inferrule[]
    {\Ga,t{:}K,\Ga' \models \varphi\{T'/s\} \quad
    \Ga,t{:}K,\Ga' \vdash T' :: \smallkind' }
    {\Ga,t{:}K,\Ga' \vdash T' :: \kref{s{:}\smallkind'}{\varphi}}$

    \begin{tabbing}
      $\Ga,\Ga'\{T/t\} \models \varphi\{T'/s\}\{T/t\} \equiv \varphi\{T'/s\}\{S/t\}$ \` by i.h.\\
      $\Ga,\Ga'\{T/t\} \models T'\{T/t\}\equiv T'\{S/t\} :: \smallkind'\{T/t\}$ \` by i.h.\\

      $\Ga , \Ga'\{T/t\} \models T'\{T/t\} \equiv T'\{S/t\} :: \kref{s{:}\smallkind'\{T/t\}}{\varphi\{T/t\}}$ \` by Eq Conversion
    \end{tabbing}
    
  \end{description}

\end{proof}

\valideq*

\begin{proof}
  By induction on the given derivation.

  \begin{description}
\item[Case:]  $\inferrule[]
    { \Ga \vdash \smallkind \leq \smallkind' \quad
     \Ga , t{:}\smallkind \vdash \varphi \Rightarrow \psi}
   { \Ga \vdash \kref{t{:}\smallkind}{\varphi} \leq
     \kref{t{:}\smallkind'}{\psi}}$

   \begin{tabbing}
     $\Ga \vdash \smallkind$ and $\Ga \vdash \smallkind'$ \` by i.h.\\
     $\Ga ,  t{:}\smallkind \vdash \varphi$ \` by inversion\\
    $\Ga , t{:}\smallkind \vdash \psi$ \` by inversion\\
    $\Ga , t{:}\smallkind' \vdash \psi$ \` by context conversion\\
     $\Ga \vdash \kref{t{:}\smallkind}{\varphi}$ \` by refinement kind w.f.\\
     $\Ga \vdash \kref{t{:}\smallkind'}{\psi}$ \` by refinement kind w.f.
   \end{tabbing}

   \item[Case:]$\inferrule*[right=(r-eqelim)]
   {\Ga \vdash T :: \kref{s{::}\smallkind}{s = S} \quad \Ga \vdash S ::
    \smallkind}
    {\Ga  \vdash T \equiv S :: \smallkind}$

    \begin{tabbing}
      $\Ga \vdash S :: \smallkind$ \` by inversion\\
      $\Ga \vdash T :: \kref{s{::}\smallkind}{s = S}$ \` by
      inversion\\
      $\Ga \vdash T :: \smallkind$ \` by subsumption\\
      $\Ga \vdash \smallkind$ \` by kind w.f.
      \end{tabbing}

 \item[Case:] $\inferrule[]
  {\Ga \vdash L :: \kname \quad
  \Ga \vdash T ::\smallkind \quad
  \Ga \vdash S :: \kref{t {::} \krecord}{L\not\in \lb{t}}}
  {\Ga \vdash \hdtype{\treccons{L}{T}{S}} \equiv T :: \ktype}$

  \begin{tabbing}
    $\Ga \vdash T :: \smallkind$  \` by inversion\\
    $\Ga \vdash T :: \ktype$ \` by subsumption\\
    $\Ga \vdash \treccons{L}{T}{S} :: \krecord$ \` by kinding\\
    $\Ga \vdash \treccons{L}{T}{S} :: \kref{t {::}
      \krecord}{\neg\m{empty}(t)}$ \` by subsumption\\
    $\Ga \vdash \hdtype{\treccons{L}{T}{S}} :: \ktype$ \` by kinding\\
    $\Ga \vdash \ktype$ \` by kind w.f.
   \end{tabbing}
\item[Case:]
  $\inferrule[]
  {\Ga \vdash T :: \smallkind}
  {\Ga \models \colof{\tcol{T}} \equiv T :: \ktype }$

  \begin{tabbing}
    $\Ga \vdash T :: \smallkind$ \` by inversion\\
    $\Ga \vdash T :: \ktype$ \` by subkinding\\
    $\Ga \vdash \tcol{T} :: \kcollection$ \` by kinding\\
    $\Ga \vdash \colof{\tcol{T}} :: \ktype$ \` by kinding\\
        $\Ga \vdash \ktype$ \` by kind w.f.
  \end{tabbing}

  
Remaining cases follow by a similar reasoning.

  \end{description}
\end{proof}

\kindpres*
\begin{proof}
Immediate from equality validity since $T \rightarrow S$ implies
$T\equiv S$.
\end{proof}

\begin{restatable}[Functionality of Equality]{lemma}{funceq}
\label{lem:funceq}  ~
  Assume $\Ga \models T_0 \equiv S_0 :: K$:
  \begin{enumerate}
  \item[(a)] If \,$\Ga , t{:}K \models T \equiv S :: K'$ then
    $\Ga \models T\{T_0/t\} \equiv S\{S_0/t\} :: K'\{T_0/t\}$.
  \item[(b)] If $\Ga , t{:}K \vdash K_1 \equiv K_2$ then
    $\Ga \vdash K_1\{T_0/t\} \equiv K_2\{S_0/t\}$.
  \item[(c)] If $\Ga , t{:}K \vdash \varphi \Leftrightarrow \psi$ then
    $\Ga \models\varphi\{T_0/t\} \Leftrightarrow \psi\{S_0/t\}$.
  \end{enumerate}

\end{restatable}
\begin{proof}

  \begin{description}
  \item[(a)]
    ~
    \begin{tabbing}
  $\Ga , t{:}K\models T \equiv S :: K'$ \` assumption\\
  $\Ga \vdash T_0\equiv S_0 :: K$ \` assumption\\
  $\Ga \vdash T_0 :: K$ and $\Ga \vdash S_0 :: K$ \` by eq. validity\\
  $\Ga , t{:}K \vdash T :: K'$ and $\Ga , t{:}K \vdash S :: K'$ \` by eq. validity\\
  $\Ga \vdash T\{T_0/t\} \equiv S\{T_0/t\} :: K'\{T_0/t\}$ \` by substitution\\
  $\Ga \vdash S\{T_0/t\} \equiv S\{S_0/t\} :: K'\{T_0/t\}$ \` by functionality\\
  $\Ga \vdash T\{T_0/t\}\equiv S\{S_0/t\} :: K'\{T_0/t\}$ \` by transitivity
\end{tabbing}

\item[(b)]
  ~
  \begin{tabbing}
  $\Ga \vdash T_0\equiv S_0 :: K$ \` assumption\\
  $\Ga , t:K \vdash K_1 \equiv K_2$ \` assumption\\
  $\Ga \vdash T_0 :: K$ and $\Ga \vdash S_0 :: K$ \` by eq. validity\\
  $\Ga ,t:K \vdash K_1$ and $\Ga , t:K \vdash K_2$ \` by eq. validity\\
  $\Ga \vdash K_1\{T_0/t\} \equiv K_2\{T_0/t\}$ \` by substitution\\
  $\Ga \vdash K_2\{T_0/t\} \equiv K_2\{S_0/t\}$ \` by functionality\\
  $\Ga \vdash K_1\{T_0/t\} \equiv K_2\{S_0/t\}$ \` by transitivity
  \end{tabbing}

\item[(c)]
  ~
  \begin{tabbing}
$\Ga \vdash T_0\equiv S_0 :: K$ \` assumption\\
$\Ga , t:K \vdash \varphi \equiv \psi$ \` assumption\\
$\Ga \vdash T_0 :: K$ and $\Ga \vdash S_0 :: K$ \` by eq. validity\\
$\Ga , t:K \vdash \varphi$ and $\Ga , t:K \vdash \psi$ \` by eq. validity\\
$\Ga \vdash \varphi\{T_0/t\} \equiv\psi\{T_0/t\}$ \` by substitution\\
$\Ga \vdash \psi\{T_0/t\} \equiv \psi\{S_0/t\}$ \` by functionality\\
$\Ga \vdash \varphi\{T_0/t\} \equiv \psi\{S_0/t\}$ \` by transitivity
  \end{tabbing}
  
\end{description}
\end{proof}

\begin{restatable}[Validity]{theorem}{valid}
\label{lem:valid}~
  \begin{enumerate}
  \item[(a)] If $\Ga \vdash K$ then $\Ga \vdash$
  \item[(b)] If $\Ga \vdash T :: K$ then $\Ga \vdash K$
  \item[(c)] If $\Ga \vdash M : T$ then $\Ga \vdash T :: \ktype$.
  \end{enumerate}
\end{restatable}

\begin{proof}
Straightforward induction on the given derivation.
\end{proof}

\begin{lemma}[Injectivity]

  If $\Ga \vdash \Pi t:K_1.K_2 \equiv \Pi t:K_1'.K_2'$ then
   $\Ga \vdash K_1 \equiv K_1'$ and $\Ga , t:K_1 \vdash K_2 \equiv K_2'$.

 \end{lemma}
 \begin{proof}
Straightforward induction on the given kind equality derivation.
 \end{proof}

 \begin{lemma}[Injectivity via Subkinding]
If $\Ga \vdash \Pi t{:}K_1.K_2 \leq K$ then
     $\Ga \vdash K \equiv \Pi t:K_1'.K_2'$ with
     $\Ga \vdash K_1 \equiv K_1'$ and $\Ga , t:K_1 \vdash K_2 \equiv K_2'$.

 \end{lemma}

\begin{restatable}[Inversion]{lemma}{invers}
  \label{lem:invers}~
  \begin{enumerate}
  \item[(a)] If $\Ga \vdash \lambda t{::}K.T :: K'$ then there is
    $K_1$ and $K_2$ such that $\Ga \vdash K' \equiv \Pi t{:}K_1.K_2$,
    $\Ga \vdash K \equiv K_1$ and $\Ga , t{:}K_1 \vdash T :: K_2$.
  \item[(b)] If $\Ga \vdash T\,S :: K$ then
    $\Ga \vdash T :: \Pi t{:}K_0.K_1$, $\Ga \vdash S :: K_0$ and
    $\Ga \vdash K \equiv K_1\{S/t\}$.
  \item[(c)] If $\Ga \vdash \lambda x{:}T.M : T'$ then there is
    $T_1$ and $T_2$ such that $\Ga \models T' \equiv T_1\rightarrow T_2 :: \kfunction$,
    $\Ga \models T \equiv T_1 :: \ktype$ and $\Ga , x{:}T_1 \vdash M : T_2$.

  \item[(d)] If  $\Ga \vdash \treccons{L}{T}{S} :: K$
    then $\Ga \vdash L :: \kname$, $\Ga \vdash T :: \ktype$,
    $\Ga \vdash S :: \kref{t{::}\krecord}{L\not\in t}$ and $\Ga \vdash K \equiv \krecord$.

  \item[(e)] If $\Ga \vdash \reccons{L}{M}{N} : T$ then there is $L',T_1,T_2$ such that $\Ga \models L \equiv L' :: \kname$,
    $\Ga \vdash \treccons{L'}{T_1}{T_2} :: \krecord$,
    $\Ga \models T \equiv \treccons{L'}{T_1}{T_2} :: \krecord$,
    $\Ga \vdash M : T_1$ and $\Ga\vdash N : T_2$.

     \item[(f)] If $\Ga \vdash T :: \kref{t{::}K}{\varphi}$ then
    $\Ga \models \varphi\{T/t\}$, $\Ga \vdash T :: K$ and
    $\Ga , t{:}K \vdash \varphi$.


  \item[(h)] If $\Ga \vdash \ite{\varphi}{M}{N} : T$ then
    $\Ga \models T \equiv \ite{\varphi}{T_1}{T_2} :: K$ with
    $\Ga , \varphi \vdash M : T_1$ and $\Ga , \neg\varphi \vdash N : T_2$.

  \item[(i)] If $\Ga \vdash \ite{\varphi}{T}{S} :: K$ then
    $\Ga \vdash \varphi$, $\Ga,\varphi \vdash T :: K$ and $\Ga,\neg\varphi \vdash S :: K$.

  \item[(j)] If $\Ga \vdash T \rightarrow S :: K$ then $\Ga \vdash K \equiv \kfunction$, $\Ga \vdash T :: \smallkind$ and $\Ga \vdash S :: \smallkind'$, for some
    $\smallkind,\smallkind'$.

  \item[(k)] If $\Ga \vdash M :: N :: T$ then
    $\Ga \models T \equiv \tcol{S} :: \kcollection$,
    $\Ga \vdash N : \tcol{S}$ and $\Ga \vdash M : S$, for some $S$.

  \item[(l)] If $\Ga \vdash \tcol{T'} :: K$ then
    $\Ga \vdash K \equiv \kcollection$ and $\Ga \vdash T' :: \smallkind$, for some $\smallkind$.

  \item[(m)] If $\Ga \vdash \tkindofP{T'}{K}{M}{N} : T$
    then $\Ga \vdash T' :: \smallkind$, $\Ga \vdash K$, $\Ga , t:K\vdash M : S$ and
    $\Ga \vdash N : S$, with $\Ga \vdash T \equiv S :: \smallkind'$ for some
    $\smallkind, \smallkind', S$.

  \item[(n)] If $\Ga \vdash \tkindofP{T'}{K}{S}{S'} :: K'$
    then $\Ga \vdash T' :: \smallkind$, $\Ga \vdash K$, $\Ga ,t{:}K \vdash S :: K''$, $\Ga\ vdash S' :: K''$ and $\Ga \vdash K' \equiv K''$, for some $\smallkind, K''$.

  \item[(o)] If $\Ga \vdash \vrec{F}{T}{M} : T$ then
    $\Ga ,F:T \vdash M : T$ and $\m{structural}(F,M)$.

  \item[(p)] If $\Ga \vdash  \tfix{F}{K_1}{K_2}{t}T' :: K$
    then $\Ga ,  F{:}\Pi t{:}K_1.K_2, t{:}K_1 \vdash T' :: K_2$,
    $\m{structural}(T',F,t)$ and $\Ga \vdash K \equiv \Pi t{:}K_1.K_2$.

  \item[(q)] If $\Ga \vdash \tmhdlabel{M} : T$ then
    $\Ga \models T \equiv L :: \kname$,
    $\Ga \vdash M : \treccons{L}{S}{U}$.

  \item[(r)] If $\Ga \vdash \tmhdterm{M} : T$ then
    $\Ga \models T \equiv S :: \ktype$ and $\Ga \vdash M : \treccons{L}{S}{U}$

  \item[(s)] If $\Ga \vdash \tl{M} : T$ then
   $\Ga \models T \equiv U :: \krecord$ and $\Ga \vdash M :
   \treccons{L}{S}{U}$.
   
  \item[(t)] If $\Ga \vdash \tmcolhd{M} : T$ then
    $\Ga \vdash M : \tcol{T}$
  \item[(u)] If $\Ga \vdash \tmcoltl{M} : T$ then
   $\Ga \vdash M : \tcol{T}$
  \item[(v)] If $\Ga \vdash \mkref{M} : T$ then
    $\Ga \models T \equiv \tref{T'}$ and $\Ga \vdash M : T'$
  \item[(w)] If $\Ga \vdash \deref{M} : T$ then
    $\Ga \models T \equiv \tref{T'} ::\kreference$,
    $\Ga \vdash M : \tref{T'}$, for some $T'$.
  \item[(x)] If $\Ga \vdash M := N : T$ then
    $\Ga \models T \equiv \one :: \ktype$,
    $\Ga \vdash M : \tref{T'}$, $\Ga \vdash N : T'$, for some $T'$.

  \item[(y)] If $\Ga \vdash M\,N : T$ then
    $\Ga \vdash M : T_1\rightarrow T_2$, $\Ga \vdash N : T_1$,
    $\Ga \vdash T \equiv T_2 :: \ktype$
  
  \item[(z)] If $\Ga \vdash M[T] : S$ then
    $\Ga \vdash M : \forall t :: K.T'$, $\Ga \vdash T :: K$,
    $\Ga \vdash T' :: \kref{f{::}\kpolyfun_K}{\inst{f}{\,T} \equiv U :: \smallkind}$
    and $\Ga \vdash S \equiv T'\{T/t\}:: \ktype$.
    
  \end{enumerate}

\end{restatable} 

\begin{proof}
  By induction on the structure of the given typing or kinding derivation, using
  validity.
  {\bf (a)}
  \begin{description}
\item[Case:] $\inferrule[]
    {\Ga \vdash \lambda t:K.T :: K'' \quad \Ga \vdash K'' \leq K'}
    {\Ga \vdash \lambda t:K.T :: K'}$
    
    \begin{tabbing}
      $\Ga \vdash K'' \equiv \Pi t{:}K_1'.K_2'$,
      $\Ga \vdash K \equiv K_1'$ and $\Ga,t:K_1' \vdash T :: K_2'$ \` by i.h.\\
      $\Ga \vdash K'' \leq \Pi t:K_1.K_2$, for some $K_1$, $K_2$ with
      $\Ga \vdash K_1' \leq K_1$ and $\Ga,t:K_1' \vdash K_2' \leq K_2$
      \\\` by inversion\\
      $\Ga \vdash K_1' \equiv K_1$ and $\Ga , t:K_1' \vdash K_2' \equiv K_2$ \` by inversion\\
      $\Ga  ,t:K_1 \vdash T :: K_2'$ \` by ctxt. conversion\\
      $\Ga  , t:K_1 \vdash T :: K_2$ \` by conversion\\
      $\Ga \vdash K \equiv K_1$ \` by transitivity
    \end{tabbing}
  \end{description}

Other cases follow by similar reasoning (or are immediate).
  
\end{proof}

Below we do not list the (very) extensive list of all inversions. They follow the same
pattern of the kinding inversion principle.

\begin{lemma}[Equality Inversion]
\label{lem:eqinv}~
  \begin{enumerate}
  \item If $\Ga \models T \equiv \lambda t:K_1.T_2 :: K'$ then
    $\Ga \models T \equiv \lambda t:K_0.T_2' :: \Pi t:K_0.K''$ with
    $\Ga \vdash K_0 \equiv K_1$ and $\Ga , t:K_0 \models T_2 \equiv T_2' :: K''$, for some $K''$.

  \item If $\Ga \models T \equiv T_0\,S_0 :: K$ then 
 $\Ga \models T \equiv T_1\,S_1 :: K$ with $\Ga \models
T_1 \equiv T_0 :: \Pi t:K_1.K_0$,$\Ga S_1 \equiv S_0 :: K_1$
and $K = K_0\{S_1/t\}$.

  \item If $\Ga \models T \equiv \treccons{L}{T}{S} :: K$ then
$\Ga \models T \equiv \treccons{L'}{T'}{S'} :: K$ with
$\Ga \models L \equiv L' :: \kname$, 
$\Ga \models T' \equiv T :: \smallkind$, 
$\Ga \models S' \equiv S :: \kref{t:\krecord}{L\not\in t}$ and
$K = \krecord$.

\item If $\Ga \vdash K \leq \kref{t:\smallkind}{\varphi}$
then $\Ga \vdash K \leq \kref{t:\smallkind'}{\psi}$
with $\Ga \vdash \smallkind \leq \smallkind'$ and
$\Ga \vdash \varphi \Rightarrow \psi$

\item If $\Ga \models T \equiv \dom{T_0} :: K$ then $\Ga \models T
  \equiv \dom{T_1} :: \ktype$ with $\Ga \models T_0 \equiv T_1 ::
  \kfunction$ and $K = \ktype$
   \end{enumerate}

\end{lemma}
\begin{proof}
  By induction on the given equality derivations, relying on validity, reflexivity, substitution, context conversion and inversion. We show two illustrative cases.

  \begin{description}
  \item[Case:] Transitivity rule
\begin{tabbing}
    $\Ga \models T \equiv S' :: K$ and
    $\Ga \models S' \equiv \treccons{L}{U}{W}:: K$ \` assumption\\
    $\Ga \models S' \equiv \treccons{L'}{U'}{W'}:: K$ with
    $\Ga \models L \equiv L :: \kname$, \\$\Ga \models U \equiv U' ::
    \smallkind$ and $\Ga \models W \equiv W' ::
    \kref{t:\krecord}{L\not\in t}$ and $K = \krecord$ \` by i.h.\\
 
  $\Ga \models T\equiv \treccons{L'}{U'}{W'}  :: K$ \` by transitivity\\
\end{tabbing}
  \end{description}

\item[Case:] $
  \inferrule[]
  {\Ga , t{:}K_0 \vdash T_1 :: K'
   \quad \Ga \vdash T_2 :: K_0}
 {\Ga \models (\lambda t{:}K_0.T_1)\,T_2
   \equiv T_1\{T_2/t\} :: K'\{T_2/t\} }$

  \begin{tabbing}
    $\Ga , t{:}K_0 \vdash T_1 :: K'$, $\Ga \vdash T_2 :: K_0$
    and $\treccons{L}{U}{W} =  T_1\{T_2/t\}$ and
    $K = K'\{T_2/t\}$ \` assumption\\

 $T_1 = T_0$ such that 
    $T_0\{T_2 /t\}= \treccons{L}{U}{W} $\\
    $\Ga , t:K_0 \vdash T_0 :: K'$ \` assumption\\
    $\Ga \vdash T_0\{T_2 /t\} :: K'\{T_2/t\}$ \` by substitution\\
    $\Ga \vdash \treccons{L}{U}{W} :: K'\{T_2/t\}$ \` by definition\\
    $\Ga \vdash L :: \kname$, $\Ga \vdash U :: \ktype$, $\Ga \vdash W
    :: \kref{t{:}\krecord}{L \not\in \lb{t}}$ \` by inversion\\
$\Ga \vdash \treccons{L}{U}{W} \equiv \treccons{L}{U}{W} :: K'\{T_2/t\}$ \` by reflexivity\\
  \end{tabbing}
  
\end{proof}

 \begin{lemma}[Subkinding Inversion]
~
   \begin{enumerate}

   \item If $\Ga \vdash \smallkind \leq \smallkind'$ then
     $\Ga \vdash \smallkind \equiv \smallkind'$ or
     $\Ga \vdash \smallkind' \equiv \ktype$.

   \item If $\Ga \vdash K \leq \kref{t{:}\smallkind'}{\varphi}$ then
     $\Ga \vdash K \equiv \kref{t{:}\smallkind}{\psi}$ with
     $\Ga \vdash \smallkind \leq \smallkind'$ and $\Ga \models \psi \Rightarrow \varphi$.

   \item If $\Ga \vdash \kref{t{:}\smallkind'}{\varphi} \leq \smallkind$
     then $\Ga \vdash \smallkind \leq \smallkind$ and $\Ga ,t{:}\smallkind'\vdash \varphi$.
   \end{enumerate}
 \end{lemma}

 \begin{proof}
By induction on the given derivation, using equality inversion.
\end{proof}

\begin{lemma}
\label{lem:upsidedown}
  If $\Ga \models T \equiv S :: K$,
  $\Ga \vdash T :: K'$ and $\Ga \vdash S :: K'$ and
  $\Ga \vdash K' \leq K$ then $\Ga \vdash T \equiv S :: K'$.
\end{lemma}
\begin{proof} By induction on the given equality derivation.\end{proof}

\unicity*
\begin{proof}
By induction on the structure of the given type/term.

\begin{description}

\item[Case:] $M$ is $\reccons{\ell}{M'}{N'}$

  \begin{tabbing}
    $\Ga \vdash \reccons{\ell}{M'}{N'} : T$ and $\Ga \vdash \reccons{\ell}{M'}{N'} : S$ \` assumption\\
    $\Ga \vdash M' : T_1$, $\Ga \vdash N' : T_2$,
    $\Ga \vdash \ell \equiv L' :: \kname$, $\Ga \vdash
    \reccons{L'}{T_1}{T_2} :: \krecord$\\ and
    $\Ga \models T \equiv \reccons{L'}{T_1}{T_2} :: \krecord$ \` inversion\\
    $\Ga \vdash M' : S_1$, $\Ga \vdash N' : S_2$,
    $\Ga \vdash \ell \equiv L'' :: \kname$, $\Ga \vdash
    \reccons{L''}{S_1}{S_2} :: \krecord$\\ and
    $\Ga \models S \equiv \reccons{L''}{S_1}{S_2} :: \krecord$ \` inversion\\
    $\Ga \models T_1 \equiv S_1 :: K_1$ and $\Ga \vdash K_1 \leq \ktype$ \` by i.h.\\
    $\Ga \models T_1 \equiv S_1 :: \ktype$ \` by conversion\\ 
    $\Ga \models T_2 \equiv S_2 :: K_2$ and $\Ga \vdash K_2 \leq \ktype$ \` by i.h.\\
    $\Ga \vdash T_1 :: \krecord$ and $\Ga \vdash T_2 :: \krecord$ \` by inversion and conversion\\
    $\Ga \models T_2 \equiv S_2 :: \krecord$ \` by Lemma~\ref{lem:upsidedown}
    
  \end{tabbing}

  \item[Case:] $T$ is $\reccons{L}{S_1}{S_2}$

    \begin{tabbing}
      $\Ga \vdash \reccons{L}{S_1}{S_2} :: K$ and
      $\Ga \vdash \reccons{L}{S_1}{S_2} :: K'$ \` assumption\\
      $\Ga \vdash L :: \kname$, $\Ga \vdash S_1 :: \ktype$,
      $\Ga \vdash S_2 :: \kref{t{:}\krecord}{K\not\in t}$ and $\Ga \vdash K \equiv \krecord$ \` by inversion\\
      $\Ga \vdash L :: \kname$, $\Ga \vdash S_1 :: \ktype$,
      $\Ga \vdash S_2 :: \kref{t{:}\krecord}{K\not\in t}$ and $\Ga \vdash K' \equiv \krecord$ \` by inversion\\
      $\Ga \vdash \krecord \leq \krecord$ \` by reflexivity
    \end{tabbing}

  \item[Case:] $M$ is $\ite{\varphi}{M'}{N'}$

    \begin{tabbing}
      $\Ga \vdash \ite{\varphi}{M'}{N'} : T$ and
      $\Ga \vdash \ite{\varphi}{M'}{N'} : S$ \` assumption\\
      $\Ga , \varphi \vdash M' : T_1$, $\Ga , \neg\varphi \vdash N' : T_2$
      and $\Ga \models T \equiv \ite{\varphi}{T_1}{T_2}$ \` by inversion\\
      $\Ga , \varphi \vdash M' : S_1$, $\Ga , \neg\varphi \vdash N' : S_2$
      and $\Ga \models S \equiv \ite{\varphi}{S_1}{S_2}$ \` by inversion\\
      $\Ga , \varphi \models T_1 \equiv S_1 :: K_1$ with $\Ga \vdash K_1 \leq \ktype$ \` by i.h.\\
      $\Ga , \neg\varphi \models T_2 \equiv S_2 :: K_2$ with
      $\Ga \vdash K_2 \leq \ktype$ \` by i.h.\\
      $\Ga \models \ite{\varphi}{T_1}{T_2} \equiv
      \ite{\varphi}{S_1}{S_2} :: \ktype$ \` by rule

    \end{tabbing}

  \item[Case:] $M$ is $\tkindofP{T'}{\smallkind}{M'}{N'}$

    \begin{tabbing}
      $\Ga \vdash \tkindofP{T'}{\smallkind}{M'}{N'} :T$ and
      $\Ga\vdash \tkindofP{T'}{\smallkind}{M'}{N'} :S$ \` assumption\\
      $\Ga \vdash T' :: \smallkind'$, $\Ga \vdash \smallkind$,
      $\Ga , t:\smallkind \vdash M' : T$ and
      $\Ga \vdash N' : T$ \` by inversion\\
      $\Ga \vdash T' :: \smallkind'$, $\Ga \vdash \smallkind$,
      $\Ga , t:\smallkind \vdash M' : S$ and
      $\Ga \vdash N' : S$ \` by inversion\\
      $\Ga \models T \equiv S :: K$ with $\Ga \vdash K \leq \ktype$ \` by i.h.
    \end{tabbing}

  \item[Case:] $T$ is $\tkindofP{T'}{\smallkind}{S_1}{S_2}$

    \begin{tabbing}
      $\Ga \vdash \tkindofP{T'}{\smallkind}{S_1}{S_2} :: K$ and
      $\Ga \vdash \tkindofP{T'}{\smallkind}{S_1}{S_2} :: K'$ \` assumption\\
      $\Ga \vdash T' :: \smallkind'$, $\Ga \vdash \smallkind$,
      $\Ga , t:\smallkind \vdash S_1 :: K$ and
      $\Ga \vdash S_2 :: K$ \` by inversion\\
      $\Ga \vdash T' :: \smallkind'$, $\Ga \vdash \smallkind$,
      $\Ga , t:\smallkind \vdash S_1 :: K'$ and
      $\Ga \vdash S_2 :: K'$ \` by inversion\\
      $\Ga \vdash K \leq K'$ or $\Ga \vdash K' \leq K$ \` by i.h.
    \end{tabbing}

      \item[Case:] $M$ is $\vrec{F}{T}{M'}$

    \begin{tabbing}
      $\Ga \vdash \vrec{F}{T}{M'} : T$ and
      $\Ga \vdash \vrec{F}{T}{M'} : S$ \` assumption\\
      $\Ga \models T \equiv T' :: \smallkind$ and
      $\Ga , F:T \vdash M' : T'$ \` by inversion\\
      $\Ga \models S \equiv S' :: \smallkind'$ and
      $\Ga , F:T \vdash M' : S'$ \` by inversion\\
      $\Ga , F:T \models T' \equiv S' :: K$ with $\Ga \vdash K \leq \ktype$ \` by i.h.\\
      $\Ga , F:T \models T \equiv T' :: \smallkind$ and
      $\Ga , F:T \models S \equiv S' :: \smallkind'$ \` by weakening\\
      $\Ga , F:T \models T \equiv S :: \ktype$ \` by transitivity and conversion\\
      $\Ga \models T \equiv S :: \ktype$ \` by strengthening
    \end{tabbing}

  \item[Case:] $T$ is $\tfix{F}{K}{K'}{t}T'$

    \begin{tabbing}
      $\Ga \vdash \tfix{F}{K_1}{K_2}{t}T' :: K$ and
      $\Ga \vdash \tfix{F}{K_1}{K_2}{t}T' :: K'$ \` assumption\\
      $\Ga, F:\Pi t:K_1.K_2 , t:K_1  \vdash T' :: K_2$,
      $\m{structural}(T',F,t)$ and $\Ga \vdash K \equiv \Pi t:K_1.K_2$
      \` by inversion\\
      $\Ga, F:\Pi t:K_1.K_2 , t:K_1  \vdash T' :: K_2$,
      $\m{structural}(T',F,t)$ and $\Ga \vdash K' \equiv \Pi t:K_1.K_2$ \` by inversion\\
      $\Ga \vdash K \leq K'$ \` by transitivity \\
    \end{tabbing}

  \item[Case:] $M$ is $\tmhdterm{M'}$

    \begin{tabbing}
      $\Ga \vdash \tmhdterm{M'} : T$ and
      $\Ga \vdash \tmhdterm{M'} : S$ \` assumption\\
      $\Ga \vdash M' : \treccons{L}{U}{W}$, and
      $\Ga \models T \equiv U :: \ktype$\\ \` by inversion\\
      $\Ga \vdash M' : \treccons{L'}{U'}{W'}$,  and
      $\Ga \models S \equiv U' :: \ktype$ \\\` by inversion\\
      $\Ga \models \treccons{L}{U}{W}\equiv \treccons{L'}{U'}{W'} :: K$ with $K \leq \ktype$ \` by i.h.\\

      $\Ga \models U \equiv U' :: \ktype$ \` by inversion\\
      $\Ga \models T \equiv U \equiv U' \equiv S$ \` by transitivity
      and symmetry

    \end{tabbing}
    
  \item[Case:] $T$ is $\hdtype{T'}$

    \begin{tabbing}
      Identical to above.
    \end{tabbing}

  \item[Case:] $T$ is $\inst{T_1}{\,T_2}$

    \begin{tabbing}
      $\Ga \vdash \inst{T_1}{\,T_2} :: K$ and $\Ga \vdash \inst{T_1}{\,T_2} :: K'$ \` assumption\\
      $\Ga \vdash T_1 :: \kpolyfun_\smallkind$,
      $\Ga \vdash T_2 :: \smallkind$ and $\Ga \vdash K \equiv \ktype$ \` by inversion\\
      $\Ga \vdash T_1 :: \kpolyfun_{\smallkind'}$,
      $\Ga \vdash T_2 :: \smallkind'$ and $\Ga \vdash K' \equiv \ktype$ \` by inversion\\
      $\Ga \vdash K \leq K'$ \` since $\Ga \vdash \ktype \leq \ktype$
    \end{tabbing}
\end{description}
\end{proof}

\preserv*
\begin{proof}
  By induction on the operational semantics and inversion on typing.
  We show the most significant cases.
  \begin{description}

  \item[Case:]$
      \inferrule[]{T_0\rightarrow T_0'}
    {\stateconfH{(\Lambda t {::} K. \, M) [T_0]} \red
    \stateconfH{(\Lambda t {::} K. \, M) [T_0']} }$

  \begin{tabbing}
    $\Ga \vdash T \equiv U :: \smallkind$ where
    $\Ga \vdash \Lambda t {::} K. \, M :T_1$, $\Ga \vdash T_0 :: K$,
    $\Ga \vdash U :: \smallkind$,\\
    $\Ga \vdash T_1 \equiv \forall t{::}K.S :: \kpolyfun_K$ and $\Ga , t{::}K \vdash S :: \smallkind$ \` by inversion\\
    $\Ga \vdash S\{T_0/t\} :: \smallkind$ \` by substitution\\
    $\Ga \vdash T_0 \equiv T_0' :: K$ \` by definition\\
    $\Ga \vdash S\{T_0/t\} \equiv  S\{T_0'/t\} :: \smallkind$ \` by functionality\\
  $\Ga \vdash U \equiv S\{T_0/t\} :: \smallkind$ \` by transitivity\\
    $\Ga \vdash U \equiv S\{T_0'/t\} :: \smallkind$ \` by transitivity\\
    $\Ga \vdash (\Lambda t {::} K. \, M) [T_0'] : S\{T_0'/t\}$ \` by typing\\
    $\Ga \vdash (\Lambda t {::} K. \, M) [T_0'] : U$ \` by conversion\\
    
  \end{tabbing}
    
  \item[Case:]
    $
 \inferrule[]
    { \stateconf{H}{M} \red \stateconf{H'}{M'} }
    { \stateconf{H}{ \reccons{\ell }{M}{N} } \red
    \stateconf{H'}{\reccons{\ell }{M'}{N}} }
  $

  \begin{tabbing}
    $\Ga \vdash_S T \equiv \treccons{L}{T'}{T''}$,
    $\Ga \vdash_S \ell \equiv L :: \kname$,
    $\Ga \vdash_S M : T'$ and $\Ga \vdash_S N : T''$ \` by inversion\\
    $\exists S'$ such that $S \subseteq S'$, $\Ga \vdash_{S'} H'$ and
    $\Ga \vdash_{S'} M' : T'$ \` by i.h.\\
    $\Ga \vdash_{S'}\reccons{\ell }{M'}{N} : \treccons{L}{T'}{T''}$ \` by RecCons rule
  \end{tabbing}

\item[Case:]
  $
   \inferrule*[]
    {\stateconf{H}{M} \red \stateconf{H'}{M'} }
    {\stateconf{H}{\tmhdterm{M}} \red \stateconf{H'}{\tmhdterm{M'}}  }
  $

  \begin{tabbing}
    $\Ga \vdash_S T \equiv T'\{S'/t\} :: \smallkind\{S'/t\}$,
    $\Ga \vdash_S M : S'$ and\\
    $\Ga \vdash S' :: \kref{t{:}\krecord}{\hdtype{t} \equiv T' :: \smallkind}$ \` by inversion\\
    $\exists S_0$ such that $S \subseteq S_0$,
    $\Ga \vdash_{S_0} H'$ and $\Ga \vdash_{S_0} M' : S'$ \` by i.h.\\
    $\Ga \vdash_{S_0} \tmhdterm{M'} : T'\{S'/t\}$ \` by typing rule
  \end{tabbing}

\item[Case:]
  $
\stateconf{H}{\tmhdterm{\reccons{\ell}{v}{v'}}} \red \stateconf{H}{v} 
  $

  \begin{tabbing}
    $\Ga \vdash_S \tmhdterm{\reccons{\ell}{v}{v'}} : T'$ and
    $\Ga \vdash_s v : T'$ \` by inversion\\
    
  \end{tabbing}

\item[Case:]
  $\inferrule[]{ \cproprestrict{ \Gamma } \vDash \varphi }{ \stateconfH{\ite{ \varphi }{M}{N} }  \red \stateconfH{M}} $

  \begin{tabbing}
   $\Ga \models T \equiv \ite{\varphi}{T_1}{T_2} :: K$ with
   $\Ga , \varphi \vdash_S M : T_1$ and $\Ga , \neg\varphi \vdash_S N : T_2$ \` by inversion\\
   $\Ga \models \ite{\varphi}{T_1}{T_2} \equiv T_1 :: K$ \` by eq. rule\\
   $\Ga \models T \equiv T_1 :: K$  \` by transitivity\\
   $\Ga \vdash_S M : T_1$ \` by cut\\
   
  \end{tabbing}

\item[Case:]
  $ \inferrule[]{ \, }{ \stateconfH{\vrec{F}{T}{M}}
    \red \stateconf{H}{M\{\vrec{F}{T}{M}/{F}\}}}$

  \begin{tabbing}
    $\Ga , F : T \vdash M : T$ and $\m{structural}(F,M)$ \` by inversion\\
    $\Ga \vdash M\{\vrec{F}{T}{M}/{F}\} : T$ \` by substitution
    
  \end{tabbing}

\item[Case:]
  $ \inferrule[]{\Ga \vdash T :: K }{
      \stateconf{H}{ \tkindofP{T'}{K}{M}{N} }
        \red
        \stateconf{H}{M\{T'/t\}}
      }
      $

      \begin{tabbing}
        $\Ga \vdash T' :: K'$, $\Ga \vdash K$, $\Ga , t{:}K \vdash M : T''$ and
        $\Ga \vdash N : T''$ \` by inversion\\
        $\Ga \vdash T :: K$ \` assumption\\
        $\Ga \vdash M\{T'/t\} : T''$ \` by substitution\\
       
      \end{tabbing}

  \end{description}
\end{proof}

\typprog*
\begin{proof}
  Straightforward induction on kinding, relying on the decidability of
  logical entailment.
\end{proof}

\progress*

\begin{proof}
  By induction on typing. Progress relies on type progress and
  on the decidability of logical entailment due to the
term-level and type-level predicate test construct.
\end{proof}


\end{document}